\documentclass[11pt, a4paper]{article}
\usepackage[left=2cm, right=2cm, top=2cm, bottom=2cm]{geometry}
\usepackage[utf8]{inputenc}
\usepackage{algorithm} %
\usepackage{algpseudocode} %
\usepackage{amsmath}
\usepackage{amssymb}
\usepackage{multirow}
\usepackage{booktabs}
\usepackage{natbib}
\usepackage{hyperref}
\usepackage{amsthm}
\usepackage{verbatim}
\usepackage{graphicx}
\graphicspath{{fig/}}
\usepackage{subcaption}
\usepackage{enumitem}
\usepackage{textcomp}
\usepackage{adjustbox}
\usepackage{mathtools}
\usepackage{cancel}
\usepackage{placeins}
\mathtoolsset{showonlyrefs}
\newcommand{\E}{\mathbb{E}}

\newcommand{\R}{\mathbb{R}}
\newcommand{\Tr}{\text{Tr}}

\newcommand{\X}{\mathcal{X}}

\newcommand{\hatthetan}{\hat{\theta}^{(n)}}
\newcommand{\thetastar}{\theta^\star}

\newcommand{\BF}{\operatorname{BF}}
\newcommand{\MMD}{\operatorname{MMD }}
\newcommand{\dsim}{x}
\newcommand{\dobs}{y}

\newcommand{\ddobsn}{\mathbf{y}_\mathbf{n}}
\newcommand{\ddsimn}{\mathbf{x}_\mathbf{n}}
\newcommand{\ddsimm}{\mathbf{x}_\mathbf{m}}
\newcommand{\ddsimmtheta}{\mathbf{x}_\mathbf{m}^{(\theta)}}
\newcommand{\Dsim}{X}
\newcommand{\Dobs}{Y}
\newcommand{\Ddobs}{\mathbf{Y}}
\newcommand{\Ddobsn}{\mathbf{Y}_\mathbf{n}}
\newcommand{\Ddsimm}{\mathbf{X}_\mathbf{m}}
\newcommand{\Ddsimmtheta}{\mathbf{X}_\mathbf{m}^{(\theta)}}
\newcommand{\empP}{\hat P}
\newcommand{\sumi}{\sum_{i=1}^n}
\newcommand{\sumjk}{\sum_{\substack{j,k=1\\k\neq j}}^m}
\newcommand{\PIF}{\operatorname{PIF}\left(z, \theta, \empP_{n}\right)}
\newcommand{\SE}{S_{\operatorname{E}}}

\newcommand{\DE}{D_{\operatorname{E}}}
\newcommand{\DEbeta}{D_{\operatorname{E}}^{(\beta)}}
\newcommand{\pis}{\pi_{S}}
\newcommand{\pishat}{\pi_{\hat S}^{(m)}}
\newcommand{\pshat}{p_{\hat S}^{(m)}}
\newcommand{\pik}{\pi_{S_k}}

\newcommand{\piebeta}{\pi_{\SE^{(\beta)}}}
\newcommand{\piMMD}{\pi_{\operatorname{MMD}}}
\newcommand{\piSL}{\pi_{\operatorname{SL}}}

\newcommand\independent{\protect\mathpalette{\protect\independenT}{\perp}}
\def\independenT#1#2{\mathrel{\rlap{$#1#2$}\mkern2mu{#1#2}}}

\DeclareMathOperator*{\argmin}{arg\,min}

\newtheorem{theorem}{Theorem}

\newtheorem{lemma}{Lemma}
\newtheorem{remark}{Remark}
\theoremstyle{definition}
\newtheorem{definition}{Definition}

\newtheorem{innercustomthm}{Theorem}
\newenvironment{customthm}[1]
{\renewcommand\theinnercustomthm{#1}\innercustomthm}
{\endinnercustomthm}

\title{Generalized Bayesian Likelihood-Free Inference}%

\author{
	Lorenzo Pacchiardi$^1$, %
    Sherman Khoo$^2 $, 
    Ritabrata Dutta$^2 $\thanks{Corresponding author: ritabrata.dutta@warwick.ac.uk}\\
	{\em \small $^1$Department of Statistics, University of Oxford, UK}\\
	{\em \small $^2$Department of Statistics, University of Warwick, UK}\\
}
\date{28th April 2023}

\begin{document}

	\maketitle

\begin{abstract}
\looseness=-1
We propose a posterior for Bayesian Likelihood-Free Inference (LFI) based on generalized Bayesian inference. To define the posterior, we use Scoring Rules (SRs), which evaluate probabilistic models given an observation.
In LFI, we can sample from the model but not evaluate the likelihood; hence, we employ SRs which admit unbiased empirical estimates.
We use the Energy and Kernel SRs, for which our posterior enjoys consistency in a well-specified setting and outlier robustness.
We perform inference with pseudo-marginal (PM) Markov Chain Monte Carlo (MCMC) or stochastic-gradient (SG) MCMC. While PM-MCMC works satisfactorily for simple setups, it mixes poorly for concentrated targets. Conversely, SG-MCMC requires differentiating the simulator model but improves performance over PM-MCMC when both work and scales to higher-dimensional setups as it is rejection-free.
Although both techniques target the SR posterior approximately, the error diminishes as the number of model simulations at each MCMC step increases.
In our simulations, we employ automatic differentiation to effortlessly differentiate the simulator model. We compare our posterior with related approaches on standard benchmarks and a chaotic dynamical system from meteorology, for which SG-MCMC allows inferring the parameters of a neural network used to parametrize a part of the update equations of the dynamical system.
\end{abstract}

	\section{Introduction}\label{sec:intro}\label{sec:scoring_rules}

	This work is concerned with performing inference for a model $P_\theta$ whose density $ p(\dobs|\theta) $ for an observation $\dobs $ is unavailable, but from which it is easy to simulate for any parameter value $ \theta $ (such models are known as intractable-likelihood or simulator models).
    Given $ \dobs $ and a prior $ \pi(\theta) $ on the parameters, the standard Bayesian posterior is $ \pi(\theta|\dobs) \propto \pi(\theta) p(\dobs|\theta) $.
    However, obtaining that explicitly or sampling from it with Markov Chain Monte Carlo (MCMC) techniques is impossible without having access to the likelihood.

	Traditional Likelihood-Free Inference (LFI) techniques exploit model simulations to approximate the exact posterior distribution when the likelihood is unavailable, by either estimating an explicit surrogate \citep{price2018bayesian, an2020robust, thomas2020likelihood} or weighting different parameter values according to the mismatch between observed and simulated data \citep{lintusaari2017, bernton2019approximate}.

	In this work, we introduce a new LFI formulation grounded in the generalized Bayesian inference framework \citep{bissiri2016general, jewson2018principles, knoblauch2019generalized}:
    given a generic loss $ \ell(\dobs, \theta) $ between a single observation $ \dobs $ and parameter $ \theta $, the generalized posterior belief on parameter values can be defined as:
	\begin{equation}\label{Eq:bissiri_post}
	\pi(\theta|\dobs) \propto \pi(\theta ) \exp ( - w\cdot \ell(\dobs, \theta));
	\end{equation} this allows to learn about the parameter value minimizing the expected loss over the data generating process\footnote{Indeed setting $ \ell (\dobs, \theta) = -\log p(\dobs|\theta) $ and $ w=1 $ recovers the standard Bayes update, which learns about the parameter value minimizing the KL divergence \citep{bissiri2016general}.} and respects Bayesian additivity (namely, the belief does not depend on the order observations are received). The learning rate $ w $ controls speed of learning.

    Here, we take $ \ell(\dobs, \theta) $ to be a Scoring Rule (SR) $ S(P_\theta,\dobs) $, which assesses the performance of $ P_\theta $ for an observation $ \dobs $, thus obtaining the \textit{scoring rule posterior} $ \pis $. If $ S(P_\theta,\dobs) $ can be estimated with samples from $P_\theta $, we can perform LFI without worrying about the missing likelihood $ p(\dobs|\theta) $. Two scoring rules allowing this while having good theoretical properties are the energy score and the kernel scores \citep{gneiting2007strictly}. 
 The energy score is given by:
	\begin{equation}
	\SE^{(\beta)}(P, \dobs) = 2 \cdot \E \left[\| \Dsim - \dobs\|_2^\beta\right] - \E\left[\|\Dsim- \Dsim'\|_2^\beta\right] ,\quad  \Dsim \independent  \Dsim' \sim P,
	\end{equation}
	where $ \beta \in (0,2) $. When $ k(\cdot, \cdot) $ is a symmetric and positive-definite  kernel, the kernel scoring rule for $ k $ can be defined as \citep{gneiting2007strictly}:
    \begin{equation}
	S_k(P, \dobs) = \E[k(\Dsim,\Dsim')] - 2\cdot\E [k(\Dsim, \dobs)],\quad  \Dsim \independent  \Dsim' \sim P.
	\end{equation}

    When inserting the kernel Score in Eq.~\ref{Eq:bissiri_post}, the MMD-Bayes method \citep{cherief2020mmd} is recovered. In this paper, we extend MMD-Bayes by framing it under a more general framework; moreover, we discuss its properties in more detail than \cite{cherief2020mmd} and employ MCMC schemes to perform inference (instead of variational inference as in \citealp{cherief2020mmd}).

    Exact sampling from the SR posterior remains impossible; still, a Pseudo-Marginal (PM) MCMC \citep{andrieu2009pseudo} where simulations from $ P_{\theta'} $ are generated for each proposed $ \theta' $ can be used to sample from a close approximation (whose error diminishes when the number of simulations at each step increases) for any SR allowing estimation from samples. While PM-MCMC works well for simple cases and is applicable to any simulator model, it mixes poorly for concentrated targets (such as those obtained when many observations are used).
    
    Alternatively, approximate samples from the SR posterior can be obtained using Stochastic-Gradient (SG) MCMC \citep{nemeth2021stochastic} by leveraging the unbiased estimates of $\nabla_\theta S(P_\theta, y)$ possible with the Energy and the kernel score. The unbiased gradient estimate necessitates the gradient of the simulated data with respect to model parameters, which can be easily obtained by implementing the simulator model with automatic-differentiation libraries. In this work, we mostly empoy adaptive stochastic gradient Langevin dynamics \cite{jones2011adaptive}, for which theoretical bounds for its error and results for asymptotic convergence exist \citep{ding2014bayesian, leimkuhler2016adaptive, leimkuhler2020hypocoercivity}; further, we show empirically that the SG-MCMC target well matches that obtained with PM-MCMC in cases where the latter mixes well, while requiring lower computational effort. Importantly, SG-MCMC has no mixing issues (as it is rejection-free). To the best of our knowledge, ours is the first ever application of gradient-based sampling methods to LFI using unbiased estimate of the gradient of the target distribution, which is enabled by the SR posterior and leads to scalable inference for high-dimensional parameter spaces.%

    Equipped with this sampling method, we empirically study concentration and outlier-robustness properties of the SR posterior, for which we also establish theoretical results. Specifically, we show asymptotic normality and a finite-sample bound on the probability of deviation of the posterior expectation of the divergence from the minimum divergence achievable by the model. We also provide a quantitative bound on the robustness of the posterior to outliers in the data.

    Qualitatively, the concentration and outlier-robustness properties of the SR posterior are independent on the value of $w$ in its definition (see Eq.~\ref{Eq:SR_posterior}). However, the choice of $w$ determines the rate of contraction of the SR posterior. A large ongoing research effort is devoted to the selection of $w$ for generalized Bayesian posteriors, resulting in methods ensuring, for instance, different forms of coverage \citep{lyddon2019general, syring2019calibrating, matsubara2022generalised} or other properties \citep{bissiri2016general, holmes2017assigning, loaiza2019focused}.
    Several of those methods (and plausibly future ones) are applicable to our framework.
    Hence, we do not delve deep into determining the optimal way to select $w$ or develop our own, mindful of the facts that this is an area of active research and that each practical use case is best tackled with a different method. Still, in our empirical evaluations of the SR posterior, it may be beneficial for different posteriors to have a similar scale. When that is required, we will either rely on hand-tuning or a previously introduced method which we revisit for our framework.

    We empirically compare the SR posterior with the popular Bayesian Synthetic Likelihood (BSL, \citealt{price2018bayesian}) approach, which is an instance of the SR posterior. However, as BSL does not provide unbiased gradient estimates, this prevents the use of SG-MCMC, which hinders the performance of BSL for concentrated and high-dimensional targets. 
    Next, we consider a real-world meteorological model \citep{lorenz1996predictability} and infer its parameters with Approximate Bayesian Computation \citep{lintusaari2017} and our SR posterior. We also use our framework to infer the parameters of a high-dimensional Neural Stochastic-Differential Equation for modelling the same data, which is unachievable with traditional (non-gradient-based) sampling methods. %

	The rest of this manuscript is organized as follows.
 In Sec.~\ref{sec:Bay_inf_SR}, we first review the scoring rules and define the SR posterior; we then discuss and compare the two sampling methods.
 Next, we study concentration properties in Section~\ref{sec:SR_post_prop} and outlier robustness in Section~\ref{sec:robustness}. Simulation studies comparing with other LFI approaches are presented in Sec.~\ref{sec:experiments}. Finally, we briefly review 
 previous works in Sec.~\ref{sec:related_works} and conclude and suggest future directions in Sec.~\ref{sec:conclusion}.

	\subsection{Notation}
	We will denote respectively by $ \mathcal{X} \subseteq \R^d $ and $ \Theta \subseteq \R^p $ the data and parameter space, which we assume to be Borel sets. We will assume the observations are generated by a distribution $ P_0 $ and use $ P_\theta $ and $ p(\cdot|\theta) $ to denote the distribution and likelihood of our model. Generic distributions will be indicated by $ P $ or $ Q $, while $ S $ will denote a generic scoring rule. Other upper-case letters will denote random variables while lower-case ones will denote observed (fixed) values. We will denote by $ \Dobs $ or $ \dobs $ the observations (correspondingly random variables and realizations) and $ \Dsim $ or $ \dsim $ the simulations.
	Subscripts will denote sample index and superscripts vector components.
	Also, we will respectively denote by $ \Ddobsn = \{ \Dobs_i\}_{i=1}^n \in \X^n $ and $ \ddobsn = \{ \dobs_i\}_{i=1}^n \in \X^n $ a set of random and fixed observations. Similarly, $ \Ddsimm = \{ \Dsim_j\}_{j=1}^m \in \X^m  $ and $ \ddsimm = \{ \dsim_j\}_{j=1}^m \in \X^m   $ denote a set of random and fixed model simulations. Finally, $ \independent $ will denote independence between random variables, while $ X \sim P$ indicates a random variable distributed according to $ P $.

	\section{Bayesian inference using scoring rules}\label{sec:Bay_inf_SR}

    \subsection{Background definitions}
	A Scoring Rule (SR, \citealp{gneiting2007strictly}) $ S $ is a function of a probability distribution over $ \X $ and of an observation in $  \X $. For a distribution $P$ and an observation $y$, we will denote the corresponding score as $S(P,y)$. %
	Assuming that $\dobs$ is a realization of a random variable $ \Dobs $ with distribution $ Q $, the expected scoring rule is defined as:
	\begin{equation}
	S(P,Q) := \E_{\Dobs \sim Q} S(P, \Dobs),
	\end{equation}
	where we overload notation in the second argument of $ S $.
	The scoring rule $ S $ is \textit{proper} relative to a set of distributions $ \mathcal{P}(\X) $ over $ \X $ if $$ S(Q, Q) \le S(P,Q) \ \forall \ P,Q \in \mathcal{P}(\X),$$ i.e., if the expected scoring rule is minimized in $ P $ when $ P=Q $. Moreover, $ S $ is \textit{strictly proper} relative to $ \mathcal{P}(\X) $ if $ P = Q $ is the unique minimum: $$ S(Q,Q) < S(P,Q)  \ \forall \ P, Q \in \mathcal{P}(\X)  \text{ s.t. }  P\neq Q.$$

	The divergence related to a proper scoring rule \citep{dawid2014theory} can be defined as $ D(P,Q) := S(P,Q) - S(Q,Q) \ge 0 $. Notice that $ P=Q \implies D(P,Q) = 0$, but there may be $ P\neq Q $ such that $ D(P,Q)=0 $. However, if $ S $ is strictly proper, $ D(P,Q) = 0 \iff P=Q $, which is the commonly used condition to define a statistical divergence (as for instance the Kullback-Leibler, or KL divergence).
	Therefore, each strictly proper scoring rule corresponds to a statistical divergence between probability distributions.

	The energy score introduced in Sec.~\ref{sec:intro} is a strictly proper scoring rule for the class of probability measures $ P $ such that $ \E_{X\sim P}\|\Dsim\|^\beta < \infty $ \citep{gneiting2007strictly}. The related divergence is the square of the energy distance, which is a metric between probability distributions (\citealt{rizzo2016energy}; see Appendix~\ref{app:energy})\footnote{The probabilistic forecasting literature \citep{gneiting2007strictly} use a different convention for the energy score and the subsequent kernel score, which amounts to multiplying our definitions by $ 1/2 $. We follow here the convention used in the statistical inference literature \citep{rizzo2016energy, cherief2020mmd, nguyen2020approximate}}.
	We will fix $ \beta=1$ in the rest of this work and we will write $ \SE $ in place of $ \SE^{(1)} $. Analogously, the kernel score is proper for the class of probability distributions for which $ \E[k(\Dsim,\Dsim')] $ is finite (by Theorem 4 in \cite{gneiting2007strictly}). Additionally, it is strictly proper under conditions which ensure that the MMD is a metric for probability distributions on $ \X $ (see Appendix~\ref{app:MMD}). These conditions are satisfied, among others, by the Gaussian kernel (which we will use in this work):
	\begin{equation}\label{Eq:gau_k}
	k(x, y)=\exp \left(-\frac{\|x-y\|_{2}^{2}}{2 \gamma^{2}}\right),
	\end{equation}
	in which $ \gamma $ is a scalar bandwidth. The divergence corresponding  to the kernel score is the squared Maximum Mean Discrepancy (MMD, \citealp{gretton2012kernel}) relative to the kernel $ k $ (see Appendix~\ref{app:MMD}). 
    \subsection{The scoring rule posterior}
    Consider now a set of independent and identically distributed observations $ \ddobsn \in \X^n $ sampled from a distribution $ P_0 $. We introduce the SR posterior for $ S $ by setting $ \ell(\dobs, \theta)  = S(P_\theta, \dobs)$ in the general Bayes update in Eq.~\eqref{Eq:bissiri_post}:
	\begin{equation}\label{Eq:SR_posterior}
	\pis(\theta|\ddobsn) \propto \pi(\theta) \exp\left\{ - w \sum_{i=1}^n S(P_\theta, \dobs_i) \right\}.
	\end{equation}

	The standard Bayes posterior is recovered from Eq.~\eqref{Eq:SR_posterior} by setting $ w=1 $ and $ S(P_\theta,\dobs) = -\log p(\dobs|\theta)$. Such choice of $ S $ is called the \textit{log score}, is strictly proper, and corresponds to the Kullback-Leibler (KL) divergence. With the same $ S $, $ w\neq1 $ yields the fractional posterior \citep{holmes2017assigning, bhattacharya2019bayesian}.

	\begin{remark}[\textbf{Bayesian additivity}]\label{remark:bayesian_additivity}
		The posterior in Eq.~\eqref{Eq:SR_posterior} satisfies Bayesian additivity (also called coherence, \citealt{bissiri2016general}): sequentially updating the belief with a set of observations does not depend on the order the observations are received. %
	\end{remark}
	
	\begin{remark}[\textbf{Non-invariance to change of data coordinates}]	\label{remark:non-invariance}
		The SR posterior is in general not invariant to change of the coordinates used for representing the observations. This is a property common to loss-based frequentist estimators and to the generalized posterior obtained from them \citep{matsubara2021robust}; see Appendix~\ref{app:coordinates} for more details.
	\end{remark}

	\subsection{Sampling the scoring rule posterior for LFI}\label{sec:simulator_models}
    Computing the energy and kernel scores, provided the likelihood is available, requires solving a double expectation, which is challenging in practice. 
    In the following, we will show how the availability of samples from simulator models allows to get unbiased estimates of the energy and kernel scores. Further, for differentiable simulator models (for which derivative of the simulated data w.r.t. to the parameters are available) we can also obtain unbiased estimators of the gradient of the scoring rules considered here under some regularity conditions. These derivatives can be effortlessly computed using automatic differentiation libraries for most simulator models \footnote{Exceptions include simulator models with thresholding involved in their simulation process or when the simulated data is discrete.}.  
    
    To sample approximately from the scoring rule posterior, we propose a pseudo-marginal Monte Carlo Markov chain (PM-MCMC) algorithm using estimators of scoring rules computed from samples of the simulator model. In addition, we propose using stochastic gradient Monte Carlo Markov chain (SG-MCMC) algorithms for differentiable simulator models. When applicable, SG-MCMC avoids two known drawbacks of PM-MCMC, namely the curse of dimensionality limiting its application to high-dimensional parameter spaces and the ``sticky" behaviour resulting in poor mixing for concentrated targets. 
    
	\subsubsection{Pseudo-marginal MCMC} \label{sec:PM}
    
    Our PM-MCMC algorithm depends upon the existence of an estimate $ \hat S(\ddsimmtheta, \dobs) $ of $ S(P_\theta, \dobs) $, where $ \ddsimmtheta = \{\dsim_j^{(\theta)}\}_{j=1}^m$ is a set of samples $ \dsim_j^{(\theta)} \sim P_\theta $, and $ \hat S $ is such that $ \hat S(\Ddsimmtheta, \dobs) \to S(P_\theta, \dobs) $ in probability as $ m \to\infty $ (i.e., it estimates the SR consistently). Unbiased estimates for $\SE^{(\beta)}$ and $S_k$ can be obtained by unbiasedly estimating the expectations using samples $\ddsimmtheta$ as following.
 \begin{equation}
	\hat S_{\text{E}}^{(\beta)}(\ddsimm^{(\theta)}, \dobs) =\frac{2}{m} \sum_{j=1}^m \left\| \dsim_j^{(\theta)} - \dobs\right\|_2^\beta - \frac{1}{m(m-1)}\sumjk \left\|\dsim_j^{(\theta)}-\dsim_k^{(\theta)}\right\|_2^\beta.
	\end{equation}
 \begin{equation}
	\hat S_k(\ddsimm^{(\theta)}, \dobs) = \frac{1}{m(m-1)}\sumjk  k(\dsim_j^{(\theta)},\dsim_k^{(\theta)} )-\frac{2}{m} \sum_{j=1}^m k(\dsim_j^{(\theta)},\dobs).
	\end{equation}
    For each proposed value of $ \theta $, we simulate $\ddsimmtheta = \{\dsim_j^{(\theta)}\}_{j=1}^m$ and estimate the target in Eq.~\eqref{Eq:SR_posterior} with:
	\begin{equation}\label{Eq:SR_post_unbiased_est}
	\pi(\theta) \exp\left\{ - w \sum_{i=1}^n \hat S(\ddsimmtheta, \dobs_i) \right\}.
	\end{equation}
	This procedure is an instance of pseudo-marginal MCMC \citep{andrieu2009pseudo}, with target:
	\begin{equation}\label{Eq:SR_post_finite_n_target}
	\pishat(\theta|\ddobsn) \propto \pi(\theta) \pshat (\ddobsn|\theta),
	\end{equation}
	where:
	\begin{equation}
	\pshat (\ddobsn|\theta) = \E\left[\exp\left\{ - w \sum_{i=1}^n \hat S(\Ddsimmtheta, \dobs_i) \right\} \right].
	\end{equation}
	For a single draw $ \ddsimmtheta $, the quantity in Eq.~\eqref{Eq:SR_post_unbiased_est} is in fact a non-negative and unbiased estimate of the target in Eq.~\eqref{Eq:SR_post_finite_n_target}; this approach is similar to what is proposed in \cite{drovandi2015bayesian} for inference with auxiliary likelihoods, which has also been used by \cite{price2018bayesian} for BSL. As it was already the case for the latter,
	the target $ \pishat(\theta|\ddobsn) $ is not the same as $ \pis(\theta|\ddobsn) $ and depends on the number of simulations $ m $; in fact, in general:
	\begin{equation}
	\E\left[\exp\left\{ - w \sum_{i=1}^n \hat S(\Ddsimmtheta, \dobs_i) \right\} \right] \neq \exp\left\{ - w \sum_{i=1}^n S(P_\theta, \dobs_i) \right\},
	\end{equation}
	even if $ \hat S(\ddsimmtheta, \dobs) $ is an unbiased estimate of $ S(P_\theta, \dobs) $.
	However, it is possible to show that, as $ m\to \infty $, $ \pishat $ converges to $ \pis $:

	\begin{theorem}\label{Th:auxiliary_lik}
		If $ \hat S(\Ddsimmtheta, \dobs_i) $ converges in probability to $ S(P_\theta, \dobs_i) $ as $ m \to\infty $  for all $ i=1, \ldots, n $, then, under some minor technical assumptions:
		\begin{equation}
		\lim_{m\to\infty} \pishat (\theta|\ddobsn)= \pis(\theta|\ddobsn), \quad \forall \theta \in \Theta.
		\end{equation}
	\end{theorem}
	The above result is an extension of the one in \cite{drovandi2015bayesian} for Bayesian inference with an auxiliary likelihood. Appendix~\ref{app:auxiliary_lik} gives the technical conditions explicitly (in Theorem~\ref{Th:auxiliary_lik_complete}) and proves the result.

    In practice, in place of the vanilla pseudo-marginal approach discussed above, we use a correlated pseudo-marginal MCMC \citep{dahlin2015accelerating, deligiannidis2018correlated, picchini2022sequentially}, which reuses the random numbers used in model simulations over subsequent proposed parameter values. This correlates the target estimates at subsequent steps and reduces the chances of the chain getting stuck due to atypical random number draws. Specifically, the $ m $ simulations used in the posterior estimate (Eq.~\ref{Eq:SR_post_unbiased_est}) are split in $ G $ groups; at each MCMC step, a new set of random numbers is proposed for the simulations in a randomly chosen group (alongside the proposed value for $ \theta $), and accepted or rejected in the standard way. 
    This algorithm still targets Eq.~\eqref{Eq:SR_post_finite_n_target}.

    \subsubsection{Stochastic Gradient MCMC}
    For the scoring rules used across this work, as well as any weighted sum of those, we can write $S(P_\theta, y) = \displaystyle \mathop{\mathbb{E}}_{X, X^{\prime} \sim P_\theta} g\left(X, X^{\prime}, y\right)$  for some function $g$; namely, the SR is defined through an expectation over (possibly multiple) samples from $P_{\theta}$.
    In the following, we assume random samples from the simulator model $P_\theta$ can be written as $X = h_{\theta}(Z)$ where $Z$ follows a base distribution $Q$ independent of the parameters $\theta$. Now:
    
    $$
    \nabla_\theta S(P_\theta, y) = \nabla_\theta \displaystyle \mathop{\mathbb{E}}_{X, X^{\prime} \sim P_\theta} g\left(X, X^{\prime}, y\right) \\
    = \nabla_\theta \displaystyle \mathop{\mathbb{E}}_{Z, Z^{\prime} \sim Q} g\left(h_{\theta}(Z), h_{\theta}(Z^{\prime}), y\right) \\
     =  \displaystyle \mathop{\mathbb{E}}_{Z, Z^{\prime} \sim Q} \nabla_\theta g\left(h_{\theta}(Z), h_{\theta}(Z^{\prime}), y\right) \\
    $$
    In the latter equality, the exchange between expectation and gradient is not a trivial step.
    Luckily, Theorem 5
    in \cite{binkowski2018demystifying} proved the above step to be valid almost surely with respect to a measure on $\theta$, under mild conditions on the functions $g$ and $h_{\theta}$ (such conditions are satisfied if both functions are differentiable). Based on this, we estimate the gradient of the scoring rule as follow:
    \begin{eqnarray}
    \label{eq:uest_SR}
    \widehat{\nabla_\theta} S(P_\theta, y) =\frac{1}{m(m-1)} \sum_{\substack{i, j=1 \\ i \neq j}}^m 
    \nabla_\theta g\left(h_{\theta}(Z_i), h_{\theta}(Z_j^{\prime}), y\right)
    , \quad Z_i \perp \!\!\! \perp Z^{\prime}_j \sim Q .
    \end{eqnarray}
    In practice, this can be easily obtained by implementing the function $h_\theta$ using automatic-differentiation libraries \citep{pytorch}.

    By relying on this construction, 
    we adapt two existing SG-MCMC \citep{nemeth2021stochastic} algorithms (stochastic gradient Noose-Hoover thermostat \citep{ding2014bayesian} and Preconditioned Stochastic Gradient Langevin \citep{li2016preconditioned}) to sample from the scoring rule posterior. As mentioned above, these algorithms are approximate, but the computational advantage they provide overweights the induced approximation. 
    
    Alternatively, Piecewise-Deterministic Markov Processes (PDMP, \citealp{fearnhead2018piecewise}) allow exact sampling with an unbiased estimate of the log-target gradient; unfortunately, however, the exact implementation of the existing algorithms requires computing an upper bound of the log-target gradient which is intractable for most practical use cases. To avoid this, approximate methods \citep{pagani2020nuzz, corbella2022automatic} are developed, which are however inconvenient for general target distributions compared to SG-MCMC methods.

        \paragraph{Adaptive Stochastic Gradient Langevin Dynamics (adSGLD)}The earliest known stochastic gradient MCMC algorithm  \citep{welling2011bayesian} is based upon the (Overdamped) Langevin Diffusion, defined by the following Stochastic Differential Equation: \begin{equation}\label{eq:OLangDiff}
        \mathrm{d} \theta(t)=-\frac{1}{2} \nabla_{\theta} U(\theta(t)) \mathrm{d} t+\mathrm{d} B_t.
        \end{equation}
        For the SR posterior, $U(\theta) = \log{\pi(\theta)} - w \sum_{i=1}^n S(P_\theta, \dobs_i)$, $\theta \in \mathbb{R}^d$ and $B_t \in \mathbb{R}^d$ is standard Brownian Motion. Under suitable regularity conditions, this continuous-time diffusion has $\pis(\theta|\ddobsn)$ as its stationary distribution \citep{roberts1996exponential, pillai2012optimal}. In practice, we are unable to simulate from this stochastic process exactly. Hence, numerical integration schemes are used to generate samples. For instance, the Euler-Maruyama method consists of the following update:
        \begin{eqnarray}
        \label{eq:EM-OLangDiff-SR}
           \theta_{t+1} \leftarrow \theta_t-\frac{\epsilon}{2} \nabla_{\theta} U(\theta(t)) + \sqrt{\epsilon}Z 
        \end{eqnarray}
        repeated over $t$, where $Z$ is a $d$-dimensional standard normal random vector and $\epsilon$ is a  discretisation step size. Following \cite{welling2011bayesian}, we propose to use the unbiased estimate of the gradient of $\nabla_{\theta} U(\theta(t))$,
        $$\widehat{\nabla_\theta} U(\theta) = \nabla_\theta \log{\pi(\theta)} - w \sum_{i=1}^n \widehat{\nabla_\theta} S(P_\theta, y_i)$$
        in the above update equation; this method is called Stochastic Gradient Langevin Dynamics (SGLD). %
        Using a sequence $\{\epsilon_i\}_{i=1}^{N}$ converging to $0$ and taking $m \to \infty$, under some condition, \cite{welling2011bayesian} shows that SGLD samples from the scoring rule posterior. 
        
        In practice, however, we do not have $\epsilon_i\to0$ neither $m\to\infty$. Hence, to ensure sampling with minimal bias for our noisy gradient scenario, we utilize the adaptive Langevin dynamics originally proposed in \cite{jones2011adaptive} and later used for Bayesian inference in \cite{ding2014bayesian}. %
        We would refer this algorithm as adaptive stochastic gradient Langevin dynamics (adSGLD), which
        runs on an augmented space $(\theta,p, \xi)$, where $\theta$ represents the parameter of interest, $p \in \mathbb{R}^d$ represents the momentum and $\xi$ represents an adaptive thermostat controlling the mean kinetic energy $\frac{1}{n} \mathbb{E}[p^{\top} p]$, along with a diffusion factor $\mathcal{A}$. Thus, the new dynamics is as follows:
        \begin{equation}
        \label{eq:adSGLD}
        \left\{ \begin{aligned}
        \mathrm{d} \theta_t &=p_t \mathrm{~d} t \\
        \mathrm{~d} p_t &=-\nabla_{\theta} U(\theta(t)) \mathrm{d} t- \xi p_t \mathrm{~d} t+ \sqrt{2\mathcal{A}} \mathcal{N}(0, I dt) \\
        d \xi &= \left(\frac{1}{n} p_{t}^{\top} p_{t}-1\right) d t
        \end{aligned} \right.
        \end{equation}
        Theoretical properies and convergence of adSGLD algorithm has been studied in \cite{ding2014bayesian}, \cite{leimkuhler2016adaptive} and \cite{leimkuhler2020hypocoercivity}. Below, we state the adSGLD algorithm, which requires fixing the hyperparameters $\epsilon$ (step size) and $\mathcal{A}$. 
        
        \begin{algorithm}
        \caption{adSGLD Algorithm for scoring rule posterior}\label{alg:cap}
        \hspace*{\algorithmicindent} \textbf{Input: $\mathcal{A}$, $\epsilon$, $\theta_0$, $N$ }\\
        \hspace*{\algorithmicindent} \textbf{Output: $\{\theta_i\}^{N}_{i=1}$ samples} 
        \begin{algorithmic}[1]
        \State Initialise $P_0 \sim N(0, I)$ and $\xi_0 \gets \mathcal{A}$ 
        \For{$i=1$ to $N$}:
        \State Estimate $\widehat{\nabla}_{\theta} U({\theta_{i-1}})$
        \State $P_i \gets P_{i-1} -\xi_{i-1} P_{i-1} \epsilon - \widehat{\nabla}_{\theta} U({\theta_{i-1}}) \epsilon+\sqrt{2 \mathcal{A}} N(0, \epsilon)$ 
        \State $\theta_{i} \gets \theta_{i-1}+ P_{i} \epsilon $
        \State $\xi_{i} \gets \xi_{i-1}+\left(\frac{1}{n} P_{i}^{\top} P_{i}-1\right) \epsilon$
        \EndFor
        \end{algorithmic}
        \end{algorithm}

    \paragraph{Preconditioned Stochastic Gradient Langevin Dynamics (pSGLD, \citealp{li2016preconditioned})}\label{sec:psgld} 
    This algorithm preconditions the log-target with a diagonal matrix $G(\theta)$ obtained through a running average of the squared gradients using the following update equations:
        \begin{equation}
        \begin{aligned}
        & G\left({\theta}_{t+1}\right)=\operatorname{diag}\left(\mathbf{1} \oslash\left(\lambda \mathbf{1} +\sqrt{V\left({\theta}_{t+1}\right)}\right)\right) \\
        & V\left({\theta}_{t+1}\right)=\alpha V\left({\theta}_t\right)+(1-\alpha) \widehat{\nabla}_{\theta} U({\theta_{t}}) \odot \widehat{\nabla}_{\theta} U({\theta_{t}})
        \end{aligned}
        \end{equation}
        with $\oslash$ and $\odot$ denoting element-wise matrix division and product respectively. The  hyperparameter $\lambda$ is a small bias term to avoid the degeneration of the preconditioner, while $\alpha\in(0,1)$ is a relative weighting between the previous and current gradients. 
        This algorithm performs well for non-convex posteriors on high-dimensional space, and in particular for the complicated posteriors characterized by deep neural networks. We state the algorithm for pSGLD below.%

        \begin{algorithm}
        \caption{pSGLD Algorithm for scoring rule posterior}\label{alg:cap}
        \hspace*{\algorithmicindent} \textbf{Input: $\lambda$, $\alpha$, $\epsilon$, $\theta_0$, $N$ }\\
        \hspace*{\algorithmicindent} \textbf{Output: $\{\theta_i\}^{N}_{i=1}$ samples} 
        \begin{algorithmic}[1]
        \State Initialise $V_0 \gets \boldsymbol{0}$ 
        \For{$i=1$ to $N$}:
        \State Estimate $\widehat{\nabla}_{\theta} U({\theta_{i}})$
        \State $V\left(\theta_{i}\right) \gets \alpha V\left(\theta_{i-1}\right)+(1-\alpha) \widehat{\nabla}_{\theta} U({\theta_{i}}) \odot \widehat{\nabla}_{\theta} U({\theta_{i}})$ 
        \State $G\left(\theta_{i}\right) \gets \operatorname{diag}\left(\mathbf{1}\left(\lambda \mathbf{1}+\sqrt{V\left(\theta_{i}\right)}\right)\right) $
        \State $\theta_{i+1} \gets 
        \theta_{i}+
        \frac{\epsilon}{2}
            G\left(\theta_{i}\right) U({\theta_{i}})
            + \mathcal{N}
            \left(0, \epsilon G
                \left(\theta_{i}
                \right)
                \right)$
        \EndFor
        \end{algorithmic}
        \end{algorithm}

        In practice, we set $\lambda$ to $10^{-5}$ and $\alpha$ to $0.99$.

        \paragraph{Choice of step size $\epsilon$} For SG-MCMC algorithms, choosing the step-size $\epsilon$ is critical, as it represents a trade-off between the speed of convergence or mixing performance and the discretisation error. In practice, SG-MCMC algorithms are often used with a constant step size due to slow mixing when $\epsilon\approx0$. 
        To tune $\epsilon$, we use a modified version of the multi-armed bandit algorithm based on the kernelized Stein discrepancy proposed in \cite{coullon2021efficient}. This algorithm identifies each arm with a specific hyperparameter configuration, and for a fixed time budget, sequentially eliminates poor hyperparameter configurations based on the kernelized Stein discrepancy between the samples and the target distribution. %

    \begin{figure}[htbp!]
    \centering
    \includegraphics[width=0.7\linewidth]{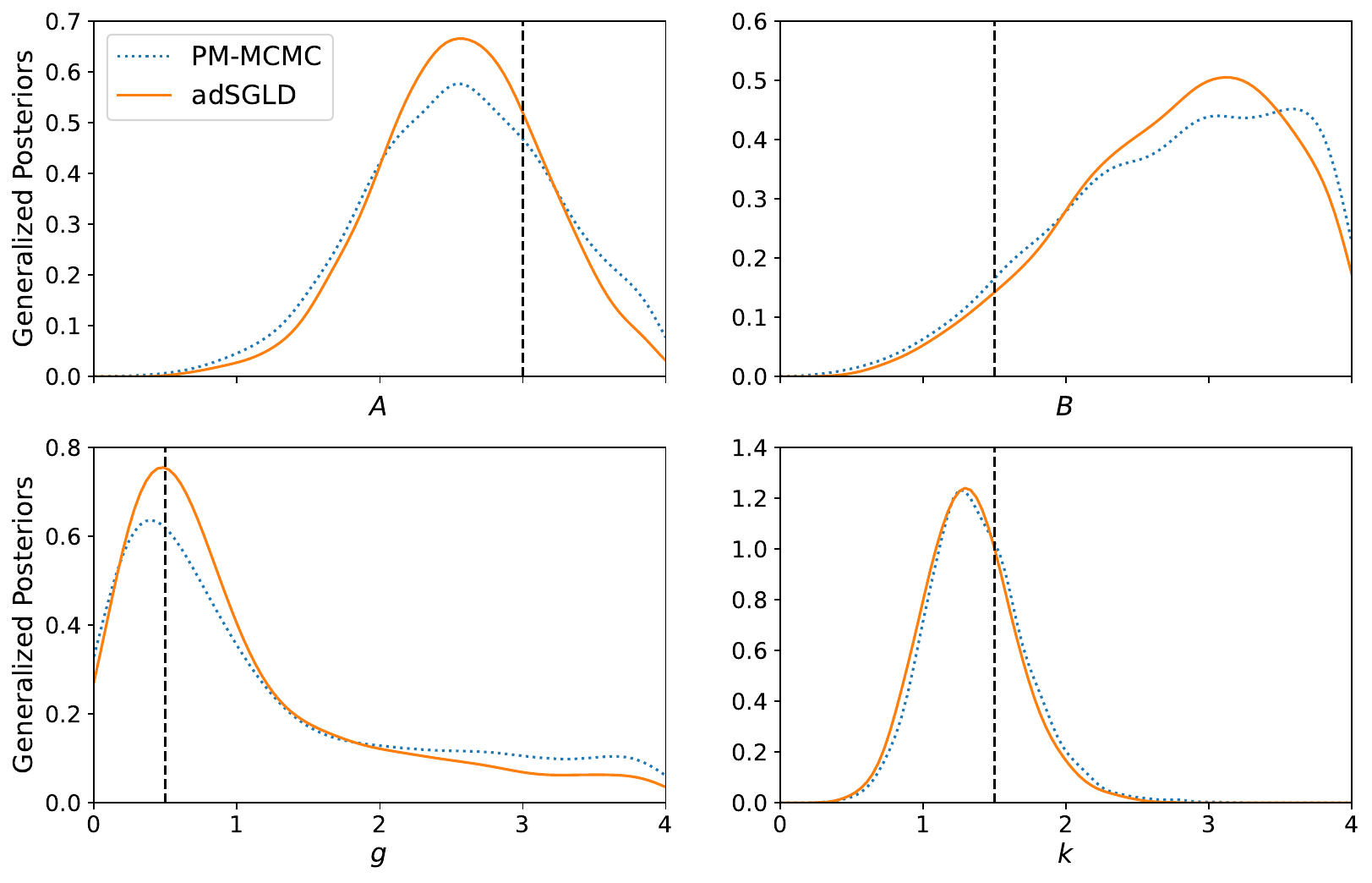}
    \caption{\textbf{Comparison of adSGLD and PM-MCMC} to sample from the marginals of the energy score posterior for the g-and-k model obtained with $n=10$. Vertical lines denote true parameter values. For both, 100000 samples with 10000 burn-in were used.}
    \label{fig:Post pm}
	\end{figure}

    \begin{figure}[htbp!]
    \centering
    \includegraphics[width=0.6\linewidth]{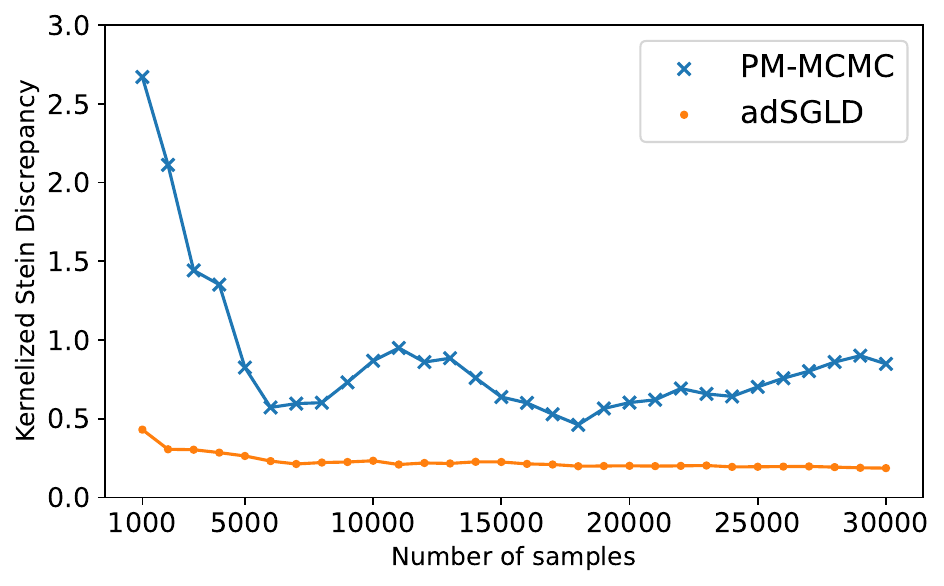}
    \caption{\textbf{Kernelized Stein Discrepancy} for first 30000 MCMC samples for the energy score posterior for the g-and-k model, on $n=10$ observations, sampled with adSGLD and PM-MCMC. KSD uses the inverse multi-quadratic kernel, with the gradients estimated using 500 simulated observations from each parameter. adSGLD both converges faster and is more accurate than the PM-MCMC algorithm.}
    \label{fig:KSD pm}
	\end{figure}

    \subsection{Comparison between PM-MCMC and SG-MCMC}\label{sec:comparison}
    To compare PM-MCMC and SG-MCMC (specifically, the adSGLD algorithm), we perform an empirical study on the univariate g-and-k model \citep{prangle2017gk}. The latter is defined in terms of the inverse of its cumulative distribution function $ F^{-1} $. Given a quantile $ q $, we define:
	\begin{equation}\label{Eq:gk}
	F^{-1}(q) = A + B \left[ q+ 0.8  \frac{1-e^ {-g z(q)}}{1+e^{-g z(q)}}\right]\left(1+z(q)^{2}\right)^{k} z(q),
	\end{equation}
	where the parameters $ A $, $ B $, $ g $, $ k $ are broadly associated to the location, scale, skewness and kurtosis of the distribution, and $ z(q) $ denotes the $ q $-th quantile of the standard normal distribution $ \mathcal{N}(0,1) $.
	Likelihood evaluation for this model is costly as it requires numerical inversion of $ F^{-1} $; instead, sampling is immediate by drawing $z \sim \mathcal{N}(0,1) $ and inputing it in place of $ z(q) $ in the expression above. We use uniform priors on $ [0,4]^4 $ on the sets of parameters $ \theta = (A,B,g,k) $. For $10$ observations from true parameter values $ A^\star =3,$ $ B^\star=1.5,$ $ g^\star=0.5,$ $ k^\star=1.5$, we perform inference with the energy score Posterior with $w=1$, setting the number of simulations per parameter value to $ m= 500 $ and run adSGLD and PM-MCMC for 110000 steps.

    Figure~\ref{fig:Post pm} shows a kernel density estimate of the samples obtained with the two methods: the two densities are similar, with the PM-MCMC one slightly broader. 
    As both sampling methods are asymptotically biased, we cannot rely on traditional MCMC diagnostics (such as the R-hat and the autocorrelation function) to quantitatively evaluate sample quality, as those only evaluate properties of the chain itself and are thus unable to measure the discrepancy between samples from an approximate sampler and exact target. To this aim, we employ the kernelized Stein discrepancy (KSD) proposed in \cite{gorham2017measuring}
    which, conveniently, can be estimated by using MCMC samples and unbiased estimates of the gradient of the log target (see Appendix~\ref{app:KSD}).
    We compute the KSD with an increasing number of samples obtained from the two methods, thus allowing to investigate which algorithm converges faster. The results can be seen in Fig.~\ref{fig:KSD pm}: the adSGLD algorithm converges faster than the PM-MCMC algorithm and produces samples that are a better approximation to the target distribution. Based on the superior performance of the adSGLD algorithm here, we will employ it for sampling from the SR posterior in the remaining simulation studies as all our considered simulator models are differentiable, unless otherwise specified. For comparison, results with PM-MCMC for some of the setups considered in the main body of the paper are reported in Appendix~\ref{app:PM_results}.

	\section{Concentration properties of the scoring rule posterior}\label{sec:SR_post_prop}

	\subsection{Asymptotic normality}\label{sec:as_norm}

    Under mild conditions, the SR posterior satisfies a Bernstein-von Mises theorem ensuring asymptotic normality. This generalizes the analogous result valid for the standard Bayesian posterior. For brevity, we give here a simplified statement, with the full one given (and proven) in Appendix~\ref{app:normal_asymptotic_proof}. Without loss of generality, we fix here $ w=1 $ (different values can be absorbed in the definition of $ S $).
    \begin{theorem}\label{Th:asymptotic_short}
    Assume the expected scoring rule $S(P_\theta, P_0) $  has a unique minimizer $ \theta^\star $ and the prior $\pi(\theta)$ is continuous and positive at $\thetastar$. Further, denote now by $ \pis^{*}\left(\cdot | \Ddobs_{\mathbf{n}}\right) $ the density of $ \sqrt{n}\left(\theta-\hatthetan\left(\Ddobs_{\mathbf{n}}\right)\right) $ when $ \theta\sim\pi_S\left(\cdot | \Ddobsn\right) $, where $  \hatthetan\left(\Ddobsn \right) $ is a sequence which converges almost surely to $ \thetastar $ as $ n\to\infty $. 
    Then, under technical assumptions (\ref{ass:thetastar} to \ref{ass:prior} in Appendix~\ref{app:normal_asymptotic_proof}), as $n \rightarrow \infty$, with probability 1 over $ \Ddobsn $:
		$$
		\int_{\mathbb{R}^p}\left|\pis^{*}\left(s | \Ddobs_{\mathbf{n}}\right)-\mathcal{N}\left(s|0,H_\star^{-1}\right)\right| d s  \to 0,
		$$
		where $ \mathcal{N}(\cdot|0, \Sigma) $ denotes the density of a multivariate normal distribution with zero mean vector and covariance matrix $ \Sigma $.    
    \end{theorem}
	Theorem~\ref{Th:asymptotic_short} implies that the SR posterior concentrates, with probability 1, on the parameter value minimizing the expected SR, if that minimizer is unique. This holds for a well specified model and strictly proper $ S $, in which case the SR posterior concentrates on the true parameter value; this property is usually referred to as \textit{posterior consistency}. However, the minimizer can be unique for misspecified or non-strict SRs as well.

    In general, the asymptotic covariance matrix $H_\star$ does \textit{not} match that of the frequentist minimizer of the SR, implying that asymptotic credible sets do not have correct frequentist coverage, even for strictly proper SR and well-specified model. This instead occurs when choosing $S$ to be the log-score and $w=1$ (thus recovering the standard posterior) with well-specified models (Section 4.1.2 in \citealp{ghosh2006introduction}). While this is a drawback of the SR posterior, we remark again how the latter is tractable for simulator models while the standard posterior is not. Additionally, in misspecified scenarios, the SR posterior achieves the outlier robustness properties discussed in Sec.~\ref{sec:robustness}, while, in that case, the standard posterior would not have exact coverage properties neither outlier robustness. Finally, in case one wants to provide correct credible sets, promising recent work addressing this mismatch \citep{frazier2023reliable} is applicable to the SR posterior.

	\begin{remark}[\textbf{Non-invariance to change of data coordinates -- continued}]
		Following on from Remark~\ref{remark:non-invariance}, notice that $ \thetastar $ depends on the data coordinates, unless the model is well specified and $ S $ is strictly proper. If that is not the case, SR posteriors using different data coordinates will concentrate on different parameter values in general. This property is coherent with the SR posterior learning about the parameter value which minimizes the expected scoring rule, which in turn depends on the chosen coordinate system. See Appendix~\ref{app:coordinates} for more details.
	\end{remark}

	\subsection{Finite-sample generalization bound}\label{sec:posterior_consistency}

	We now consider the energy and kernel score posteriors and their corresponding divergences, and provide a bound on the probability of deviation of the posterior expectation of the divergence from the minimum divergence achievable by the model. The bound holds with finite number of samples and does not require the model to be well specified nor the minimizer of the divergence to be unique. Such results are usually referred to as generalization bounds \citep{cherief2020mmd}.
	For our bound to hold, we require the following \textit{prior mass condition} with respect to a divergence $ D $:
	\begin{enumerate}[label=\textbf{A\arabic*}]
		\item \label{ass:prior_mass} The prior has density $ \pi(\theta) $ (with respect to Lebesgue measure) which satisfies
		\begin{equation}
		\int_{B_{n}\left(\alpha_{1}\right)} \pi(\theta) \mathrm{d} \theta \geq e^{-\alpha_{2} \sqrt{n}}
		\end{equation}
		for some constants $ \alpha_1, \alpha_2>0 $ and for all positive $ n \in \mathbb{N} $, where we define the sets: $$ B_{n}\left(\alpha_{1}\right):=\left\{\theta \in \Theta:\left|D\left({P}_{\theta}, {P}_0\right)-D\left({P}_{\thetastar},  {P}_0\right)\right| \leq \alpha_{1} / \sqrt{n}\right\}, $$
		where $ \thetastar \in \argmin_{\theta \in \Theta} D(P_{\theta}, P_0) $, which is assumed to be nonempty.
	\end{enumerate}

	Assumption~\ref{ass:prior_mass} constrains the amount of prior mass given to $ D $-balls with size decreasing as $ n^{-1/2} $ to decrease slower than $ e^{-\alpha_{2}\sqrt{n}} $ for some $ \alpha_{2} $. It is therefore a weak condition, as it bounds the mass by a quickly decreasing function while the radius is decreasing more slowly. Similar assumptions are taken in \citealt{cherief2020mmd,matsubara2021robust}, where some examples of explicit verification can be found.

	Our result
	(proved in Appendix~\ref{app:posterior_consistency}) assumes either a bounded kernel $ k $ for the kernel score posterior, or bounded $ \X $ for the energy score posterior.

	\begin{theorem}\label{Th:posterior_consistency}
		The following two statements hold for any $\epsilon>0 $:
		\begin{enumerate}

			\item Let the kernel $ k $ be such that $ \sup_{x \in \X} k(x,x) \le\kappa<\infty  $, and let $ D_k $ be the divergence associated to $ S_k $. Consider $ \thetastar \in \argmin_{\theta \in \Theta} D_k(P_{\theta}, P_0)  $; if the prior $ \pi(\theta) $ satisfies Assumption~\ref{ass:prior_mass} for $ D_k $, we have for the kernel Score posterior $ \pik $:
			\begin{equation}
			{P_0}\left(\left|  \int_{\Theta} D_k(P_\theta, P_0) \pik(\theta|\Ddobs_{\mathbf{n}} ) \mathrm{d} \theta-D_k(P_{\thetastar}, P_0)\right| \geq  \epsilon\right)  \leq 2 e^{ -\frac{1}{2}\left(\frac{\sqrt{n}\epsilon - \alpha_{1} - \alpha_{2}/w}{8\kappa}\right)^2 }.
			\end{equation}
			\item Assume the space $ \X $ is bounded such that $ \sup_{x,y \in \X}\|x-y\|_2 \le B <\infty $, and let $ \DEbeta  $ be the divergence associated with $ S_E^{(\beta)} $. Consider $ \thetastar \in \argmin_{\theta \in \Theta} \DEbeta(P_{\theta}, P_0)  $; if the prior $ \pi(\theta) $ satisfies Assumption~\ref{ass:prior_mass} for $ \DEbeta $, we have for the energy score posterior $\piebeta $:
			\begin{equation}
			{P_0}\left(\left|  \int_{\Theta} \DEbeta(P_\theta, P_0) \piebeta(\theta|\Ddobs_{\mathbf{n}} ) \mathrm{d} \theta-\DEbeta(P_{\thetastar}, P_0)\right| \geq  \epsilon\right)  \leq 2 e^{ -\frac{1}{2}\left(\frac{\sqrt{n}\epsilon - \alpha_{1} - \alpha_{2}/w}{8B^\beta}\right)^2 }.
			\end{equation}
		\end{enumerate}
	\end{theorem}
	As $ \epsilon $	or $ n $ increases, the bound on the probability tends to $ 0 $; for $ n\to\infty $, this implies that the SR posterior concentrates on those parameter values for which the model achieves minimum divergence from the data generating process $ P_0 $, ensuring therefore consistency in the well-specified case. With respect to Theorem~\ref{Th:asymptotic_short}, Theorem~\ref{Th:posterior_consistency} provides guarantees on the infinite sample behavior of the SR posterior even when $ \thetastar $ is not unique; however, this result does not describe the specific form of the asymptotic distribution, which Theorem~\ref{Th:asymptotic_short} instead does.

 	\begin{figure}[htbp!]
      \centering
    \begin{subfigure}[b]{1.0\textwidth}
        \centering
        \includegraphics[width=0.9\linewidth]{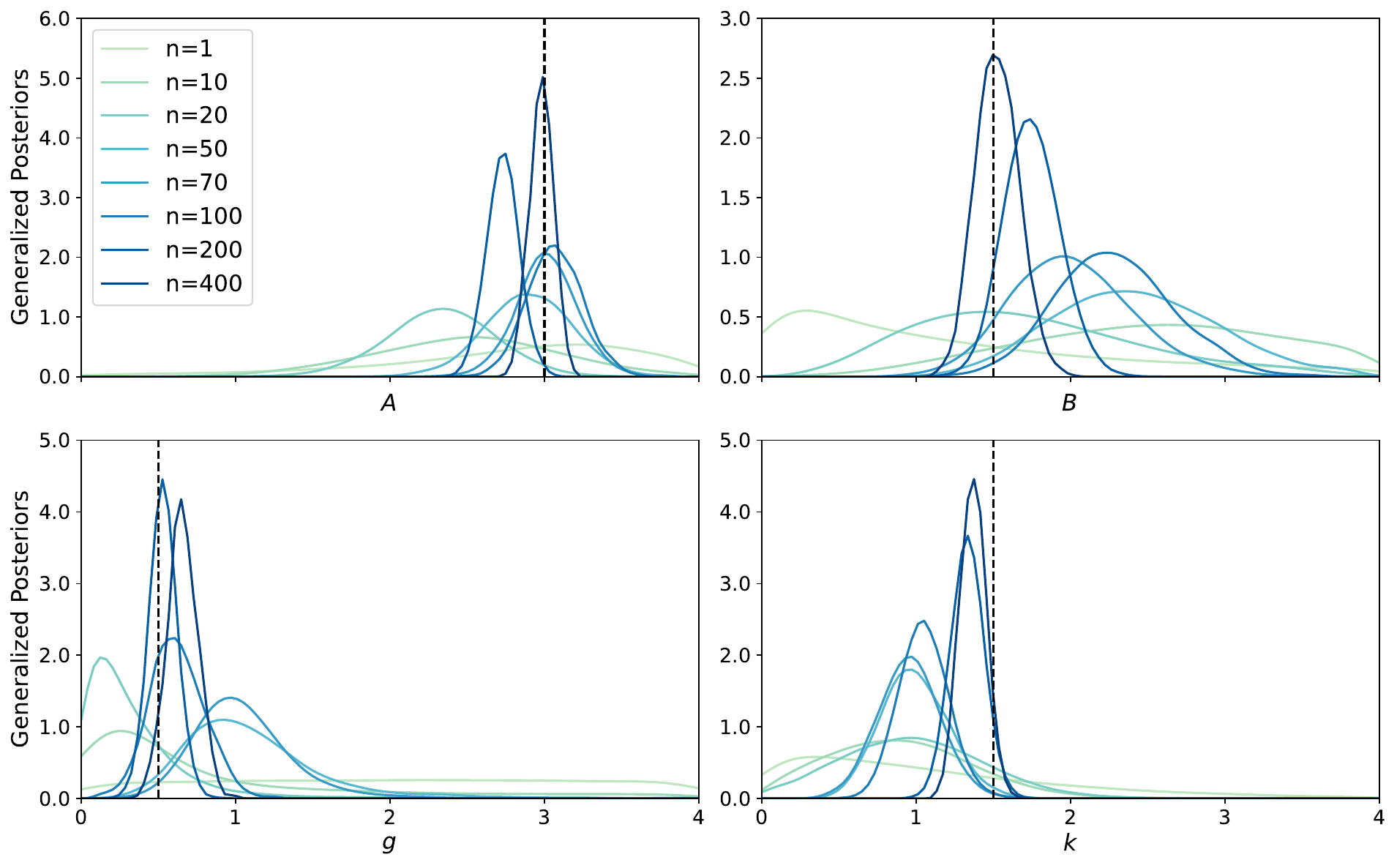}   
        \caption{Marginals of energy score posterior.}
	\end{subfigure}
     \begin{subfigure}[b]{1.0\textwidth}
        \centering
	   \includegraphics[width=0.9\linewidth]{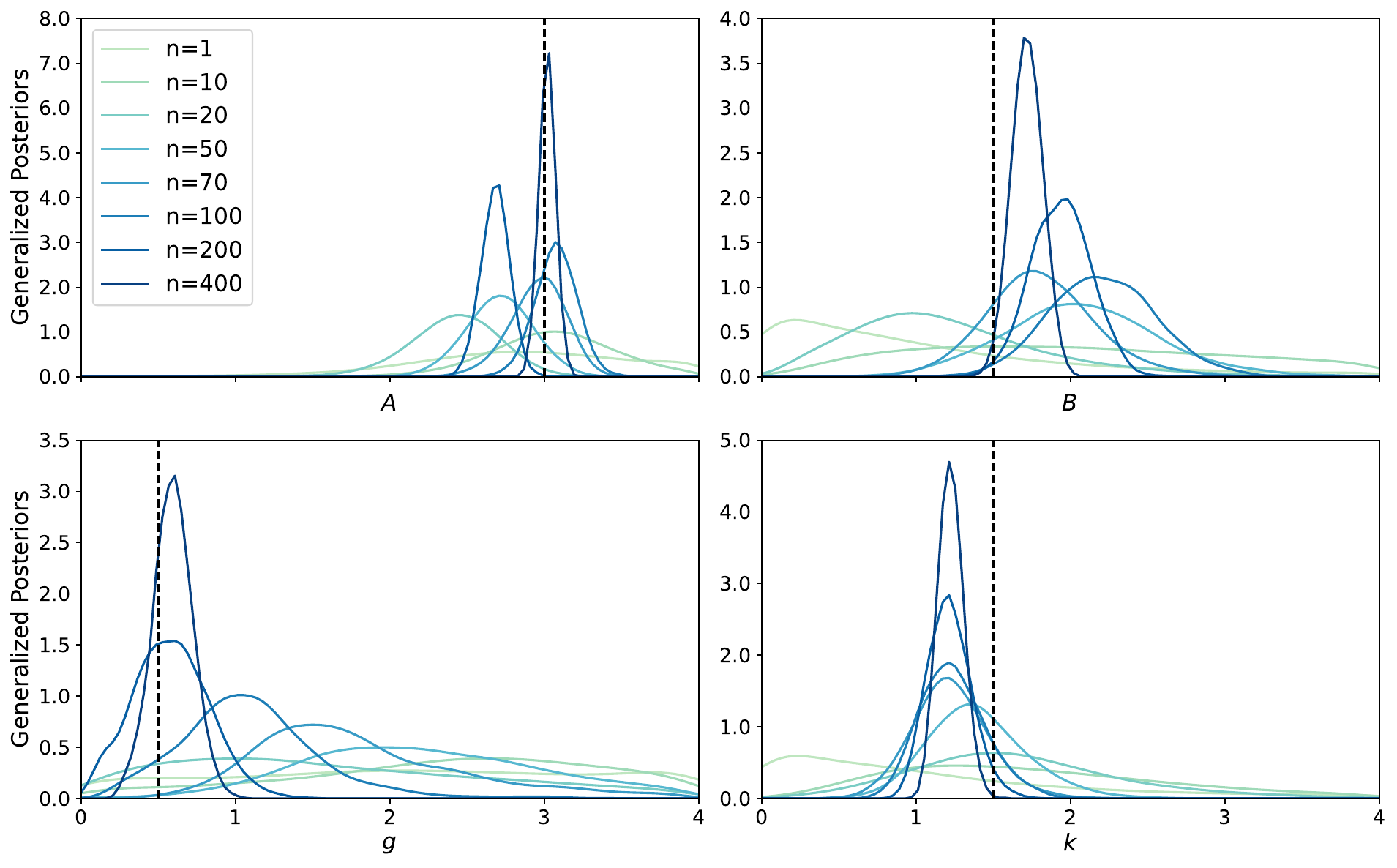}        
        \caption{Marginals of kernel score posterior.}
	\end{subfigure}
		\caption{\textbf{Posterior concentration of univariate g-and-k model,} illustrated by marginals of (a) energy score and (b) kernel score posteriors for the different parameters of the univariate g-and-k model, with increasing number of observations ($n= 1, 10, 20, \ldots, 400) $. Darker (respectively lighter) colors denote a larger (smaller) number of observations. The densities are obtained by kernel density estimator on the MCMC output. The energy and kernel score posteriors concentrate around the true parameter value (dashed vertical line).
              }
		\label{fig:post_concentration_uni_gk}
	\end{figure}

    \subsection{Posterior concentration of univariate g-and-k model}
    \label{sec:fix_w}
    To empirically evaluate the concentration of the SR posterior, we consider the g-and-k model introduced in Sec.~\ref{sec:comparison} and sample from the energy and kernel score posteriors for an increasing number of observations generated from $ A^\star =3,$ $ B^\star=1.5,$ $ g^\star=0.5,$ $ k^\star=1.5$. 

    For the same value of $w$, the scale of the two SR posteriors is different as it depends on the values taken by the SR itself. As here we aim to compare the concentration speed of the two posteriors, we set $w$ such that they have roughly the same scale (for the same number of observations. In other use cases,  as mentioned in the introduction, $w$ can be selected to achieve different goals (often, to match some frequentist property, see \citealp{lyddon2019general, syring2019calibrating, matsubara2022generalised}).

    In practice, we adapt a method proposed in \cite{bissiri2016general} which does not require repeated posterior inference and knowledge of the likelihood function. Specifically, notice that:
	\begin{equation}
	\log \underbrace{\left\{\frac{\pis(\theta | \dobs)}{\pis\left(\theta^{\prime} | \dobs\right)}/ \frac{\pi(\theta)}{\pi\left(\theta^{\prime}\right)}\right\}}_{\BF_S(\theta, \theta^\prime; \dobs)}=-w\left\{S(P_\theta, \dobs)-S(P_{\theta^{\prime}}, \dobs)\right\} %
	\end{equation}
	where $ \BF_S(\theta, \theta^\prime; \dobs) $ denotes the Bayes Factor of $ \theta $ with respect to $ \theta' $ for observation $ \dobs $. Therefore, $ w $ can be determined by fixing $ \BF_S(\theta, \theta^\prime; \dobs) $ for a single choice of $ \theta, \theta',\dobs$.
	Consider now another SR posterior $ \pi_{S'} (\theta|\dobs) $ with Bayes Factor $ {\BF}_{S'} $; setting:
	\begin{equation}
	w = -\frac{\log {\BF_{S'}}(\theta, \theta^\prime; \dobs)}{S(P_\theta, \dobs)-S(P_{\theta^{\prime}}, \dobs)},
	\end{equation}
	ensures $ {\BF_{S'}}(\theta, \theta^\prime; \dobs) = {\BF_S}(\theta, \theta^\prime; \dobs) $. If  $\pi_S$ and $\pi_{S'}$ are obtained from the same prior distribution and the latter uses $w=1$, that corresponds to $w\left\{S(P_\theta, \dobs)-S(P_{\theta^{\prime}}, \dobs)\right\} = S'(P_\theta, \dobs)-S'(P_{\theta^{\prime}}, \dobs) $
 As we have no reason to prefer a specific choice of $ (\theta, \theta') $, we set $ w $ to be the median of $ \frac{S'(P_\theta, \dobs)-S'(P_{\theta^{\prime}}, \dobs)}{S(P_\theta, \dobs)-S(P_{\theta^{\prime}}, \dobs)} $ over values of $ (\theta, \theta') $ sampled from the prior.
    In doing so, we ensure the median variation of the SR (multiplied by the corresponding $w$  between two parameter values sampled from the prior is the same across the two posteriors.
    Additionally, if $ P_\theta $ is an intractable-likelihood model, we estimate $ S(P_\theta,\dobs) $ and $ S'(P_\theta,\dobs) $ by generating data $ \ddsimmtheta $ for each considered values of $ \theta $.

    Hence, we set $w=1$ for the energy score posterior and use the above method to tune $w$ for the kernel score posterior, yielding $w=28.1$; the bandwidth of the Gaussian kernel was tuned as discussed in Appendix~\ref{app:bandwidth}. Figure~\ref{fig:post_concentration_uni_gk} reports the results; with the chosen values of $w$, the two posteriors concentrate at roughly the same speed close to the true parameter values.

    In Appendix~\ref{app:gk_PM} we report similar results achieved with PM-MCMC; due to the stickyness of the chain, those only run satisfactorily up to $n=100$.

	\section{Global bias-robustness of scoring rule posterior}\label{sec:robustness}
	We establish now robustness with respect to contamination in the dataset for the kernel score posterior with bounded kernel and the energy score posterior with bounded $ \X $. %

	First, consider the empirical distribution of the observations $ \empP_n = \frac{1}{n}\sum_{i=1}^n \delta_{\dobs_i} $. If we define:
	\begin{equation}
	L(\theta, \empP_n) := \frac{1}{n} \sum_{i=1}^n S(P_\theta, \dobs_i) = \E_{Y\sim\empP_n} S(P_\theta, \Dobs)
	\end{equation}
	for a scoring rule $ S $, the SR posterior in Eq.~\eqref{Eq:SR_posterior} can be rewritten as: $$ \pis(\theta|\ddobsn) = \pis(\theta|\empP_n) \propto \pi(\theta) \exp\left\{-wn L(\theta, \empP_n) \right\}. $$
	Next, consider the $ \epsilon $\textit{-contamination distribution} $ \empP_{n,\epsilon, z} = (1-\epsilon) \empP_n + \epsilon \delta_z $, obtained by perturbing the fixed empirical distribution with an outlier $ z $ of weight $ \epsilon $. In this setup, the \textit{posterior influence function} \citep{ghosh2016robust} can be defined as:
	\begin{equation}
	\PIF := \left.\frac{{d}}{{d} \epsilon} \pis\left(\theta \left| \empP_{n, \epsilon, z}\right.\right)\right|_{\epsilon=0},
	\end{equation}
	which measures the rate of change of the posterior in $ \theta $ when an infinitesimal perturbation in $ z $ is added to the observations. We say the SR posterior is $ C $-globally bias-robust if: \begin{equation}
	\sup_{\theta\in\Theta} \sup_{z \in \X} \left| \PIF\right|\le C,
	\end{equation}
    for some $C<\infty$. The definition of global bias-robustness in \cite{matsubara2021robust} corresponds to the one above holding for a value $C<\infty$.
	\begin{theorem}\label{Th:robustness}
		The following two independent statements hold:
		\begin{enumerate}

			\item Consider a kernel $ k $ such that $ \sup_{x \in \X} k(x,x) \le\kappa <\infty$; then, the kernel score posterior $ \pik(\cdot|\ddobsn) $ is C-globally bias-robust with $ C  \le 8 w n \kappa e^{6wn\kappa } \sup_{\theta\in\Theta} \pi(\theta)$.

			\item Alternatively, assume the space $ \X $ is bounded such that $ \sup_{x,y \in \X}\|x-y\|_2 \le B <\infty $; then, the energy score posterior $ \piebeta(\cdot|\ddobsn)  $ is globally bias-robust with $ C\le 8 w n B^\beta e^{2wnB^\beta } \sup_{\theta\in\Theta} \pi(\theta)$.
		\end{enumerate}
	\end{theorem}
	Proof is given in Appendix~\ref{app:robustness}. The Gaussian kernel (used across this work)
    is bounded. Our theoretical result does not hold for the energy score posterior when $ \X $ is unbounded. However, in practice (see below) we still find the energy score posterior to be robust to outliers in examples with unbounded $ \X $.

	\begin{figure}[htbp!]
		\includegraphics[width=1.0\linewidth]{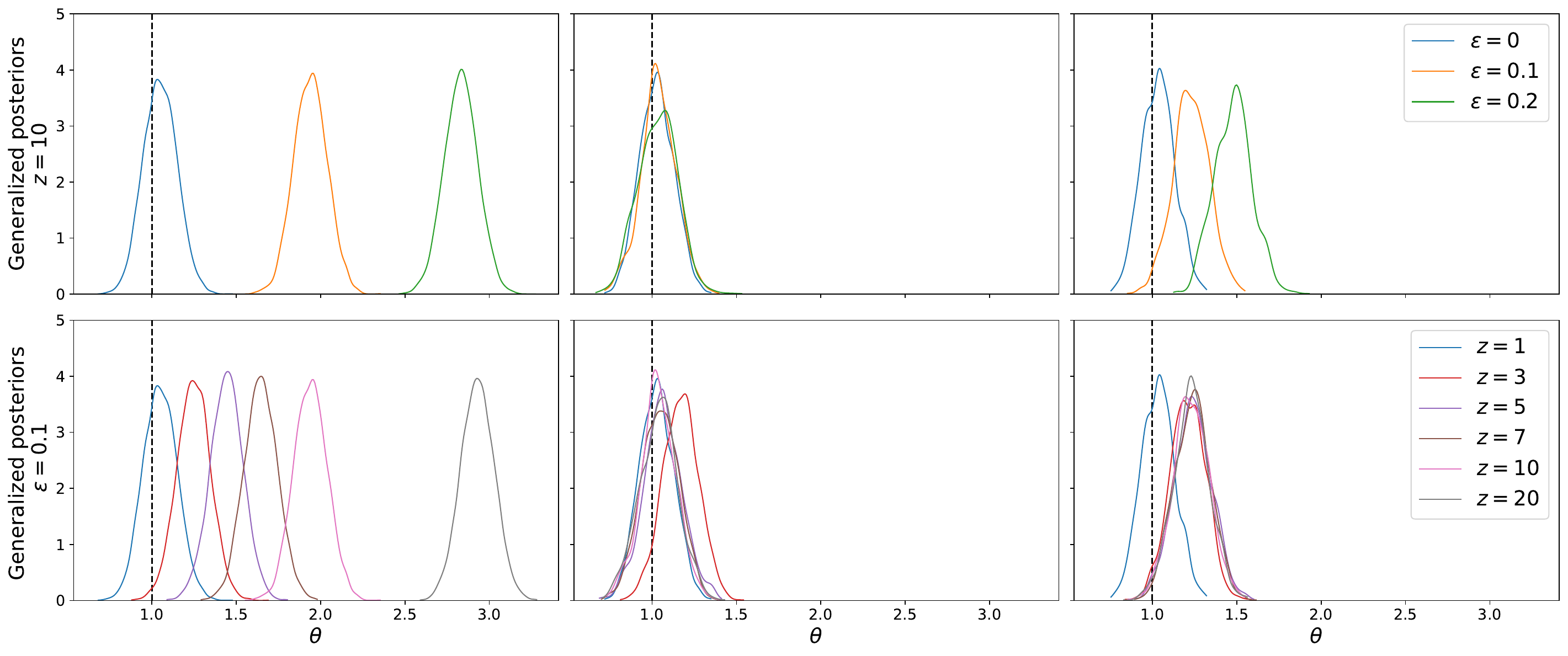}
		\caption{(Standard Bayes, kernel score, energy score from L to R) Posterior distribution for the misspecified normal location model, following experimental setup introduced in \cite{matsubara2021robust}. First row: fixed outliers location $ z=10 $ and varying proportion $ \epsilon $; second row: fixed outlier proportion $ \epsilon $, varying location $ z $. From both rows, it can be seen that both Kernel and energy score are more robust with respect to Standard Bayes. The densities are obtained by KDE on the MCMC output thinned by a factor 10.}
		\label{fig:normal_loc}
	\end{figure}

	\subsection{Robustness for normal location model}\label{sec:normal_location}
	To illustrate robustness of our scoring rule posterior, we consider a univariate normal model with fixed standard deviation $ P_\theta =\mathcal N(\theta,1) $. Similar to
	\cite{matsubara2021robust}, we consider 100 observations, a proportion $ 1-\epsilon $ of which is generated by $ P_\theta $ with $ \theta=1 $, while the remaining proportion $ \epsilon $ is generated by $ \mathcal N(z,1) $ for some value of $ z $. Therefore, $ \epsilon $ and $ z $ control respectively the number and location of outliers. The prior distribution on $ \theta $ is set to $ \mathcal N(0,1) $. To perform inference with our proposed SR posterior, we employ correlated pseudo-marginal MCMC with $ m=500 $, $ G=50 $ and 60000 MCMC steps, of which 40000 are burned-in. Additionally, we perform standard Bayesian inference (as the likelihood is available here).
	For the SR posteriors, $ w $ is fixed in order to get approximately the same posterior variance as standard Bayes in the well-specified case ($ \epsilon=0 $); values are reported in Appendix~\ref{app:normal_location}, together with the proposal sizes for MCMC and the resulting acceptance rates.

	We consider $ \epsilon $ taking values in $ (0,0.1,0.2) $ and $ z $ in $(1,3,5,7,10,20) $; in Fig.~\ref{fig:normal_loc}, some results are shown.
	Results for all combinations of $z$ and $ \epsilon $ are available in Fig.~\ref{fig:normal_loc_2} in Appendix.
	The kernel score posterior is highly robust with respect to outliers, while the energy score posterior performs slightly worse. As expected, the standard Bayes posterior shifts significantly when either $ \epsilon $ or $ z $ are increased. We highlight that Theorem~\ref{Th:robustness} only ensures robustness for small values of $ \epsilon $ and all values of $ z $ for the kernel score posterior, which is in fact experimentally verified (the robustness result for the energy score posterior does not apply here as $ \X $ is unbounded); however, we find empirically that both SR posteriors are more robust than the standard Bayes one, when both $ z $ and $ \epsilon $ are increased.

	\section{Empirical comparison with popular LFI methods}\label{sec:experiments}
	We present here simulation studies to compare our approach to two popular LFI schemes, Bayesian Synthetic Likelihood (BSL, \citealp{price2018bayesian}) and Approximate Bayesian Computation (ABC, \citealp{lintusaari2017}), and showcase the ability of SG-MCMC to sample from the scoring rule posterior of models with high-dimensional parameter space. Precisely, we first study the posterior concentration of the energy and kernel score  posteriors compared to BSL in Sec.~\ref{sec:BSL_multi_gk} for both well-specified and misspecified models; next, in Sec.~\ref{sec:abc_stoch_l96}, we consider a meteorological model with high-dimensional timeseries dataset, and compare the posterior predictive accuray of the scoring rule posterior with that obtained with SMC-ABC \citep{del2012adaptive}. Finally in Sec.~\ref{sec:neural_stoch_l96}, we consider a neural extension of the meteorological model considered in Sec.~\ref{sec:abc_stoch_l96} with a high-dimensional ($> 100$) parameter space; there, SG-MCMC allows to sample from the high-dimensional SR posterior, thus enabling a better posterior predictive accuracy than the lower dimensional model considered in Sec.~\ref{sec:abc_stoch_l96}. 

	Throughout, the kernel score uses the Gaussian kernel with bandwidth set from simulations as illustrated in Appendix~\ref{app:bandwidth}; further, we set $ w=1 $ in the energy score posterior and set $w$ for the kernel score posterior with the strategy discussed in Sec.~\ref{sec:fix_w}.
	The LFI techniques are run using the \texttt{ABCpy} Python library \citep{dutta2021abcpy}, 
    code for reproducing all results is available at \href{https://github.com/Shermjj/GenBayes_LikelihoodFree_ScoringRules_SGMCMC}{this link}.

	\subsection{Comparison with Bayesian synthetic likelihood: Multivariate g-and-k model}\label{sec:BSL_multi_gk}

    Bayesian Synthetic Likelihood (BSL, \citealp{price2018bayesian}) considers the following approximate posterior:
 \begin{equation}\label{Eq:SL_exact_post}
	\piSL(\theta|\dobs) \propto \pi(\theta) \mathcal{N}(\dobs; \mu_\theta, \Sigma_\theta),
	\end{equation}
	where $ \mathcal{N}(\dobs; \mu_\theta, \Sigma_\theta) $ denotes the multivariate normal density with mean vector $ \mu_\theta $ and variance matrix $ \Sigma_\theta $ evaluated in $ \dobs $. BSL is a specific case of our SR posterior (Eq.~\ref{Eq:SR_posterior}) for $w=1$  and the so-called \textit{Dawid--Sebastiani} scoring rule (Appendix~\ref{app:DS_score}), which is non-strictly proper (hence,  multiple minimizers of the expected score can exist even for well-specified models, which implies that the posterior may fail to concentrate asymptotically).
  
    A PM-MCMC where empirical estimates of $ \mu_\theta$ and $\Sigma_\theta $ are obtained from model simulations can be used to sample from an approximation of the BSL posterior, analogously to what we discussed in Sec.~\ref{sec:simulator_models}; it is instead impossible to obtain unbiased gradient estimates of the log-posterior, which prevents SG-MCMC from being applied.

    We consider here the multivariate extension \cite{drovandi2011likelihood, jiang2018approximate} of the univariate g-and-k model introduced earlier. 
    Specifically, we draw a multivariate normal $ (Z^1, \ldots, Z^5) \sim \mathcal{N}(0, \Sigma) $, where $ \Sigma\in \R^{5\times 5} $ has a sparse correlation structure: $ \Sigma_{kk} = 1 $, $ \Sigma_{kl}=\rho $ for $ |k-l| =1 $ and 0 otherwise; each component of $ Z $ is then transformed as in the univariate case (Eq.~\ref{Eq:gk}). The sets of parameters are $ \theta = (A,B,g,k,\rho) $. We use uniform priors on $ [0,4]^4 \times [-\sqrt 3/3, \sqrt 3/3]$.  

    For BSL, we use correlated PM-MCMC with $ m= 500 $, $ G=500 $ and run for 110000 steps, of which 10000 are burned in. For the SR posteriors, we instead use adSGLD, similarly with $m = 500$ and with 110000 steps and 10000 burn-in.  Additional details are given in Appendix~\ref{app:gk_SG}.

    \looseness=-1
    In Appendix~\ref{app:PM_results}, results obtained using PM-MCMC for the SR posteriors are provided. The same appendix provides results for BSL on the univariate g-and-k model; there, PM-MCMC run satisfactorily up to $n=100$,  showing how BSL fails to concentrate as it is based on a non-strictly proper SRs.

    \begin{figure}[htbp!]
    \centering
    \begin{subfigure}[b]{1.0\textwidth}
		\includegraphics[width=0.9\linewidth]{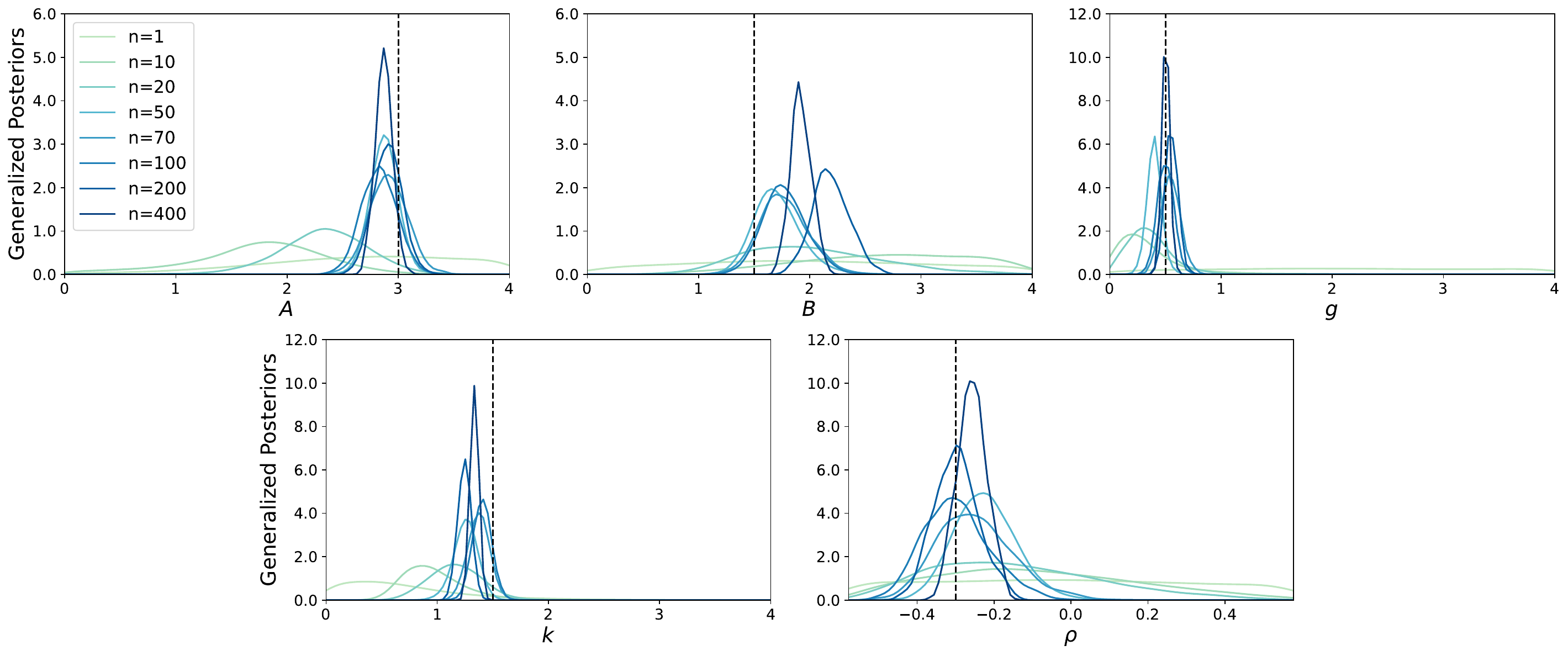}
  \caption{Marginals of energy score posterior}  
    \end{subfigure}
    \begin{subfigure}[b]{1.0\textwidth}
  		\includegraphics[width=0.9\linewidth]{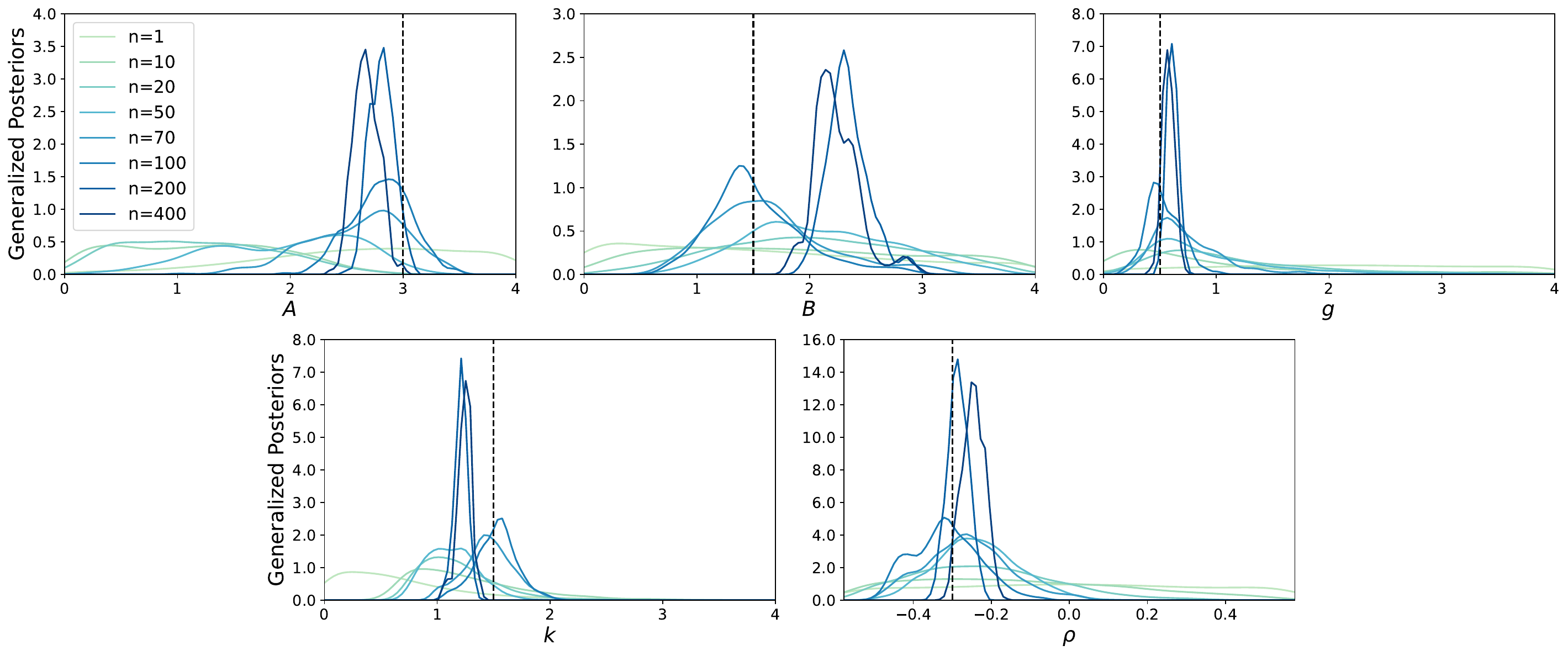}
	\caption{Marginals of kernel score posterior}
    \end{subfigure}
    \begin{subfigure}[b]{1.0\textwidth}
        	\includegraphics[width=0.9\linewidth]{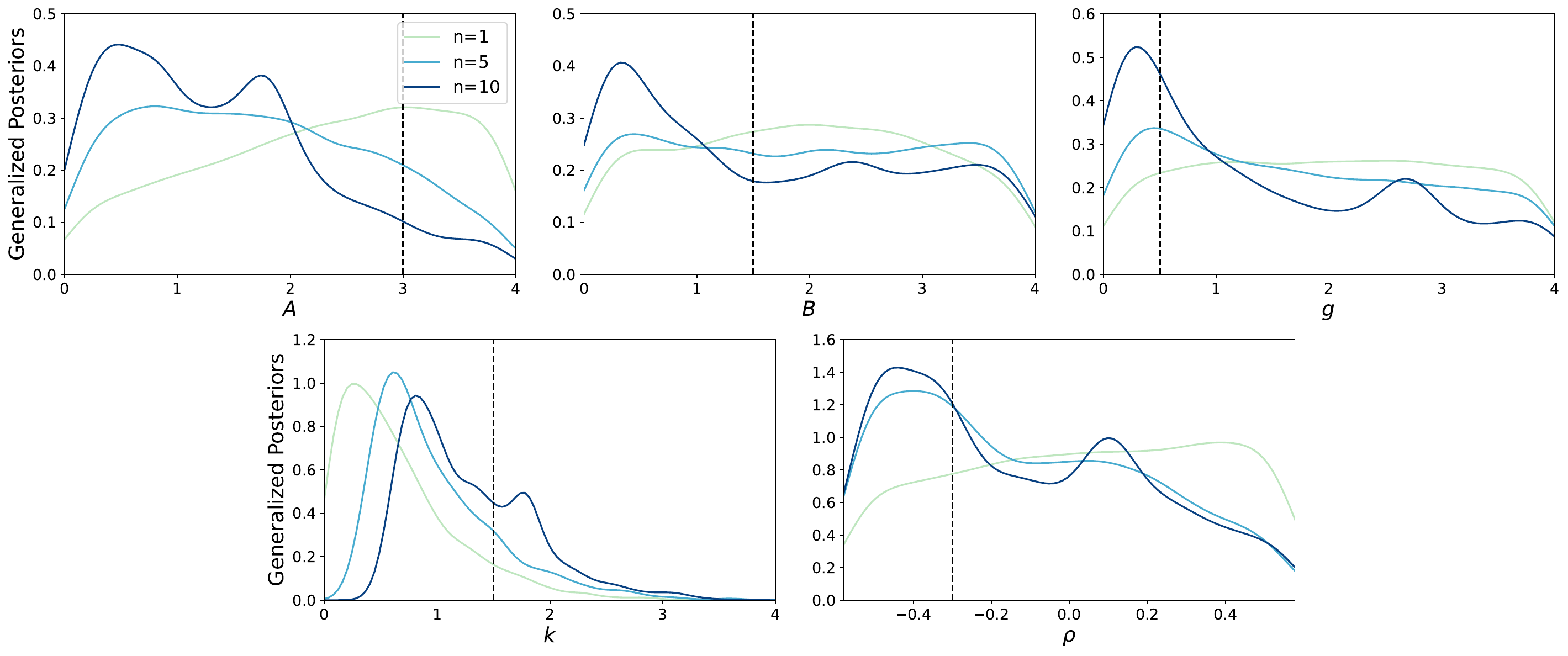}
	\caption{Marginals of Bayesian synthetic likelihood}  
    \end{subfigure}
  	\caption{\textbf{Posterior concentration of well-specified multivariate g-and-k model,} illustrated by marginals of (a) energy score, (b) kernel score and (c) Bayesian synthetic likelihood posteriors, with increasing number of observations ($n= 1, 10, \ldots, 400) $. Darker (respectively lighter) colors denote a larger (smaller) number of observations. The vertical line represents the true parameter value. Both the energy and kernel score posteriors (run with adSGLD) concentrate close to the true parameter value, while BSL was run with PM-MCMC which did not converge for $n>10$.       }
		\label{fig:gk_new}
	\end{figure}

	\subsubsection{Well-specified case}\label{sec:experiments_gk_well_spec}
	We consider synthetic observations generated from parameter values $ A^\star =3,$ $ B^\star=1.5,$ $ g^\star=0.5,$ $ k^\star=1.5$ and $ \rho^\star=-0.3 $. 
    The results are given in Figure~\ref{fig:gk_new}. 
    With increasing $ n $, both the energy and kernel score posterior concentrates close to the true value for all parameters (dashed vertical line), as expected when using strictly proper SRs. For this example, the PM-MCMC targeting the BSL posteriors do not converge beyond respectively 1 and 10 observations.

    \begin{figure}[htbp!]
    \centering
    \begin{subfigure}[b]{1.0\textwidth}
		\includegraphics[width=0.9\linewidth]{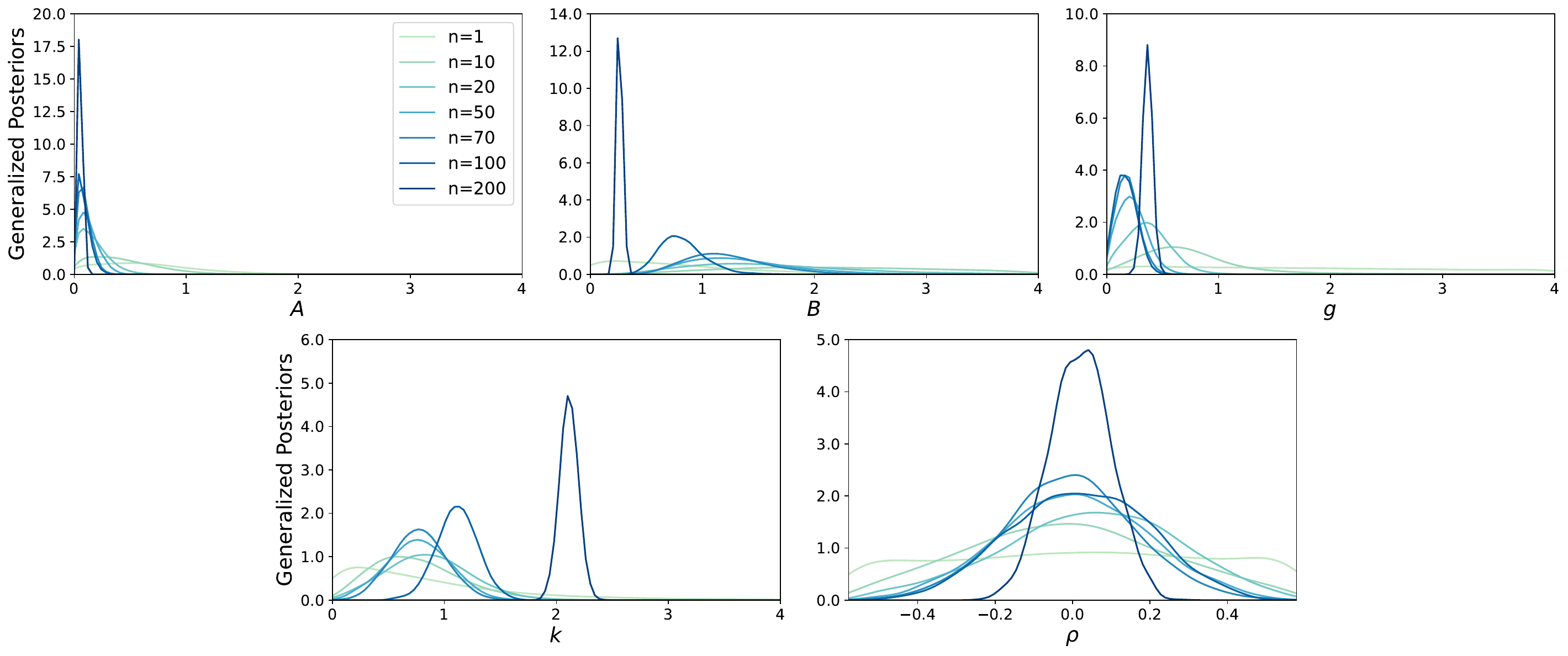}
  \caption{Marginals of energy score posterior}  
    \end{subfigure}
    \begin{subfigure}[b]{1.0\textwidth}
  		\includegraphics[width=0.9\linewidth]{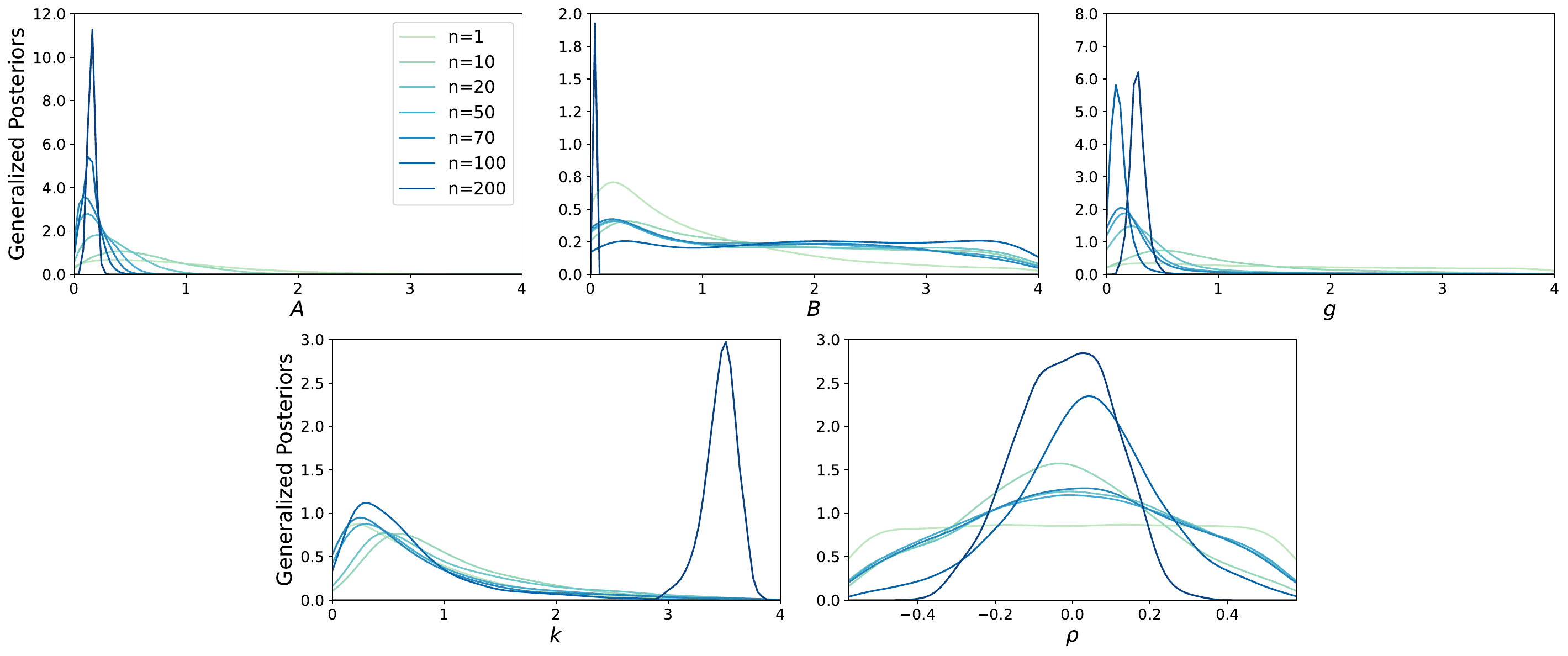}
	\caption{Marginals of kernel score posterior}
    \end{subfigure}
    \begin{subfigure}[b]{1.0\textwidth}
        	\includegraphics[width=0.9\linewidth]{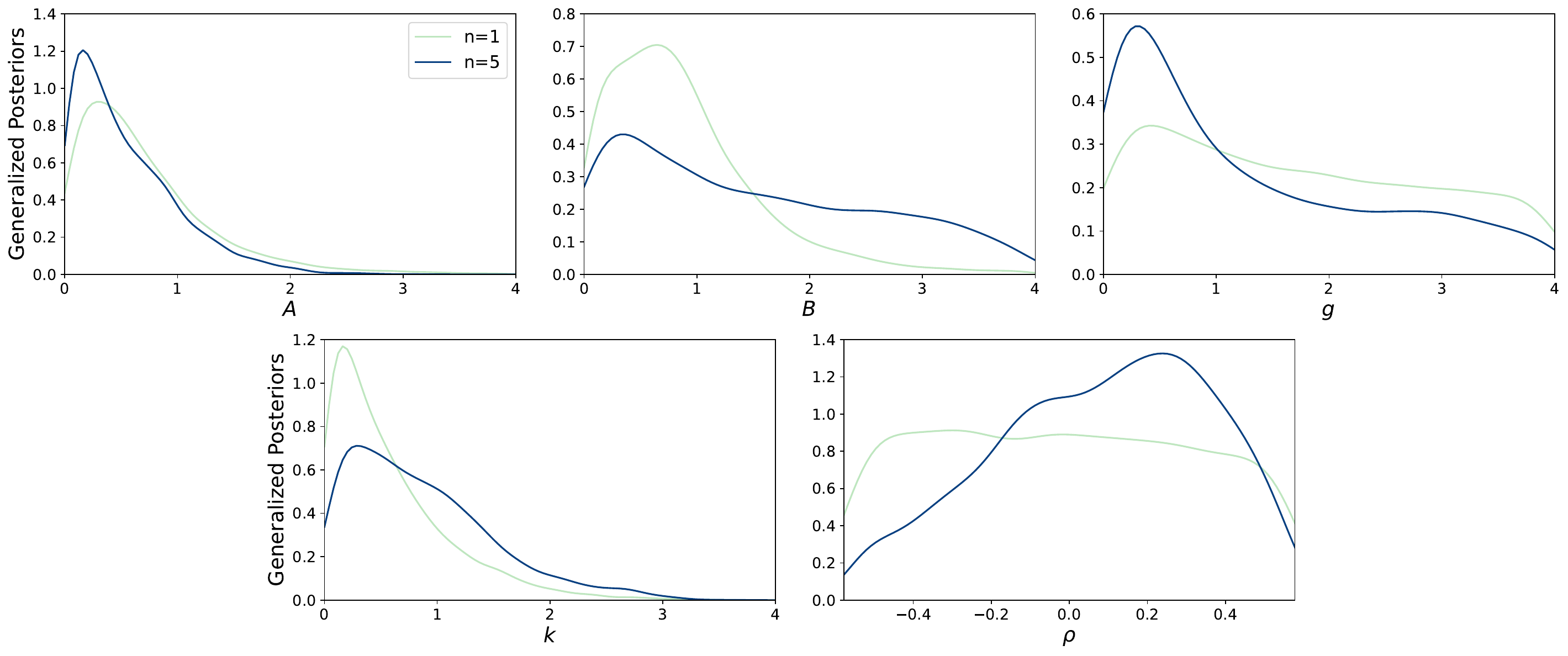}
	\caption{Marginals of Bayesian synthetic likelihood}  
    \end{subfigure}
  	\caption{\textbf{Posterior concentration of misspecified multivariate g-and-k model,} illustrated by marginals of (a) energy score, (b) kernel score and (c) Bayesian synthetic likelihood posteriors, with increasing number of observations ($n= 1, 10, \ldots, 400) $. Darker (respectively lighter) colors denote a larger (smaller) number of observations. The vertical line represents the true parameter value. Both the energy and kernel score posteriors (run with adSGLD) concentrate close to the true parameter value, while BSL was run with PM-MCMC which did not converge for $n>5$. }	\label{fig:gk_cauchy_new}
	\end{figure}

	\subsubsection{Misspecified setup}\label{sec:experiments_gk_misspec}

	Next, we consider as data generating process the Cauchy distribution, which has fatter tails than the g-and-k one. 
    The five components of each observation are drawn independently from the univariate Cauchy distribution (i.e., no correlation between components). For the SR posteriors, we use the values of $ w $ which were obtained with our heuristics in the well-specified case;
	additional experimental details are reported in Appendix~\ref{app:gk_Cauchy_prop_size}.
    Results are in Figure~\ref{fig:gk_cauchy_new}. The energy and kernel score posteriors concentrate on slightly different parameter value, corresponding to the unique minimzers of the expected SR (which are therefore different in these two cases). The PM-MCMC targeting the BSL posterior did not converge for $n>5$.

	\subsection{Comparison with approximate Bayesian Computation: Stochastic Lorenz96 model}\label{sec:abc_stoch_l96}

     \begin{figure}[htb]
    \centering
  		\begin{subfigure}[b]{\textwidth}
			\centering
			\includegraphics[width=\linewidth]{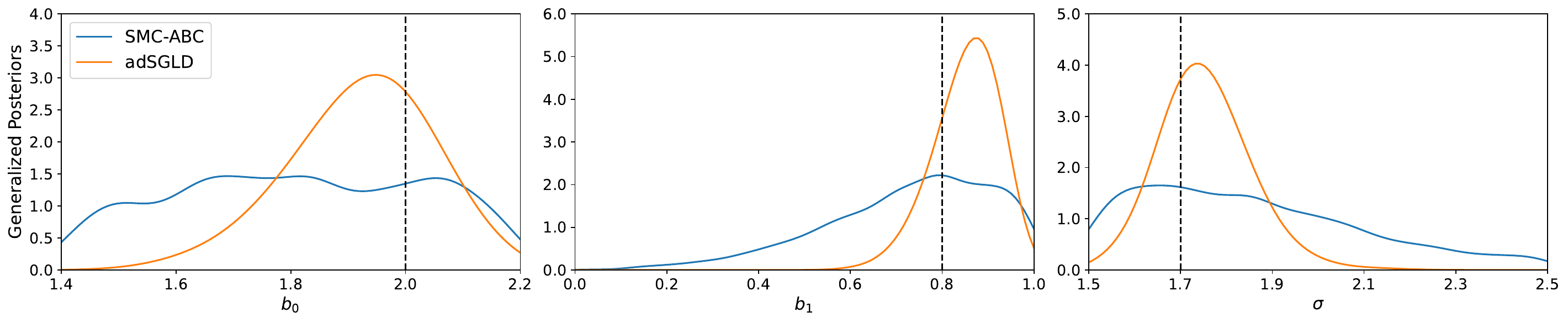}
			\caption{Marginal Posteriors.}
			\label{fig:Lorenz96_post}
		\end{subfigure}
    \begin{subfigure}[b]{0.6\textwidth}
        \centering
        \includegraphics[width=0.8\linewidth]{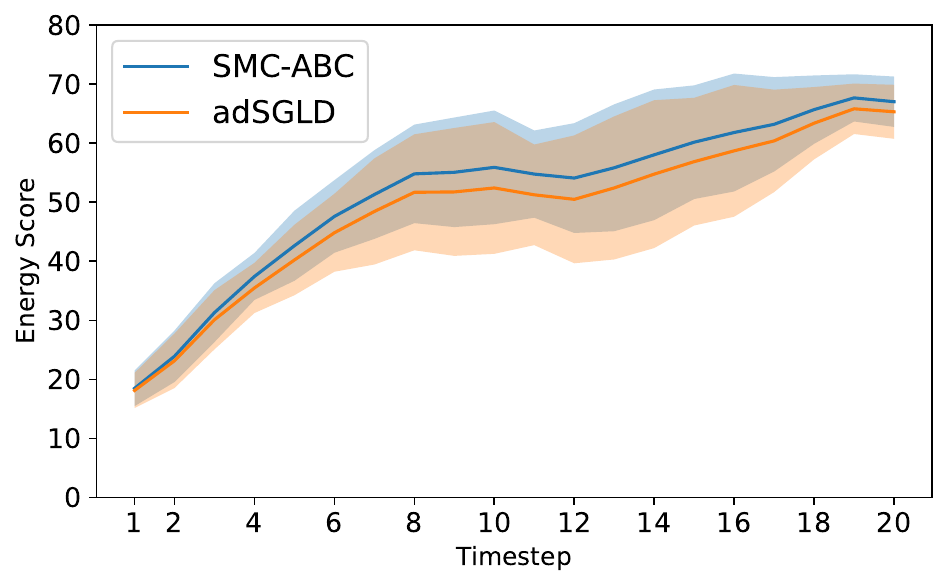}
        \caption{Predictive accuracy using energy score.}
			\label{fig:Lorenz96_predictive}
	\end{subfigure}
		\caption{\textbf{Comparison between SMC-ABC and energy score posteriors inferred using 250,000 model simulations, for the linearly parametrized Lorenz96 model.} (a) Marginal posterior distribution of the parameters of SMC-ABC posterior and energy score posterior using adSGLD, for a single observed set $\ddsimn$ (vertical line representing the true parameter $\thetastar$). (b) Energy score between posterior predictive and each time-step of the original observation. This is repeated for 5 observations $\ddsimn$ (each using $n=10$ here), and a \textit{t}-distribution at each time-step is fitted to the energy score values. The solid line and shaded region respectively represent the mean and the 95\% confidence interval of the fitted \textit{t}-distribution. Lower energy score indicates better predictive performance. }
		\label{fig:stoch-Lorenz-ABC-SGNHT}
    \end{figure}
    
    The Lorenz96 model \citep{lorenz1996predictability} is an important benchmark in meteorology \citep{arnold2013stochastic_paper} and was previously studied in the LFI literature \citep{thomas2020likelihood, jarvenpaa2020batch, pacchiardi2020score}.
    Here, we consider the stochastic \textit{parametrized} version introduced by \cite{wilks2005effects}, defined by the following set of Ordinary Differential Equations (ODEs):
	\begin{equation}\label{Eq:Lorenz}
	\frac{dx_k}{dt} = - x_{k-1}(x_{k-2}-x_{k+1}) -x_k + 10 - g(x_k, t; \theta); \quad k=1, \ldots, K,
	\end{equation}
	where cyclic boundary conditions imply that we take $ K+1 = 1 $ in the indices. The stochastic forcing term $ g $ depends on parameters $ \theta = (b_0, b_1, \sigma_e) $, and is defined upon discretizing the ODEs with a time-step $ \Delta t$:
	\begin{equation}
 \label{eq:lin_param_lorenz}
g(x, t; \theta) = {b_0 + b_1 x} + \sigma_e  \eta(t), \quad \eta(t) \sim \mathcal{N}(0,1).
\end{equation}

    In practice, we took $ K=8 $ and integrated the model using the Euler-Maruyama scheme starting from a fixed initial condition $ x(0) $ for 20 additional time-steps on the interval $ t \in [0,1.5] $ (corresponding to $ \Delta t = 3/40 $). We generate 5 independent sets of observed data $\ddsimn$, each using $n=10$ time-series simulated from the model using $ \theta^\star = ( 2, 0.8, 1.7 )$.  
    As prior distribution, we consider a uniform distribution on the region $ [1.4,2.2]\times[0,1]\times[1.5,2.5] $. 

    We run inference for the energy score posterior using adSGLD with $ m=10 $ and 25000 MCMC steps, of which 5000 are burned-in. We compare the inferred energy score posterior with the posterior obtained by Sequential Monte Carlo Approximate Bayesian Computation (SMC-ABC, \citealp{del2012adaptive}) using the Euclidean distance between simulated and observed dataset as discrepancy measure.
    The SMC-ABC algorithm was run for $25$ generations with $m=10$ simulations for every parameter value to draw $1000$ samples from the posterior distribution; with this setup, the two algorithms each use $250,000$ model simulations. Further details are given in Appendix~\ref{app:Lorenz}.
    The comparison between these two posteriors in Figure~\ref{fig:Lorenz96_post} illustrates how the energy score posterior assigns more probability to parameter values close to $\thetastar$ than the SMC-ABC posterior.
    Moreover, to assess the out-of-sample performance of the inferred posterior, we implement the following posterior predictive check: given draws from a posterior $ \pi(\theta|\ddsimn) $, we generate simulations from the model for the corresponding parameter value, which are therefore samples from the posterior predictive
	\begin{equation}\label{Eq:post_pred}
	p(\dobs_\text{new}|\ddsimn) = \int p(\dobs_\text{new}|\theta) \pi(\theta|\ddsimn)d\theta;
	\end{equation}
	from these samples, we assess how well the posterior predictive matches the original observation by computing the energy score between the posterior predictive distribution and the observations $ \ddsimn $ at each time-step. The results in Figure~\ref{fig:Lorenz96_predictive} show how the energy score posterior predictive matches the original observation than the SMC-ABC posterior predictive.

	\subsection{High dimensional neural stochastic parametrization for Lorenz96}\label{sec:neural_stoch_l96}

    \begin{figure}[htbp!]
    \centering
    \begin{subfigure}[t]{0.46\textwidth}
      \centering
        \includegraphics[width=\linewidth]{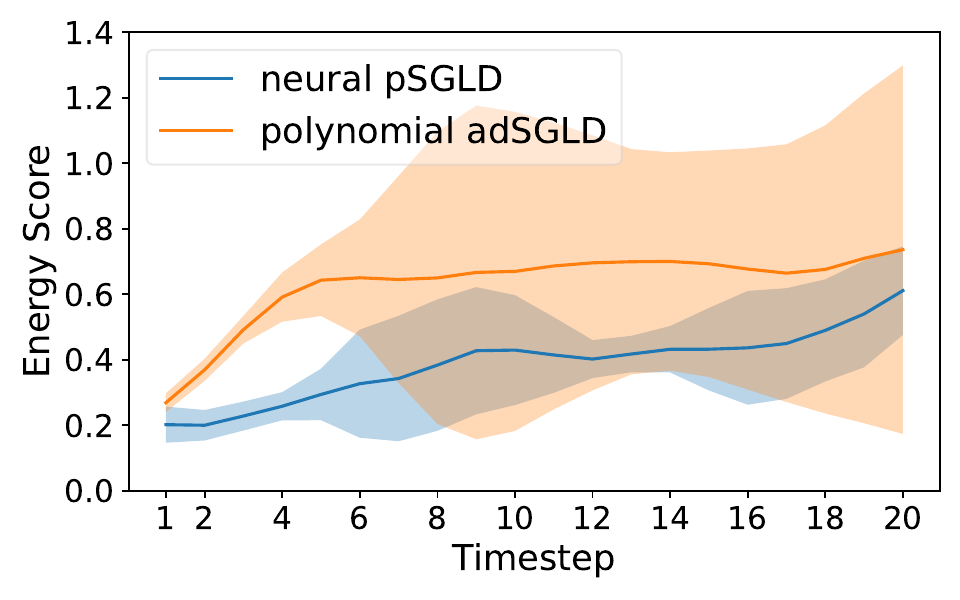}
        \caption{Predictive accuracy using energy score.}
    \end{subfigure}
    \begin{subfigure}[t]{0.46\textwidth}
    \centering
    \includegraphics[width=\linewidth]{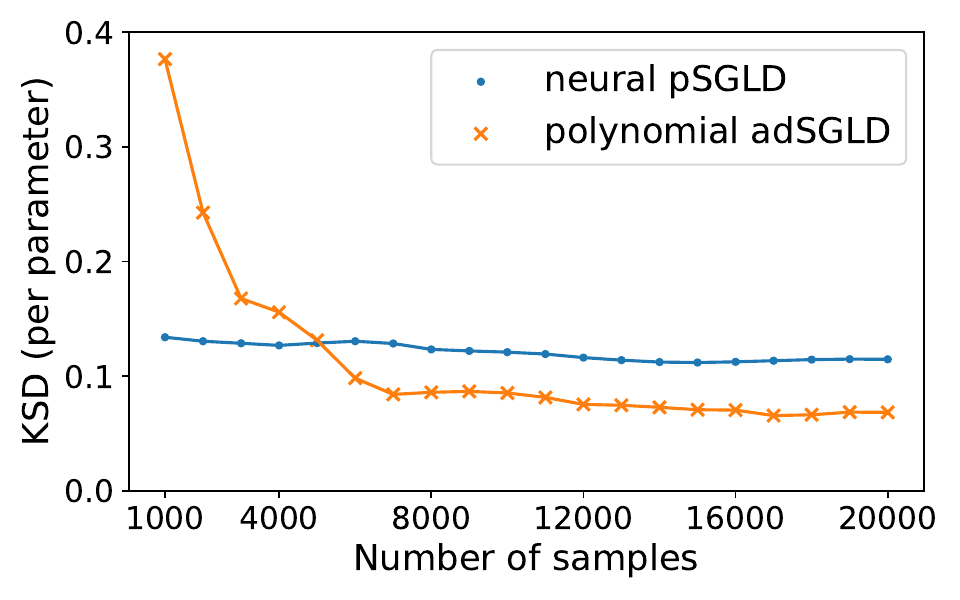}
        \caption{kernelized Stein discrepancy.}
    \end{subfigure}
    \caption{\textbf{Comparison between neural and linear stochastic parametrizations for the Lorenz96 model.} The posterior for the neural parametrization is sampled pSGLD algorithm more suited to high-dimensional spaces than the adSGLD used for the linear one. 
    (a) Energy score between posterior predictive and each time-step of the original observation. This is repeated for 5 observations $\ddsimn$ (each using $n=1$ here), and a \textit{t}-distribution at each time-step is fitted to the energy score values. The solid line and shaded region respectively represent the mean and the 95\% confidence interval of the fitted \textit{t}-distribution. Lower energy score indicates better predictive performance.
 (b) KSD divided by the dimension of parameter space to assess the convergence of adSGLD for linear stochastic parametrization and pSGLD for neural stochastic parametrization.}
    \label{fig:highdim_score_neural}
    \end{figure}
    
    The stochastic model considered in the previous section is a simplification of the original Lorenz96 model \citep{lorenz1996predictability}, which is a chaotic system including interacting slow and fast variables described by the following differential equations:
    \begin{equation}
        \begin{aligned}
    \label{eq:full_lorenz}
    \frac{\mathrm{d} x_k}{\mathrm{~d} t} & =-x_{k-1}\left(x_{k-2}-x_{k+1}\right)-x_k+F-\frac{h c}{b} \sum_{j=J(k-1)+1}^{k J} y_j \\
    \frac{\mathrm{d} y_j}{\mathrm{~d} t} & =-\operatorname{cby}_{j+1}\left(y_{j+2}-y_{j-1}\right)-c y_j+\frac{h c}{b} X_{\mathrm{int}[(j-1) / J]+1},
    \end{aligned}
    \end{equation}
    where $k=1, \ldots, K$, and $j=1, \ldots, J K$, and cyclic boundary conditions are assumed, so that index $k=K+1$ corresponds to $k=1$ and similarly for $j$. %
    
    The stochastic model in Eq.~\eqref{eq:lin_param_lorenz} was derived by considering the part of the above ODE dealing with slow variables only and modelling the effect of the fast variables with the stochastic linear parametrization $g(y,t;\theta)$  \citep{wilks2019369}. To improve on this, we replace that with a high-dimensional parameterisation using a neural network: 
    	\begin{equation*}
    g(x, t; \theta) = f(x ; \theta) + \sigma_e  \eta(t), \quad \eta(t) \sim \mathcal{N}(0,1).
    \end{equation*}
    where $f(x;\theta)$ is a multi-layer perceptron with one hidden layer using a ReLU activation function. Altogether, this model has $111$ parameters, on each of which we put an independent $ \mathcal{N}(0, 10) $ prior.
    
    To compare the linear and neural stochastic parametrizations, we simulate a timeseries from the full Lorenz96 model in equation~\eqref{eq:full_lorenz} and consider this as the observed data, by fixing $K=8,\ J=32,\ h=1,\ b=10,\ c=10$ and $F=10$. We then integrate the above equations with a 4th order Runge-Kutta integrator with $d t=0.001$, starting from $x_k=y_j=0$ for $k=2, \ldots, K$ and $j=2, \ldots J K$ and $x_1=y_1=1$. We discard the first 2 time units and record the values of $\mathbf{x}$ every $\Delta t=0.2$. This is done for a total of $21$ timesteps. We repeat this process $5$ times by perturbing the initial value with Gaussian noise; in this way, we generate 5 observations which slightly differ for the initial conditions (there is no other source of randomness as Eq.~\eqref{eq:full_lorenz} is deterministic).

    For the linearly parametrized Lorenz96 model we follow the same setup as in Sec.~\ref{sec:abc_stoch_l96} and use adSGLD to sample from the energy score posterior. In contrast, we opt to use pSGLD (Sec.~\ref{sec:psgld}) for the $111$-dimensional neural Lorenz96 model. For both cases, we use $m = 500$ and $20 000$ MCMC steps. In Figure~\ref{fig:highdim_score_neural} we compare the inferred Scoring rule posterior via their predictive performance and convergence using KSD divided by the number of  parameters (as the KSD grows linearly with the number of parameters). From this example, it is evident how  SG-MCMC (more specifically pSGLD)  enables sampling  over a very high-dimensional parameter space very efficiently, which allows to leverage a more expressive model to improve the representation of the observed data. 

	\section{Related approaches}\label{sec:related_works}

	Scoring rules have been previously used to generalize Bayesian inference: \cite{giummole2019objective} considered an update similar to ours, but fixed $ w=1 $ and adjusted the parameter value (similarly to what was done in \citealp{pauli2011bayesian} and \citealp{ruli2016approximate}) so that the posterior has the same asymptotic covariance matrix as the frequentist minimum scoring rule estimator. Instead, \cite{loaiza2019focused} considered a time-series setting in which the task is to learn about the parameter value which yields the best prediction, given the previous observations. Finally, \cite{jewson2018principles} motivated Bayesian inference using general divergences (beyond the KL one which underpins standard Bayesian inference) in an M-open setup, and discussed posteriors which employ estimators of the divergences from observed data; some of these estimators can be written using scoring rules. However, none of the above works considered explicitly the LFI setup. 

	A parallel work \citep{matsubara2021robust} investigates the generalized posterior obtained by using a kernelized Stein Discrepancy \citep{chwialkowski2016kernel, liu2016kernelized}. This posterior is shown to satisfy robustness and consistency properties, and is computationally convenient for doubly-intractable models (i.e., for which the likelihood is available, but only up to the normalizing constant). In contrast, our work focuses on models that do not have an explicit likelihood.

    As mentioned before, previous LFI methods such as MMD-Bayes \cite{cherief2020mmd} and BSL \cite{price2018bayesian} fall under our SR posterior framework. So do the semi-parametric BSL \cite{an2020robust} and the ratio-estimation methods \cite{thomas2020likelihood}; we discuss these methods in Appendices~\ref{app:semiBSL} and \ref{app:ratio_estimation}.

	Interestingly, \cite{dellaporta2022robust}, introduced a new LFI method which, similar to ours, enjoys outlier robustness and posterior consistency; however, their method is derived from the Bayesian non-parametric learning framework of \cite{lyddon2018nonparametric, fong2019scalable} rather than the generalized Bayesian posterior of \cite{bissiri2016general}.

    Finally, \cite{duffield2022bayesian} also uses stochastic-gradient MCMC for sampling from a generalized posterior; however, instead of a reparametrization trick, the unbiased gradient estimate is obtained through a specific property of the system they consider (a quantum computer).

	\section{Conclusion}\label{sec:conclusion}

In this work, we introduced a generalized Bayesian posterior for likelihood-free inference relying on scoring rules which can be easily estimated with samples from the simulator model. This \textit{scoring rule posterior} generalizes previous approaches \citep{price2018bayesian, cherief2020mmd}. While pseudo-marginal MCMC enambles approximate sampling of the posterior for simple cases, it mixes pooorly for concentrated targets, even employing advanced schemes \citep{picchini2022sequentially}; hence, we adapted stochastic-gradient MCMC methods to our framework, by exploiting automatic differentiation to compute gradients for the simulator model. As these new sampling schemes allow to sample the scoring rule posterior for high-dimensional parameter spaces,  we were able to empirically validate the concentration and outlier-robustness results we proved theoretically, focusing on the kernel and the energy scores. Our comparison with the popular Approximate Bayesian Computation and Bayesian Synthetic Likelihood showed how the scoring rule posterior enables more informative parameter inference, scaling to higher number of samples and parameters.

 We remark once again how the scoring rule posterior does \textit{not} aim to approximate the standard Bayesian posterior, as most LFI methods do: it instead learns about the parameter value minimizing the expected scoring rule; importantly, outlier robustness is achieved as a consequence of this relaxation. Although we only focused on the specific notion of robustness to outliers, it is possible that suitably-chosen scoring rules provide robustness to other forms of misspecification (such as the distance in Prokhorov metric studied in \citealp{briol2019statistical} or the adversarial contamination method in \citealp{cherief2022finite}); we leave this investigation for future work.

\subsubsection*{Acknowledgment}
LP received support by the EPSRC and MRC through the OxWaSP CDT programme (EP/L016710/1), which also funded part of the computational resources used to perform this work. RD is funded by EPSRC (grant nos. EP/V025899/1, EP/T017112/1) and NERC (grant no. NE/T00973X/1). \\
We thank Jeremias Knoblauch, François-Xavier Briol, Takuo Matsubara, Geoff Nicholls, Benedict Leimkuhler and Sebastian Schmon for valuable feedback and suggestions on earlier versions of this work. We also thank Alex Shestopaloff for providing code for exact MCMC for the M/G/1 model.

\bibliographystyle{abbrvnat}

%{\footnotesize\bibliography{references.bib}}
{\footnotesize
}

\appendix

\section{Proofs of theoretical results}\label{app:proofs}

\subsection{Precise statement and proof of Theorem~\ref{Th:auxiliary_lik}}\label{app:auxiliary_lik}

We recall here for simplicity the useful definitions. We consider the SR posterior:
\begin{equation}\
\pis(\theta|\ddobsn) \propto \pi(\theta) \underbrace{\exp\left\{ - w \sum_{i=1}^n S(P_\theta, \dobs_i) \right\}}_{p_S(\ddobsn|\theta)}.
\end{equation}
Further, we recall the form of the target of the pseudo-marginal MCMC:
\begin{equation}\label{}
\pishat(\theta|\ddobsn) \propto \pi(\theta) \pshat (\ddobsn|\theta),
\end{equation}
where:
\begin{equation}\label{}
\pshat (\ddobsn|\theta) = \E\left[\exp\left\{ - w \sum_{i=1}^n \hat S(\Ddsimmtheta, \dobs_i) \right\} \right] 
=	\int \exp\left\{ - w \sum_{i=1}^n \hat S(\ddsimmtheta, \dobs_i) \right\}  \prod_{j=1}^{m}p(\dsim_j^{(\theta)}|\theta) d\dsim_{1}d\dsim_2\cdots d\dsim_m.
\end{equation}

The complete version of Theorem~\ref{Th:auxiliary_lik} is given in the following: 
\begin{theorem}\label{Th:auxiliary_lik_complete}
	Assume the following:
	\begin{enumerate}
\item $ \hat S(\Ddsimmtheta, \dobs_i) $ converges in probability to $ S(P_\theta, \dobs_i) $ as $ m \to\infty $  for all $ i=1, \ldots, n $.
\item $
\sup _{m} {\E}\left[\left|\exp\{ - w \sum_{i=1}^n \hat S(\Ddsimmtheta, \dobs_i)\}\right|^{1+\delta}\right]<\infty
$ for some $\delta>0$
\item $\inf_{m} \int_{\Theta} \pshat(\ddobsn | \theta) \pi(\theta) d \theta>0$ and $\sup _{\theta \in \Theta} p_S(\ddobsn | \theta)<$
$\infty$.
	\end{enumerate} 	
Then, \begin{equation}\
\lim _{m \rightarrow \infty} \pishat(\theta | \ddobsn)=\pis(\theta | \ddobsn).
\end{equation}
\end{theorem}

\subsubsection{Proof of Theorem~\ref{Th:auxiliary_lik_complete}}
In order to prove Theorem~\ref{Th:auxiliary_lik_complete}, we extend the proof for the analogous result for Bayesian inference with an auxiliary likelihood \citep{drovandi2015bayesian}. Our setup is slightly more general as we do not constrain the update to be defined in terms of a likelihood; notice that the original setup in \cite{drovandi2015bayesian} is recovered when we consider $ S $ being the negative log likelihood, for some auxiliary likelihood.

We begin by stating a useful property:
\begin{lemma}[Theorem 3.5 in \cite{billingsley1999convergence}]\label{lemma:1}
	If $ X_n $ is a sequence of uniformly integrable random variables and $ X_n $ converges in distribution to $ X $, then $ X $ is integrable and $ \E[X_n] \to \E[X] $ as $ n\to\infty $.
\end{lemma}

\begin{remark}[Remark 1 in \cite{drovandi2015bayesian}]\label{lemma:2}
	A simple sufficient condition for uniform integrability is that for some $ \delta >0 $:
	\begin{equation}\
	\sup_n \E[|X_n|^{1+\delta}]  <\infty.
	\end{equation}\end{remark}

The result in the main text is the combination of the following two Theorems, which respectively generalize Results 1 and 2 in \cite{drovandi2015bayesian}:
\begin{theorem}[Generalizes Result 1 in \cite{drovandi2015bayesian}]\label{Th:Res1}
	Assume that $ \pshat(\ddobsn|\theta) \to p_S(\ddobsn|\theta) $ as $ m\to\infty $ for all $ \theta $ with positive prior support; further, assume $\inf_{m} \int_{\Theta} \pshat(\ddobsn | \theta) \pi(\theta) d \theta>0$ and $\sup _{\theta \in \Theta} p_S(\ddobsn | \theta)<$
	$\infty$. Then
	\begin{equation}\
	\lim _{m \rightarrow \infty} \pishat(\theta | \ddobsn)=\pis(\theta | \ddobsn).
	\end{equation}Furthermore, if $f: \Theta \rightarrow \mathbb{R}$ is a continuous function satisfying $\sup _{m} \int_{\Theta}|f(\theta)|^{1+\delta} \pis^{(m)}(\theta | \ddobsn) d \theta<\infty$
	for some $\delta>0$ then
	$$
	\lim _{m \rightarrow \infty} \int_{\Theta} f(\theta) \pishat(\theta | \ddobsn) d \theta=\int_{\Theta} f(\theta) \pis(\theta | \ddobsn) d \theta.
	$$
\end{theorem}
\begin{proof}
	The first part follows from the fact that the numerator of
	$$
	\pishat(\theta | \ddobsn)=\frac{\pshat(\ddobsn | \theta) \pi(\theta)}{\int_{\Theta} \pshat(\ddobsn | \theta) \pi(\theta) d \theta}
	$$
	converges pointwise and the denominator is positive and converges by the bounded convergence theorem.

	For the second part, if for each $m \in \mathbb{N}, \theta_{m}$ is distributed according to $	\pishat(\cdot | \ddobsn)$ and $\theta$ is distributed according to $\pis(\cdot | \ddobsn)$ then $\theta_{m}$ converges to $\theta$ in distribution as $m \rightarrow \infty$ by Scheffé's lemma \citep{scheffe1947useful}. Since $f$ is continuous, $f\left({\theta}_{m}\right)$ converges in distribution to $f(\theta)$ as $n \rightarrow \infty$ by the continuous mapping theorem and we conclude by application of Remark~\ref{lemma:2} and Lemma~\ref{lemma:1}. \end{proof}

The following gives a convenient way to ensure $ \pshat(\ddobsn| \theta) \to p_S(\ddobsn|\theta) $:
\begin{theorem}[Generalizes Result 2 in \cite{drovandi2015bayesian}]\label{Th:Res2}
	Assume that $\exp\{ - w \sum_{i=1}^n \hat S(\Ddsimmtheta, \dobs_i)\}$ converges in probability to $p_S(\ddobsn | \theta)$ as $m \rightarrow \infty$. If
	$$
	\sup _{m} {\E}\left[\left|\exp\{ - w \sum_{i=1}^n \hat S(\Ddsimmtheta, \dobs_i)\}\right|^{1+\delta}\right]<\infty
	$$
	for some $\delta>0$ then $\pshat(\ddobsn | \theta) \rightarrow p_S(\ddobsn | \theta)$ as $m \rightarrow$
	$\infty$.\end{theorem}
\begin{proof}
	The proof follows by applying Remark~\ref{lemma:2} and Lemma \ref{lemma:1}.
\end{proof}
We are finally ready to prove Theorem~\ref{Th:auxiliary_lik_complete}:
\begin{proof}[Proof of Theorem~\ref{Th:auxiliary_lik_complete}]
First, notice how the convergence in probability of $ \hat S(\Ddsimmtheta, \dobs_i) $ to $ S(P_\theta, \dobs_i) $ (assumption 1 in Theorem~\ref{Th:auxiliary_lik_complete}) and the continuity of the exponential function imply convergence in probability of  $\exp\{ - w \sum_i \hat S(\Ddsimmtheta, \dobs_i)\}$ to $p_S(\ddobsn | \theta)$. That, together with assumption 2 in Theorem~\ref{Th:auxiliary_lik_complete}, satisfy the requirements of Theorem~\ref{Th:Res2}. 
With the latter and assumption 3 in Theorem~\ref{Th:auxiliary_lik_complete}, Theorem~\ref{Th:Res1} holds, which yields the result.
\end{proof}

\subsection{Proof and more details on Theorem \ref{Th:asymptotic_short}}\label{app:normal_asymptotic}
\subsubsection{Complete statement of Theorem \ref{Th:asymptotic_short}}\label{app:normal_asymptotic_complete}

    We proceed here with stating the more precise version of the result provided in Sec.~\ref{sec:as_norm}.
	Specifically, we show that the SR posterior satisfies (under some conditions) a Bernstein-von Mises theorem ensuring asymptotic normality. Without loss of generality, we fix here $ w=1 $ (other values can be absorbed in the definition of $ S $). The proof relies on the following assumptions:

	\begin{enumerate}[label=\textbf{A\arabic*}]

		\item \label{ass:thetastar} The expected scoring rule $ S(P_\theta, P_0) $ is finite for all $ \theta \in \Theta $; further, it has a unique minimizer:
		$$ \theta^\star =\argmin_{\theta \in \Theta} S(P_\theta, P_0)=\argmin_{\theta \in \Theta} D(P_\theta, P_0). $$
		Additionally, $  H_\star:=\nabla_\theta^2 S(P_\theta, P_0)\big|_{\theta=\thetastar} $ is positive definite.

		\item \label{ass:differentiable} Let us denote $ S'''(P_\theta, \dobs) _{jkl} = \frac{\partial^3}{\partial \theta_j\partial \theta_k\partial \theta_l}S(P_\theta, Y) $. There exists an open neighborhood $ E \subseteq \R^d$ of $ \thetastar $ whose closure $ \bar E \subseteq \Theta $ is such that, for all $ j,k,l\in\{1, \ldots, d\} $:
		\begin{itemize}
			\item $ \theta\to S'''(P_\theta, \dobs)_{jkl}$ is continuous in $ E $ and exists in $ \bar E $ for any fixed $ \dobs \in \X $,
			\item $ \dobs \to S'''(P_\theta, \dobs)_{jkl}$ is measurable for any fixed $ \theta \in \bar E $,
			\item $ \E_{P_0} \sup_{\theta\in\bar E} \left|  S'''(P_\theta, \dobs)_{jkl} \right|< \infty$.
		\end{itemize}

		\item \label{ass:neighborhood} For $ E $ defined above, there exists a compact $ K \subseteq E $, with $ \thetastar $ in the interior of $ K $, such that:
		\begin{equation}\label{}
		P_0\left\{ \liminf_n\inf_{\theta\in\Theta\backslash K} \frac{1}{n} \sum_{i=1}^n S(P_\theta, \Dobs_i) > S(P_{\thetastar}, P_0)  \right\}=1.
		\end{equation}

		\item \label{ass:prior} The prior has a density $\pi(\theta)$ with respect to Lebesgue measure; $ \pi(\theta)  $ is continuous and positive at $\thetastar$.
	\end{enumerate}

	Assumption~\ref{ass:neighborhood} is a regularity condition which can be replaced
	with clearer (but less general) Assumptions; see Appendix~\ref{app:normal_asymptotic_proof}.
	In Assumption \ref{ass:thetastar}, $ H_\star $ generalizes the standard Fisher information, which can be obtained by setting $ S(P_\theta, \dobs) = -\log p(\dobs|\theta) $. Additionally, uniqueness of $ \thetastar $ is obtained by strictly proper $ S $ and a well-specified model (in which case observations were generated from $ P_{\thetastar} $). If the model class is misspecified, a strictly proper $ S $ does not guarantee a unique minimizer (as in fact there may be pathological cases where multiple minimizers exist).

	\begin{theorem}\label{Th:normal_asymptotic}
		Let Assumptions \ref{ass:thetastar} to \ref{ass:prior} be true. Then, there is a sequence $  \hatthetan\left(\Ddobsn \right) $ which converges almost surely to $ \thetastar $ as $ n\to\infty $. Denote now by $ \pis^{*}\left(\cdot | \Ddobs_{\mathbf{n}}\right) $ the density of $ \sqrt{n}\left(\theta-\hatthetan\left(\Ddobs_{\mathbf{n}}\right)\right) $ when $ \theta\sim\pi_S\left(\cdot | \Ddobsn\right) $. Then as $n \rightarrow \infty$, with probability 1 over $ \Ddobsn $:
		$$
		\int_{\mathbb{R}^p}\left|\pis^{*}\left(s | \Ddobs_{\mathbf{n}}\right)-\mathcal{N}\left(s|0,H_\star^{-1}\right)\right| d s  \to 0,
		$$
		where $ \mathcal{N}(\cdot|0, \Sigma) $ denotes the density of a multivariate normal distribution with zero mean vector and covariance matrix $ \Sigma $.
	\end{theorem}

    In Appendix~\ref{app:normal_asymptotic_discussion}, we discuss assumptions and compare Theorem~\ref{Th:normal_asymptotic} with related results. Next, in Appendix~\ref{app:normal_asymptotic_proof}, we prove the Theorem.

\subsubsection{Discussion and comparison with related results}\label{app:normal_asymptotic_discussion}

\paragraph{Discussion on assumptions}

The uniqueness of the minimizer of the expected scoring rule $ \thetastar $ (in Assumption \ref{ass:thetastar}) is satisfied in a well specified setup if $ S $ is a strictly proper scoring rule (in which case $ P_{\thetastar}=P_0 $). If the model class is not well specified, a strictly proper $ S $ does not guarantee the minimizer to be unique (as in fact there may be pathological cases where multiple minimizers exist).

Additionally, it may be the case that, for a specific $ P_0 $ and misspecified model class $ P_\theta $, the minimizer of $ S(P_\theta, P_0) $ is unique even if $ S $ is not strictly proper; in fact, in general, being not strictly proper means that there exist at least one pair of values $ 	\theta^{(1)}, \theta^{(2)}$ for which $ S(P_{\theta^{(1)}}, P_{\theta^{(2)}}) = S(P_{\theta^{(1)}}, P_{\theta^{(1)}}) $, but it may be that the $ \argmin_{\theta \in \Theta} S(P_\theta, P_0)$ is unique for that specific choice of $ P_0 $, as the minimizer is in a region of the parameter space for which there are no other parameter values which lead to the same value of the scoring rule.

Our proof below builds on Theorem 5 in \cite{miller2019asymptotic}; to do so, we require regularity conditions on the third order derivatives of the SR (in Assumptions~\ref{ass:differentiable} or, alternatively, \ref{ass:differentiable_bis} below). It may be possible however to relax these assumptions to assuming $ \theta \to S(P_\theta, \dobs) $ can be locally written as a quadratic function of $ \theta $, with bounded coefficient for the third order term; this is usually called a Locally Asymptotically Normal (LAN) condition. With such, it would be possible to apply Theorem 4 in \cite{miller2019asymptotic} (more general than Theorem 5) to show our result.

\paragraph{Related results}
Appendix A in \cite{loaiza2019focused} provides a result which holds with non-i.i.d. (independent and identically distributed) data, with a generalized posterior based on scoring rules with a similar formulation to ours. Additionally, they replace our assumptions on differentiability (which ensure the existence of the Taylor series expansion in the proof below) with assuming the difference of the cumulative scoring rules have a LAN form. Finally, they only show convergence in probability.

Another related result can be found in \cite{matsubara2021robust}, which studies a generalized posterior based on kernelized Stein Discrepancy; similarly to us, they build on \cite{miller2019asymptotic}, and provide almost sure convergence. However, they exploit Theorem 4 in \cite{miller2019asymptotic}, while we rely on Theorem 5. In \cite{matsubara2021robust}, third order differentiability conditions are assumed, analogously to our Assumption \ref{ass:differentiable}. The remaining assumptions in \cite{matsubara2021robust} are similar to ours, including prior continuity and uniqueness of the minimizer $ \thetastar $.

Finally, we remark that, if multiple minimizers of $ S(P_\theta, P_0) $ exist (in finite number), it may be possible to obtain an asymptotic \textit{fractional} normality result, which ensures the SR posterior converges to a mixture of normal distributions centered in the different minimizers; see for instance \citep{frazier2021synthetic} for an example of such results in the setting of BSL. We leave this for future work.

\subsubsection{Alternative statements and proof}\label{app:normal_asymptotic_proof}

First, let us reproduce Theorem 5 in \cite{miller2019asymptotic}, on which our proof is based, for ease of reference. Here, convergence and boundedness for vectors $ v \in \R^p $, matrices $ M\in \R^{p\times p } $ and tensors $ T \in \R^{p\times p \times p}$ are defined with respect to Euclidean-Frobenius norms, that is: $ |v| = \left(\sum_j v_j^2\right)^{1/2} $, $  \|M\| = \left(\sum_{jk} M_{jk}^2\right)^{1/2}  $ and $   \|T\| = \left(\sum_{jkl} T_{jkl}^2\right)^{1/2}  $.

\begin{theorem}[Theorem 5 in \cite{miller2019asymptotic}]\label{Th:miller}
	Let $ \Theta \subseteq \R^p $. Let $ E \subseteq \Theta $ be open (in $ \R^p $) and bounded. Fix $ \thetastar \in E $ and let $ \pi : \Theta\to \R $ be a probability density with respect to Lebesgue measure. Consider the following family of distributions:
	\begin{equation}\label{}
	\pi_n(\theta) = \frac{\pi(\theta)\exp(-nf_n(\theta))}{\int_\Theta \pi(\theta)\exp(-nf_n(\theta))},
	\end{equation}
	where $ f_n:\Theta\to\R $ is a family of functions.
	Under the following conditions:
	\begin{enumerate}[label=\textbf{C\arabic*}]
		\item \label{cond:prior} $ \pi $ is continuous at $ \thetastar $ and $ \pi(\thetastar)>0 $,
		\item \label{cond:third_der} $ f_n $ have continuous third derivatives in $ E $,
		\item \label{cond:pointwise_convergence} $ f_n\to f $ pointwise for some $ f:\Theta\to\R $,
		\item \label{cond:posdef} $ f''(\thetastar) $ is positive definite,
		\item \label{cond:unif_bounded} $ f_n''' $ is uniformly bounded in $ E $,
		\item \label{cond:6} Either one of the following holds:
		\begin{enumerate}
			\item \label{cond:K} for some compact $ K \subseteq E $, with $ \thetastar $ in the interior of $ K $, $ f(\theta)> f(\thetastar)\ \forall \theta \in K \backslash \{\thetastar\}$ and $ \liminf_n\inf_{\theta\in\Theta\backslash K} f_n(\theta)>f(\thetastar) $, or
			\item \label{cond:convex} each $ f_n $ is convex and $ f'(\thetastar)=0 $;
		\end{enumerate}
	\end{enumerate}
	then, there is a sequence $ \theta_n\to\thetastar $ such that $ f'_n(\theta_n) =0$ for all $ n $ sufficiently large, $f_n(\theta_n)\to f(\thetastar)$ and, letting $ q_n $ be the density of $ \sqrt{n}(\theta-\theta_n) $ when $ \theta\sim \pi_n $:
	\begin{equation}\label{key}
	\int \left| q_n(s) - \mathcal{N}\left(s|0, (f''(\thetastar))^{-1}\right) \right| ds \to 0 \text{ as }  n\to\infty,
	\end{equation}
	that is, $ q_n $ converges to $ \mathcal{N}\left(0, (f''(\thetastar))^{-1}\right) $ in total variation. Additionally, \ref{cond:convex} implies \ref{cond:K} under the other conditions.
\end{theorem}

Notice that Theorem~\ref{Th:miller} considers deterministic $ f_n $ and $ f $. In order to prove our result, therefore, we will verify the different conditions hold almost surely, which implies almost sure convergence.

Besides the assumptions considered in the main text (i.e. \ref{ass:thetastar}-\ref{ass:prior}), it is possible to prove the asymptotic normality result in Theorem~\ref{Th:normal_asymptotic} under alternative sets of assumptions. For this reason, we introduce the following:

\begin{enumerate}[label=\textbf{A\arabic*bis}]
	\setcounter{enumi}{1}
	\item \label{ass:differentiable_bis} The parameter space $ \Theta $ is open, convex, and bounded; the function $ \theta\to S(P_\theta, \dobs) $, for any fixed $ \dobs\in\X $, can be extended to the closure $ \bar \Theta $. Let us denote $ S'''(P_\theta, \dobs) _{jkl} = \frac{\partial^3}{\partial \theta_j\partial \theta_k\partial \theta_l}S(P_\theta, Y) $. For all $ j,k,l\in\{1, \ldots, d\} $:
	\begin{itemize}
		\item $ \theta\to S'''(P_\theta, \dobs)_{jkl}$ is continuous in $ \Theta $ and exists in $ \bar \Theta $ for any fixed $ \dobs \in \X $,
		\item $ \dobs \to S'''(P_\theta, \dobs)_{jkl}$ is measurable for any fixed $ \theta \in \bar \Theta $,
		\item $ \E_{P_0} \sup_{\theta\in\bar \Theta} \left|  S'''(P_\theta, \dobs)_{jkl} \right|< \infty$.
	\end{itemize}

	\item \label{ass:convex} For each $\dobs \in \X$, the function:
	$ \theta \to S(P_\theta, \dobs)  $
	is convex.
\end{enumerate}

The following extended form of Theorem~\ref{Th:normal_asymptotic} includes the formulation in the main text as well as two alternative sets of assumptions.

\begin{customthm}{\ref{Th:normal_asymptotic} - extended version}
	\textit{	Under either one of the following sets of assumptions:
		\begin{enumerate}
			\item \label{item:set1} \ref{ass:thetastar}, \ref{ass:differentiable}, \ref{ass:neighborhood}, \ref{ass:prior} (the set originally used in the main text),
			\item \label{item:set3} \ref{ass:thetastar}, \ref{ass:differentiable}, \ref{ass:convex}, \ref{ass:prior},
			\item \label{item:set2} \ref{ass:thetastar}, \ref{ass:differentiable_bis}, \ref{ass:prior},
		\end{enumerate}
		the statement of Theorem~\ref{Th:normal_asymptotic} in the main text holds.
}\end{customthm}

We next move to proving our result.

Assumption~\ref{ass:differentiable} is used in the original set of assumptions to ensure the second part of Condition~\ref{cond:K} holds almost surely. Under set of assumptions~\ref{item:set3}, convexity of the scoring rules (Assumption~\ref{ass:convex}) is used to show Condition~\ref{cond:convex}; alternatively, with set of assumptions~\ref{item:set2}, the constraints on $ \Theta $ are used to imply the second part of Condition~\ref{cond:K} with probability 1 using Theorem 7 in \cite{miller2019asymptotic}. In both cases, Assumption~\ref{ass:differentiable} is not explicitly needed anymore -- as in fact it is implied by the remaining assumptions. However, we are unable to remove Assumption~\ref{ass:differentiable} under no constraints on $ \Theta $ or the convexity of $ \theta\to S(P_\theta, y) $.

We now give our proof:

\begin{proof}[Proof of Theorem~\ref{Th:normal_asymptotic} - extended version]
	In order to obtain our result, we identify $$ f_n\left(\theta \right) =  \frac{1}{n} \sum_{i=1}^n S(P_\theta, \Dobs_i) , \qquad f(\theta) = S(P_\theta, P_0); $$ this implies that $ f_n $ is now a random quantity: as such, we show the conditions for Theorem~\ref{Th:miller} hold almost surely over the stochasticity induced by $ \Ddobsn $.

	All three sets of assumptions include Assumptions \ref{ass:thetastar} and \ref{ass:prior}; therefore, under all three sets of assumptions:
	\begin{itemize}
		\item \ref{cond:prior} corresponds to our Assumption~\ref{ass:prior},
		\item \ref{cond:pointwise_convergence} holds almost surely thanks to the strong law of large numbers, as $ S(P_\theta, P_0) $ is finite $ \forall \theta \in \Theta $ by Assumption~\ref{ass:thetastar},
		\item \ref{cond:posdef} is implied by Assumption~\ref{ass:thetastar}.
	\end{itemize}

	Therefore, we are left with establishing \ref{cond:third_der}, \ref{cond:unif_bounded} and \ref{cond:6} separately for the different sets of assumptions.

	\paragraph{Set of assumptions \ref{item:set1} (used in Appendix~\ref{app:normal_asymptotic}):}
	\begin{itemize}

		\item \ref{cond:third_der} is implied by our Assumption~\ref{ass:differentiable} to hold with probability $ 1 $ for all $ n $.

		\item In order to show \ref{cond:unif_bounded}, we proceed in similar manner as in Theorem 13 in \cite{miller2019asymptotic}. For any $ j,k,l\in\{1,\ldots,d\} $, Assumption~\ref{ass:differentiable} implies that, with probability 1, $ f_n'''(\theta)_{jkl} = \frac{1}{n} \sum_{i=1}^n  S'''(P_\theta, \Dobs_i)_{jkl} $ is uniformly bounded on $ \bar E $ by the uniform law of large number (Theorem 1.3.3 in \citealt{ghosh2003bayesian}). Letting $ C_{jkl}(\Dobs_1, \Dobs_2, \ldots) $ be such a uniform bound for each $ j,k,l $, we have that with probability $ 1 $, for all $ n\in \mathbb{N} $, $ \theta\in\bar E $, $  \| f_n'''(\theta) \|^2 =\sum_{jkl} (f_n'''(\theta)_{jkl})^2\le \sum_{jkl} C_{jkl}(\Dobs_1, \Dobs_2, \ldots) < \infty$. Thus, $ f_n'''(\theta) $ is almost surely uniformly bounded on $ \bar E $, and hence on $ E $

		\item The first part of \ref{cond:K} is implied by Assumption~\ref{ass:thetastar}, while the second part holds almost surely by Assumption~\ref{ass:neighborhood}.
	\end{itemize}
	\paragraph{Set of assumptions \ref{item:set3}:}
	The only difference here is that we replace Assumption~\ref{ass:neighborhood} with the stronger convexity Assumption~\ref{ass:convex}; therefore, \ref{cond:third_der} and \ref{cond:unif_bounded} are shown in the same way as with set of assumptions~\ref{item:set1}.

	Next, consider \ref{cond:convex}: the first part is implied to hold with probability 1 for all $ n $ by Assumption~\ref{ass:convex}, as the sum of convex functions is convex. The second part is instead consequence of $ \thetastar $ being a stationary point of $ f $ due to  Assumption~\ref{ass:thetastar}, and of $ f'(\thetastar) $ existing due to \ref{cond:posdef}.
	\paragraph{Set of assumptions \ref{item:set2}:}
	Under these assumptions, we fix $ E=\Theta $ in the statement of Theorem~\ref{Th:miller}, as we consider $ \Theta $ to be open and bounded. With that, we can exploit Assumption~\ref{ass:differentiable_bis} and follow the same steps as with set of assumptions~\ref{item:set1} to show that, over $ \Theta $, \ref{cond:third_der} and \ref{cond:unif_bounded} hold with probability 1.

	The first part of \ref{cond:K} is implied by Assumption~\ref{ass:thetastar}, for any choice of $ K $; it now remains to show the second part. First, Theorem 7 in \cite{miller2019asymptotic} implies that $ f_n \to f $ uniformly almost surely, as in fact $ f_n $ have continuous third derivatives by \ref{cond:third_der}, $ f_n''' $ is uniformly bounded with probability 1 by \ref{cond:unif_bounded}, and $ f_n\to f $ with probability 1 due to \ref{cond:pointwise_convergence} holding with probability 1.

	Therefore, with probability 1:
	\begin{equation}\label{}
	\liminf_n\inf_{\theta\in\Theta\backslash K} f_n(\theta) = \inf_{\theta\in\Theta\backslash K} f(\theta) > \inf_{\theta\in\Theta} f(\theta) = f(\thetastar),
	\end{equation}
	where the first equality is due to uniform convergence allowing to ``swap'' the infimum and the limit.
\end{proof}

\subsection{Proof of Theorem~\ref{Th:posterior_consistency}}\label{app:posterior_consistency}

First, we prove a finite sample generalization bound which is valid for the generalized Bayes posterior with a generic loss, assuming a concentration property and prior mass condition. Next, we will use this Lemma to prove Theorem ~\ref{Th:posterior_consistency} reported in the main body of the paper (in Section~\ref{sec:posterior_consistency}), by first proving concentration results for Kernel and energy scores.

We remark that our Theorem~\ref{Th:posterior_consistency} is similar to Theorem 1 in \cite{matsubara2021robust} for the kernelized Stein Discrepancy (KSD) posterior, but provides a tighter probability bound. As the kernel used in KSD is unbounded, in fact, \cite{matsubara2021robust} had to rely on weaker results with respect to the ones used to prove Theorem~\ref{Th:posterior_consistency}. With a similar approach, a result for unbounded $ k $ or $ \X $ may be obtained in our case; we leave this for future exploration.

\subsubsection{Lemma for generalized Bayes posterior with generic loss}

In this Subsection, we consider the following generalized Bayes posterior:
\begin{equation}
\label{Eq:general_Bayes_post}
\pi_L(\theta|\ddobsn ) \propto\pi(\theta)\exp\left\{-w n L(\theta, \ddobsn) \right\},
\end{equation}
where $\ddobsn = \{\dobs_i\}_{i=1}^n$ denote the observations, $ \pi $ is the prior and $ L(\theta, \ddobsn)  $ is a generic loss function (which does not need to be additive in $ \dobs_i $). Here, the SR posterior for the scoring rule $ S $ corresponds to choosing: $$ L\left(\theta , \ddobsn\right) =  \frac{1}{n} \sum_{i=1}^n S(P_\theta, \dobs_i).$$

First, we state a result concerning this form of the posterior which we will use later (taken from \citealt{knoblauch2019generalized}), and reproduce here the proof for convenience:

\begin{lemma}[Theorem 1 in \cite{knoblauch2019generalized}]\label{lemma:knoblauch}
	Provided that $  \int_{\Theta} \pi(\theta)\exp\left\{-w n L(\theta, \ddobsn)\right\} d\theta <\infty$, $ \pi_L(\cdot|\ddobsn ) $ in Eq.~\eqref{Eq:general_Bayes_post} can be written as the solution to a variational problem:
	\begin{equation}\label{Eq:knoblauch}
	\pi_L(\cdot|\ddobsn )=\underset{\rho \in \mathcal{P}(\Theta)}{\arg \min }\left\{w n \mathbb{E}_{\theta \sim \rho}\left[L(\theta, \ddobsn)\right]+\mathrm{KL}(\rho \| \pi)\right\},
	\end{equation}
	where $ \mathcal{P}(\Theta) $ denotes the set of distributions over $ \Theta $, and $ \mathrm{KL} $ denotes the KL divergence.
\end{lemma}
\begin{proof}
	We follow here (but adapt to our notation) the proof given in \cite{knoblauch2019generalized}, which in turn is based on the one for the related result contained in \cite{bissiri2016general}.

	Notice that the minimizer of the objective in Eq.~\eqref{Eq:knoblauch} can be written as:
	\begin{equation}\label{Eq:knoblauch_2}
	\begin{aligned}
	\pi^\star(\cdot|\ddobsn )&=\underset{\rho \in \mathcal{P}(\Theta)}{\arg \min }\left\{ \int_{\Theta}\left[\log \left(\exp \left\{w n L(\theta, \ddobsn)\right\}\right) + \log \left(\frac{\rho(\theta)}{\pi(\theta)}\right)\right]\rho(\theta) d\theta \right\}\\
	&=\underset{\rho \in \mathcal{P}(\Theta)}{\arg \min }\left\{\int_{\Theta}\left[\log \left(\frac{\rho(\theta)}{\pi(\theta) \exp \left\{ - w nL(\theta, \ddobsn)\right\}}  \right)\right]\rho(\theta) d\theta \right\}.
	\end{aligned}
	\end{equation}
	As we are only interested in the minimizer $ \pi^\star(\cdot|\ddobsn) $ (and not in the value of the objective), it holds that, for any constant $ Z>0 $:
	\begin{equation}\label{Eq:knoblauch_3}
	\begin{aligned}
	\pi^\star(\cdot|\ddobsn )&=\underset{\rho \in \mathcal{P}(\Theta)}{\arg \min }\left\{\int_{\Theta}\left[\log \left(\frac{\rho(\theta)}{\pi(\theta) \exp \left\{ - w nL(\theta, \ddobsn)\right\}Z^{-1}}  \right)\right]\rho(\theta) d\theta - \log Z \right\}\\
	&= \underset{\rho \in \mathcal{P}(\Theta)}{\arg \min }\left\{ \mathrm{KL}\left(\rho(\theta) \| \pi(\theta) \exp \left\{ - w nL(\theta, \ddobsn)\right\}Z^{-1} \right)  \right\}.
	\end{aligned}
	\end{equation}
	Now, we can set $ Z=\int_{\Theta} \pi(\theta)\exp\left\{-w n L(\theta, \ddobsn)\right\} d\theta  $ (which is finite by assumption) and notice that we get:
	\begin{equation}\label{}
	\pi^\star(\cdot|\ddobsn )=\underset{\rho \in \mathcal{P}(\Theta)}{\arg \min }\left\{\mathrm{KL}\left(\rho|| \pi_L(\cdot|\ddobsn ) \right) \right\},
	\end{equation}
	which yields $ \pi^\star(\cdot|\ddobsn ) = \pi_L(\cdot|\ddobsn )  $ as the KL is minimized uniquely if the two arguments are the same.
\end{proof}

Next, we prove a finite sample (as it holds for fixed number of samples $ n $) generalization bound. Our statement and proof generalize Lemma 8 in \cite{matsubara2021robust} (as we consider a generic loss function $ L(\theta, \ddobsn) $, while they consider the kernelized Stein Discrepancy only).

In order to do this, let $ J $ be a function of the parameter $ \theta $, with $ J(\theta) $ representing some loss (of which we will assume $ L(\theta, \ddobsn) $ is a finite sample estimate; the meaning of $ J $ will be made clearer in the following and when applying this result to the SR posterior).

We will assume the following \textit{prior mass condition}, which is more generic with respect to the one considered in the main body of this manuscript (Assumption~\ref{ass:prior_mass}):
\begin{enumerate}[label=\textbf{A\arabic*bis}]
	\setcounter{enumi}{4}
	\item \label{ass:prior_mass_bis} Denote $ \thetastar \in \argmin_{\theta \in \Theta} J(\theta)$, which is supposed to be non-empty. The prior has a density $ \pi(\theta) $ (with respect to Lebesgue measure) which satisfies
	\begin{equation}\label{}
	\int_{B_{n}\left(\alpha_{1}\right)} \pi(\theta) \mathrm{d} \theta \geq e^{-\alpha_{2} \sqrt{n}}
	\end{equation}
	for some constants $ \alpha_1, \alpha_2>0 $, where we define the sets $$ B_{n}\left(\alpha_{1}\right):=\left\{\theta \in \Theta:\left|J\left({\theta}\right)-J\left(\thetastar\right)\right| \leq \alpha_{1} / \sqrt{n}\right\}. $$
\end{enumerate}
Assumption~\ref{ass:prior_mass_bis} constrains the minimum amount of prior mass which needs to be given to $ J $-balls with decreasing size, and is in general quite a weak condition (similar assumptions are taken in \cite{cherief2020mmd, matsubara2021robust}).

Next, we state our result, which as mentioned above generalizes Lemma 8 in \cite{matsubara2021robust}:
\begin{lemma}\label{lemma:Matsubara_8}
	Consider the generalized posterior $ \pi_L(\theta|\ddobsn ) $ defined in Eq.~\eqref{Eq:general_Bayes_post}, and assume that:
	\begin{itemize}
		\item (concentration) for all $ \delta \in(0,1] $:
		\begin{equation}\label{Eq:concentration_assumption}
		P_0 \left\{  \left| L(\theta, \Ddobsn) - J(\theta) \right| \le   \epsilon_{n}(\delta)  \right\} \ge 1-\delta,
		\end{equation}
		where $ \epsilon_{n}(\delta) \ge 0 $ is an approximation error term;
		\item $ J(\thetastar) = \min_{\theta \in \Theta} J(\theta) $ is finite;
		\item Assumption \ref{ass:prior_mass_bis} holds.
	\end{itemize}
	Then, for all $ \delta \in(0,1] $, with probability at least $ 1-\delta $:
	\begin{equation}
	\int_{\Theta} J(\theta) \pi_L(\theta|\Ddobsn ) \mathrm{d} \theta \leq J(\thetastar) + \frac{\alpha_{1} + \alpha_{2}/w}{\sqrt{n}} + 2\epsilon_{n}(\delta),
	\end{equation}
	where the probability is taken with respect to realisations of the dataset $ \Ddobsn = \{\Dobs_i\}_{i=1}^n, \Dobs_i \overset{iid}{\sim} P_0 \text{ for } i=1, \ldots, n$; this also implies the following statement:
	\begin{equation}
	{P}_0\left(\left|\int_{\Theta} J(\theta) \pi_L(\theta|\Ddobsn ) \mathrm{d} \theta-J(\thetastar)\right| \geq  \frac{\alpha_{1}+\alpha_{2}/w}{\sqrt{n}} + 2 \epsilon_{n}(\delta)\right) \le \delta.
	\end{equation}
\end{lemma}

This result ensures that, with high probability, the expectation over the posterior of $ J(\theta)  $ is close to the minimum $ J(\thetastar) $, provided that the distribution of $ L(\theta, \Ddobsn) $ (where $ \Ddobsn\sim P_0^n $ is a random variable) satisfies a concentration bound, which constrains how far $ L(\theta, \Ddobsn) $ is distributed from the loss function $ J(\theta) $. Notice that this result does not require the minimizer of $ J $ to be unique.

Typically the approximation error term $ \epsilon_{n}(\delta) $ is such that $ \epsilon_{n}(\delta) \xrightarrow{\delta\to 0} + \infty $ and $ \epsilon_{n}(\delta) \xrightarrow{n\to\infty} 0  $. If the second limit is verified, the posterior concentrates, for large $ n $, on the values of $ \theta $ which minimize $ J $. In practical cases (as for instance for the SR posterior), it is common to have $ J(\theta) = D(\theta, P_0) $, i.e., corresponding to a loss function relating $ \theta $ with the data generating process $ P_0 $.

We now prove the result.
\begin{proof}[Proof of Lemma~\ref{lemma:Matsubara_8}]

	Due to the absolute value in Eq.~\eqref{Eq:concentration_assumption}, the following two inequalities hold simultaneously with probability (w.p.) at least $ 1-\delta $:
	\begin{equation}\label{Eq:proof_lemma8_ineq1}
	J(\theta) \leq L(\theta, \Ddobsn)+ \epsilon_{n}(\delta),
	\end{equation}
	\begin{equation}\label{Eq:proof_lemma8_ineq2}
	L(\theta, \Ddobsn) \leq J(\theta)+ \epsilon_{n}(\delta).
	\end{equation}
	Taking expectation with respect to the generalized posterior on both sides of Eq.~\eqref{Eq:proof_lemma8_ineq1} yields, \\w.p.~$\ge~1-~\delta $:
	\begin{equation}
	\int_{\Theta} J(\theta) \pi_L(\theta|\Ddobsn ) \mathrm{d} \theta \leq \int_{\Theta} L(\theta, \Ddobsn) \pi_L(\theta|\Ddobsn ) \mathrm{d} \theta+ \epsilon_{n}(\delta).
	\end{equation}
	We now want to apply the identity in Eq.~\eqref{Eq:knoblauch}; therefore, we add $(wn)^{-1} \mathrm{KL}\left(\pi_L(\cdot|\Ddobsn )\| \pi\right) \geq 0$ in the right hand side such that, w.p.~$ \ge 1-\delta $:
	\begin{equation}
	\int_{\Theta} J(\theta) \pi_L(\theta|\Ddobsn ) \mathrm{d} \theta \leq \frac{1}{wn}\left\{\int_{\Theta} wn  L(\theta, \Ddobsn) \pi_L(\theta|\Ddobsn ) \mathrm{d} \theta+\mathrm{KL}\left(\pi_L(\cdot|\Ddobsn )\| \pi\right)\right\}
	+ \epsilon_{n}(\delta).
	\end{equation}
	Now by Eq.~\eqref{Eq:knoblauch}:
	\begin{equation}\label{Eq:proof_lemma8_variational}
	\begin{aligned}
	\int_{\Theta} J(\theta) \pi_L(\theta|\Ddobsn ) \mathrm{d} \theta & \leq \frac{1}{wn} \inf _{\rho \in \mathcal{P}(\Theta)}\left\{\int_{\Theta} w n L(\theta, \Ddobsn) \rho(\theta) \mathrm{d} \theta+\mathrm{KL}(\rho \| \pi)\right\}+ \epsilon_{n}(\delta)\\
	&=\inf _{\rho \in \mathcal{P}(\Theta)}\left\{\int_{\Theta} L(\theta, \Ddobsn) \rho(\theta) \mathrm{d} \theta+\frac{1}{wn} \mathrm{KL}(\rho \| \pi)\right\}+ \epsilon_{n}(\delta),
	\end{aligned}
	\end{equation}
	where $ \mathcal{P}(\Theta) $ denotes the space of probability distributions over $ \Theta $. Putting now Eq.~\eqref{Eq:proof_lemma8_ineq2} in Eq.~\eqref{Eq:proof_lemma8_variational} we have, w.p. $ \ge 1-\delta $:
	\begin{equation}\label{Eq:proof_lemma8_variational_2}
	\int_{\Theta} J(\theta) \pi_L(\theta|\Ddobsn ) \mathrm{d} \theta \le \inf _{\rho \in \mathcal{P}(\Theta)}\left\{\int_{\Theta} J(\theta) \rho(\theta) \mathrm{d} \theta+\frac{1}{wn} \mathrm{KL}(\rho \| \pi)\right\} +2 \epsilon_{n}(\delta),
	\end{equation}
	and using the trivial bound $J(\theta) \leq J(\thetastar)+\left|J(\theta)-J(\thetastar)\right|$ we get:
	\begin{equation}\
	\int_{\Theta} J(\theta) \pi_L(\theta|\Ddobsn ) \mathrm{d} \theta \le J(\thetastar)+\inf _{\rho \in \mathcal{P}(\Theta)}\left\{\int_{\Theta}\left|J(\theta)-J(\thetastar)\right| \rho(\theta) \mathrm{d} \theta+\frac{1}{wn} \mathrm{KL}(\rho \| \pi)\right\}+ 2\epsilon_n(\delta).
	\end{equation}
	Finally, we upper bound the infimum term by exploiting the prior mass condition in Assumption \ref{ass:prior_mass_bis}. Specifically, letting $ \Pi(B_n) = \int_{B_{n}} \pi(\theta) d\theta $, we take $ \rho(\theta) = \pi(\theta) /\Pi(B_n) $ for $ \theta \in B_n $ and $ \rho(\theta)=0  $ otherwise. By Assumption \ref{ass:prior_mass_bis}, we have therefore $ \int_{B_{n}}\left|J(\theta)-J(\thetastar)\right| \rho(\theta) \mathrm{d} \theta \leq \alpha_{1} / \sqrt{n} $ and that $\operatorname{KL}(\rho \| \pi)=\int_{\Theta} \log (\rho(\theta) / \pi(\theta)) \rho(\theta) \mathrm{d} \theta=\int_{B_{n}}-\log \left(\Pi\left(B_{n}\right)\right) \pi(\theta) \mathrm{d} \theta / \Pi\left(B_{n}\right)=-\log \Pi\left(B_{n}\right) \leq \alpha_{2} \sqrt{n}$. Thus, we have:
	\begin{equation}
	\int_{\Theta} J(\theta) \pi_L(\theta|\Ddobsn ) \mathrm{d} \theta \leq J(\thetastar)
	+ \frac{\alpha_{1}+\alpha_{2}/w}{\sqrt{n}} + 2 \epsilon_{n}(\delta),
	\end{equation}
	as claimed in the first statement.

	In order to obtain the second statement, notice that:
	\begin{equation}\
	J(\theta) - J(\thetastar) \ge 0, \quad \forall \ \theta \in \Theta \implies \int_{\Theta} J(\theta) \pi_L(\theta|\Ddobsn ) \mathrm{d} \theta-J(\thetastar) \geq 0;
	\end{equation}
	thus:
	\begin{equation}
	{P}_0\left(\left|\int_{\Theta} J(\theta) \pi_L(\theta|\Ddobsn ) \mathrm{d} \theta-J(\thetastar)\right| \leq   \frac{\alpha_{1}+\alpha_{2}/w}{\sqrt{n}} + 2 \epsilon_{n}(\delta)\right) \geq 1-\delta;
	\end{equation}
	taking the complement yields the result.
\end{proof}

\subsubsection{Case of Kernel and energy score posteriors}

We now state and prove concentration results of the form in Eq.~\eqref{Eq:concentration_assumption} for the Kernel and energy scores. Here, we will assume $ \Dsim \independent \Dsim' $ To this regards, notice that the kernel SR posterior can be written as:
\begin{equation}\label{}
\begin{aligned}
\pik(\theta|\ddobsn) &\propto \pi(\theta) \exp\left\{ - w \sum_{i=1}^n \left[\E_{\Dsim, \Dsim'\sim P_\theta}k(\Dsim,\Dsim') - 2 \E_{\Dsim \sim P_\theta} k(\Dsim, \dobs_i)\right] \right\}\\
&\propto \pi(\theta) \exp\left\{ - w \sum_{i=1}^n \left[\E_{\Dsim, \Dsim'\sim P_\theta}k(\Dsim,\Dsim') - 2 \E_{\Dsim \sim P_\theta} k(\Dsim, \dobs_i)  + \frac{1}{n-1} \sum_{\substack{j=1\\j\neq i}}^n k(\dobs_{i}, \dobs_j) \right] \right\},
\end{aligned}
\end{equation}
as in fact the terms $ k(\dobs_{i}, \dobs_j) $ are independent of $ \theta $. From the second line in the above expression and the form of the generalized Bayes posterior with generic loss in Eq.~\eqref{Eq:general_Bayes_post}, we can identify:
\begin{equation}\label{Eq:empirical_loss_kerSR}
L(\theta, \ddobsn) =  \E_{\Dsim, \Dsim'\sim P_\theta}k(\Dsim,\Dsim') - \frac{2}{n} \sum_{i=1}^n \E_{\Dsim \sim P_\theta} k(\Dsim, \dobs_i) + \frac{1}{n(n-1)} \sum_{\substack{i,j=1\\i\neq j}}^n k(\dobs_{i}, \dobs_j) .
\end{equation}
Similarly, the energy score posterior can be obtained by identifying in Eq.~\eqref{Eq:general_Bayes_post}:
\begin{equation}\label{Eq:empirical_loss_enSR}
L(\theta, \ddobsn) =  \frac{2}{n} \sum_{i=1}^n \E_{\Dsim \sim P_\theta} ||\Dsim- \dobs_i||_2^\beta - \frac{1}{n(n-1)} \sum_{\substack{i,j=1\\i\neq j}}^n ||\dobs_{i}- \dobs_j||_2^\beta  - \E_{\Dsim, \Dsim'\sim P_\theta}||\Dsim-\Dsim'||_2^\beta;
\end{equation}
this can be obtained by simply setting $ k(x,y) = -\|x-y\|_2^\beta $ in Eq.~\eqref{Eq:empirical_loss_kerSR}, as the Kernel SR with that choice of kernel recovers the Energy SR.

For both SRs, $ L(\theta, \Ddobsn) $ is an unbiased estimator (with respect to $ \Dobs_i \sim P_0 $) of the associated divergences; in fact, considering $ \Dsim \independent\Dsim'\sim P_\theta $ and $ \Dobs\independent\Dobs' \sim P_0 $, the associated divergence for Kernel SR is the squared MMD (see Appendix~\ref{app:MMD}):
\begin{equation}\label{Eq:kernel_SR_divergence}
D_k(P_\theta, P_0) = \E k(\Dsim, \Dsim') + \E k(\Dobs, \Dobs')- 2 \E k(\Dsim,\Dobs),
\end{equation}
while, for the Energy SR, the associated divergence is the squared Energy Distance:
\begin{equation}\label{Eq:energy_SR_divergence}
\DEbeta(P_\theta, P_0) = 2 \E ||\Dsim-\Dobs||_2^\beta - \E ||\Dsim- \Dsim'||_2^\beta - \E|| \Dobs- \Dobs'||_2^\beta.
\end{equation}

In order to prove our concentration results, we will exploit the following Lemma:

\begin{lemma}[McDiarmid's inequality, \citealt{mcdiarmid1989method}]\label{lemma:mcdiarmid}
	Let $ g $ be a function of $ n $ variables $ \ddobsn = \{\dobs_i\}_{i=1}^n $, and let
	\begin{equation}\label{}
	\delta_{i} g(\ddobsn):=\sup _{z\in \X} g\left(\dobs_{1}, \ldots, \dobs_{i-1}, z, \dobs_{i+1}, \ldots, \dobs_{n} \right)-\inf _{z\in\X} g\left(\dobs_{1}, \ldots, \dobs_{i-1}, z, \dobs_{i+1}, \ldots, \dobs_{n}\right),
	\end{equation}
	and $ \left\|\delta_{i} g\right\|_{\infty}:=\sup _{\ddobsn\in\X^n}\left|\delta_{i} g(\ddobsn)\right| $. If $ \Dobs_1, \ldots, \Dobs_n $ are independent random variables:
	\begin{equation}\label{Eq:McDiarmids}
	{P}\left(g\left(\Dobs_{1}, \ldots, \Dobs_{n}\right)-\E g\left(\Dobs_{1}, \ldots, \Dobs_{n}\right) \geq \varepsilon\right) \leq e^{-2 \varepsilon^{2} / \sum_{i=1}^{n}\left\|\delta_{i} g\right\|_{\infty}^{2}}.
	\end{equation}
\end{lemma}

We are now ready to prove two concentration results of the form of Eq.~\eqref{Eq:concentration_assumption}. The first holds for the Kernel SR assuming a bounded kernel, while the latter holds for the Energy SR assuming a bounded $ \X $. Let us start with a simple equality stated in the following Lemma:

\begin{lemma}\label{lemma:equality_concentration}
	For $ L(\theta, \Ddobsn) $ defined in Eq.~\eqref{Eq:empirical_loss_kerSR} and $ D_k(P_\theta, P_0) $ defined in Eq.~\eqref{Eq:kernel_SR_divergence}, we have:
	\begin{equation}
	L(\theta, \Ddobsn) - D_k(P_\theta, P_0) = g(\Dobs_1, \Dobs_2, \ldots, \Dobs_n) - \E [g(\Dobs_1, \Dobs_2, \ldots, \Dobs_n)]
	\end{equation}
	for \begin{equation}	\label{Eq:g_kerSR}
	g(\Dobs_1, \Dobs_2, \ldots, \Dobs_n) = \frac{1}{n(n-1)} \sum_{\substack{i,j=1\\i\neq j}}^n k(\Dobs_{i}, \Dobs_j) - \frac{2}{n} \sum_{i=1}^n \E_{\Dsim \sim P_\theta} k(\Dsim, \Dobs_i).
	\end{equation}
	Similar expression holds for $ L(\theta, \Ddobsn) $ defined in Eq.~\eqref{Eq:empirical_loss_enSR} and $ \DEbeta(P_\theta, P_0) $ defined in Eq.~\eqref{Eq:energy_SR_divergence}, by setting $k(x,y) = -\|x-y\|_2^\beta $.
	
\end{lemma}
\begin{proof}
	First, notice that, for $ L(\theta, \Ddobsn) $ defined in Eq.~\eqref{Eq:empirical_loss_kerSR} and $ D_k(P_\theta, P_0) $ defined in Eq.~\eqref{Eq:kernel_SR_divergence}:
	\begin{equation}\
	\begin{aligned}
	L(\theta, \Ddobsn) - D_k(P_\theta, P_0) &= \cancel{\E_{\Dsim, \Dsim'\sim P_\theta}k(\Dsim,\Dsim')} - \frac{2}{n} \sum_{i=1}^n \E_{\Dsim \sim P_\theta} k(\Dsim, \Dobs_i) + \frac{1}{n(n-1)} \sum_{\substack{i,j=1\\i\neq j}}^n k(\Dobs_{i}, \Dobs_j) + \\&-  \cancel{\E_{\Dsim, \Dsim'\sim P_\theta} k(\Dsim, \Dsim')} - \E_{\Dobs,\Dobs'\sim P_0} [k(\Dobs, \Dobs')]+ 2 \E_{\Dsim \sim P_\theta, \Dobs\sim P_0} [k(\Dsim,\Dobs)] \\
	& =  \frac{1}{n(n-1)} \sum_{\substack{i,j=1\\i\neq j}}^n k(\Dobs_{i}, \Dobs_j) - \frac{2}{n} \sum_{i=1}^n \E_{\Dsim \sim P_\theta} k(\Dsim, \Dobs_i) +\\&- \left(\E_{\Dobs,\Dobs'\sim P_0} [k(\Dobs, \Dobs')] - 2 \E_{\Dsim \sim P_\theta, \Dobs\sim P_0} [k(\Dsim,\Dobs)]  \right)\\
	&=  g(\Dobs_1, \Dobs_2, \ldots, \Dobs_n) - \E [g(\Dobs_1, \Dobs_2, \ldots, \Dobs_n)],
	\end{aligned}
	\end{equation}
	where the expectation in the last line is with respect to $ \Dobs_i \sim P_0,\ i=1,\ldots, n $, and where we set $ g $ as in Eq.~\eqref{Eq:g_kerSR}.
\end{proof}

Now, we give the concentration result for the kernel SR:
\begin{lemma}\label{lemma:concentration_MMD2}
	Consider $ L(\theta, \ddobsn) $ defined in Eq.~\eqref{Eq:empirical_loss_kerSR} (corresponding to the loss function defining the kernel score posterior) and $ D_k(P_\theta, P_0) $ defined in Eq.~\eqref{Eq:kernel_SR_divergence}; if the kernel is such that $\sup_{x, y\in\mathcal{X}}| k(x, y) | \le \kappa <\infty $, we have:
	\begin{equation}\
	P_0\left(\left|	L(\theta, \Ddobsn) - D_k(P_\theta, P_0) \right| \le \sqrt{-\frac{32\kappa^2}{n}\log \frac{\delta}{2}}\right) \ge 1- \delta.
	\end{equation}
\end{lemma}
\begin{proof}
	First, we write:
	\begin{equation}\label{}
	L(\theta, \Ddobsn) - D_k(P_\theta, P_0) = g(\Dobs_1, \Dobs_2, \ldots, \Dobs_n) - \E [g(\Dobs_1, \Dobs_2, \ldots, \Dobs_n)],
	\end{equation}
	where $ g $ is defined in Eq.~\eqref{Eq:g_kerSR} in Lemma~\ref{lemma:equality_concentration}. Next, notice that:
	\begin{equation}\label{}
	\begin{aligned}
	P_0(\left|g(\Ddobsn) - \E [g(\Ddobsn)]\right| \ge \epsilon)&\le P_0(g(\Ddobsn) - \E [g(\Ddobsn)] \ge \epsilon) + P_0(g(\Ddobsn) - \E [g(\Ddobsn)] \le - \epsilon) \\
	&= P_0(g(\Ddobsn) - \E [g(\Ddobsn)] \ge \epsilon) + P_0(-g(\Ddobsn) - \E [-g(\Ddobsn)] \ge \epsilon)
	\end{aligned}
	\end{equation}
	by the union bound. We use now McDiarmid's inequality (Lemma~\ref{lemma:mcdiarmid}) to prove the result. Consider first $ P_0(g(\Ddobsn) - \E [g(\Ddobsn)] \ge \epsilon)  $; thus:
	\begin{equation}\label{}
	\begin{aligned}
	\left|\delta_i g(\Ddobsn) \right|  &= \left|\sup_z \left\{   \frac{2}{n(n-1)} \sum_{\substack{j=1\\j\neq i}}^n k(z, \Dobs_j) - \frac{2}{n} \E_{\Dsim \sim P_\theta} k(\Dsim,z) \right\} - \inf_z \left\{   \frac{2}{n(n-1)} \sum_{\substack{j=1\\j\neq i}}^n k(z, \Dobs_j) - \frac{2}{n} \E_{\Dsim \sim P_\theta} k(\Dsim,z) \right\}\right|  \\
	&=\left| \sup_z \left\{   \frac{2}{n(n-1)} \sum_{\substack{j=1\\j\neq i}}^n k(z, \Dobs_j) - \frac{2}{n} \E_{\Dsim \sim P_\theta} k(\Dsim,z) \right\} + \sup_z \left\{   \frac{2}{n} \E_{\Dsim \sim P_\theta} k(\Dsim,z) - \frac{2}{n(n-1)} \sum_{\substack{j=1\\j\neq i}}^n k(z, \Dobs_j) \right\}\right|  \\
	&\le\sup_z \left|    \frac{2}{n(n-1)} \sum_{\substack{j=1\\j\neq i}}^n k(z, \Dobs_j) - \frac{2}{n} \E_{\Dsim \sim P_\theta} k(\Dsim,z) \right|+ \sup_z \left|   \frac{2}{n} \E_{\Dsim \sim P_\theta} k(\Dsim,z) - \frac{2}{n(n-1)} \sum_{\substack{j=1\\j\neq i}}^n k(z, \Dobs_j) \right|  \\
	&=2 \cdot \frac{2}{n}\sup_z \left|    \frac{1}{n-1} \sum_{\substack{j=1\\j\neq i}}^n k(z, \Dobs_j) - \E_{\Dsim \sim P_\theta} k(\Dsim,z) \right| \le\frac{4}{n} \sup_z \left\{     \frac{1}{n-1} \sum_{\substack{j=1\\j\neq i}}^n \left|k(z, \Dobs_j)\right| + \E_{\Dsim \sim P_\theta} \left|k(\Dsim,z) \right|\right\} \\
	&\le\frac{4}{n} \left\{     \frac{1}{n-1} \sum_{\substack{j=1\\j\neq i}}^n \underbrace{\sup_z  \left|k(z, \Dobs_j)\right|}_{\le\kappa} + \E_{\Dsim \sim P_\theta} \underbrace{\sup_z \left|k(\Dsim,z) \right|}_{\le\kappa}\right\} \le \frac{4}{n} \left\{\frac{1}{n-1}\cdot (n-1) \kappa + \kappa \right\} = \frac{8\kappa}{n}
	\end{aligned}
	\end{equation}
	As the bound does not depend on $ \Ddobsn $, we have that $ ||\delta_i g ||_\infty \le \frac{8\kappa}{n} $, from which McDiarmid's inequality (Lemma~\ref{lemma:mcdiarmid}) gives:
	\begin{equation}\
	P_0(g(\Ddobsn) - \E [g(\Ddobsn)] \ge \epsilon) \le \exp\left(- \frac{2 \epsilon^2}{n \cdot \frac{64 \kappa^2}{n^2}}\right) = e^{-\frac{n \epsilon^2}{32 \kappa^2}}.
	\end{equation}
	For the bound on the other side, notice that  $ ||\delta_i(- g) ||_\infty  =||\delta_i g ||_\infty $; therefore, we also have
	\begin{equation}\
	P_0(-g(\Ddobsn) - \E [-g(\Ddobsn)] \ge \epsilon) \le e^{-\frac{n \epsilon^2}{32 \kappa^2}},
	\end{equation}
	from which:
	\begin{equation}\label{}
	P_0(\left|g(\Ddobsn) - \E [g(\Ddobsn)]\right| \ge \epsilon) \le 2 e^{-\frac{n \epsilon^2}{32 \kappa^2}}.
	\end{equation}
	Defining the right hand side of the bound as $ \delta $, we get:
	\begin{equation}\label{}
	P_0\left(\left|g(\Ddobsn) - \E [g(\Ddobsn)]\right| \ge \sqrt{-\frac{32\kappa^2}{n}\log \frac{\delta}{2}}\right) \le \delta,
	\end{equation}
	from which the result is obtained taking the complement.
\end{proof}

We now give the analogous result for the energy score:

\begin{lemma}\label{lemma:concentration_energy}
	Consider $ L(\theta, \ddobsn) $ defined in Eq.~\eqref{Eq:empirical_loss_enSR} (corresponding to the loss function defining the energy score posterior) and $ \DEbeta(P_\theta, P_0) $ defined in Eq.~\eqref{Eq:energy_SR_divergence}; assume that the space $ \X  $ is bounded such that $ \sup_{x,y \in \X}||x-y||_2 \le B  <\infty $; therefore, we have:
	\begin{equation}\
	P_0\left(\left|	L(\theta, \Ddobsn) - \DEbeta(P_\theta, P_0) \right| \le \sqrt{-\frac{32B^{2\beta}}{n}\log \frac{\delta}{2}}\right) \ge 1- \delta.
	\end{equation}
\end{lemma}
\begin{proof}
	We rely on Lemma~\ref{lemma:concentration_MMD2}; in fact, recall that the kernel score recovers the energy score for $ k(x,y) = -\|x-y\|_2^\beta $. With this choice of $ k $, Eqs.~\eqref{Eq:empirical_loss_kerSR} and \eqref{Eq:kernel_SR_divergence} (considered in Lemma~\ref{lemma:concentration_MMD2}) respectively recover Eqs.~\eqref{Eq:empirical_loss_enSR} and \eqref{Eq:energy_SR_divergence}.

	Additionally, assuming $ \X $ to be bounded ensures that $|k(x,y)| =  \|x-y\|_2^\beta \le B^\beta$; therefore, we can apply Lemma~\ref{lemma:concentration_MMD2} with $ \kappa= B^\beta $, from which the result follows.
\end{proof}

We are finally ready to prove our generalization bound:

\begin{proof}[Proof of Theorem~\ref{Th:posterior_consistency}]
	The proof consists in verfying the assumptions of Lemma~\ref{lemma:Matsubara_8}, for both the energy and kernel score posteriors. First, notice that \ref{ass:prior_mass} is a specific case of \ref{ass:prior_mass_bis} by identifying $ J(\theta) = D_k(P_\theta, P_0)  $ or $ J(\theta) = \DE(P_\theta, P_0)  $. We therefore need to verify the first and second assumptions only.

	Let us first consider the kernel score posterior (part 1 of Theorem~\ref{Th:posterior_consistency}).
	Recall that, for positive-definite, Cauchy-Schwarz inequality holds:
	\begin{equation}\label{}
	|k(x,y)| \le \sqrt{k(x,x) k(y,y)}.
	\end{equation} 
	Hence, the boundedness assumption stated in part 1 of Theorem~\ref{Th:posterior_consistency}  implies that in Lemma~\ref{lemma:concentration_MMD2}:
	\begin{equation}\label{}
	\sup_{x,y \in \X}|k(x,y)| \le \sup_{x,y \in \X} \sqrt{k(x,x) k(y,y)} \le \kappa.
	\end{equation} 
	Also, the kernel score posterior 	corresponds to the generalized Bayes posterior in Eq.~\eqref{Eq:general_Bayes_post} by choosing $ L(\theta, \Ddobsn) $ defined in Eq.~\eqref{Eq:empirical_loss_kerSR}; with this choice of $ L(\theta, \Ddobsn) $, Lemma~\ref{lemma:concentration_MMD2} holds, which corresponds to the first assumption of Lemma~\ref{lemma:Matsubara_8} with $ J(\theta) = D_k(P_\theta, P_0)  $ ($ D_k $ being the divergence related to the kernel SR, defined in Eq.~\eqref{Eq:kernel_SR_divergence}) and:
    \begin{equation}\
	\epsilon_{n}(\delta)  = \sqrt{-\frac{32\kappa^2}{n}\log \frac{\delta}{2}}.
	\end{equation}
	Finally, we have that $ D_k(P_\theta, P_0)\ge 0 $, which ensures the second assumption of  Lemma~\ref{lemma:Matsubara_8}. Thus, we have, from Lemma~\ref{lemma:Matsubara_8}:
	\begin{equation}\
	{P_0}\left(\left|\int_{\Theta} D_k(P_\theta, P_0)   \pik(\theta|\Ddobsn ) \mathrm{d} \theta-D_k(P_{\thetastar}, P_0)  \right| \geq  \frac{1}{\sqrt n}\left( {\alpha_{1}+\frac{\alpha_{2}}{w}} +8 \kappa \sqrt{-2\log \frac{\delta}{2}}\right) \right) \leq \delta;
	\end{equation}
	by defining the deviation term as $ \epsilon $ and inverting the relation, we obtain the result for the kernel Score Posterior.

	The same steps can be taken for the the energy score posterior; specifically, we notice that it corresponds to the generalized Bayes posterior in Eq.~\eqref{Eq:general_Bayes_post} by choosing $ L(\theta, \Ddobsn) $ defined in Eq.~\eqref{Eq:empirical_loss_enSR}; with this choice of $ L(\theta, \Ddobsn) $, Lemma~\ref{lemma:concentration_energy} holds, which corresponds to the first assumption of Lemma~\ref{lemma:Matsubara_8} with $ J(\theta) = \DEbeta(P_\theta, P_0)  $ ($ \DEbeta $ being the divergence related to the Energy SR defined in Eq.~\eqref{Eq:energy_SR_divergence}) and:
	\begin{equation}\
	\epsilon_{n}(\delta)  = \sqrt{-\frac{32B^{2\beta}}{n}\log \frac{\delta}{2}}.
	\end{equation}
	Finally, we have that $ \DEbeta(P_\theta, P_0)\ge 0 $, which ensures the second assumption of  Lemma~\ref{lemma:Matsubara_8}. Thus, we have, from Lemma~\ref{lemma:Matsubara_8}:
	\begin{equation}\
	{P_0}\left(\left|\int_{\Theta} \DEbeta(P_\theta, P_0)   \piebeta(\theta|\Ddobsn ) \mathrm{d} \theta-\DEbeta(P_{\thetastar}, P_0)  \right| \geq  \frac{1}{\sqrt n}\left( {\alpha_{1}+\frac{\alpha_{2}}{w}} + 8 B^\beta\sqrt{-2\log \frac{\delta}{2}}\right) \right) \leq \delta;
	\end{equation}
	by defining the deviation term as $ \epsilon $ and inverting the relation, we obtain the result for the energy score Posterior.

\end{proof}

We remark here that Theorem 1 in \cite{cherief2020mmd} proved a similar generalization bound for the kernel score posterior holding in expectation (rather than in high probability, as for our bounds), albeit under a slightly different prior mass condition.

\subsection{Proof of Theorem~\ref{Th:robustness}}\label{app:robustness}

Global bias-robustness (for a generic constant $ C<\infty $) was shown in \cite{matsubara2021robust} for their kernelized Stein discrepancy (KSD) posterior. Here, we provide an upper bound for the constant $ C $ for both the kernel and energy score posteriors. 

To prove our result, we first generalize Lemma 5 in \cite{matsubara2021robust} (our Lemma~\ref{lemma:robustness}), which provides bounds on the constant for global bias-robustness for a generalized Bayes posterior depending on bounds on the loss function defining the posterior.

Across this Section, we define as $ \empP_n = \frac{1}{n}\sum_{i=1}^n \delta_{\dobs_i} $ the empirical distribution given by the observations $  \ddobsn = \{\dobs_i\}_{i=1}^n$ (considered to be non-random here) and consider the generalized Bayes posterior:
\begin{equation}\label{Eq:generalized_Bayes_posterior_empP}
\pi_L(\theta|\empP_n ) \propto\pi(\theta)\exp\left\{-w n L(\theta, \empP_n) \right\},
\end{equation}
from which the SR posterior in Eq.~\eqref{Eq:SR_posterior} with scoring rule $ S $ is recovered with:
\begin{equation}\label{}
L(\theta, \empP_n) = \frac{1}{n} \sum_{i=1}^n S(P_\theta, \dobs_i) = \E_{Y \sim \empP_n} S(P_\theta, \Dobs),
\end{equation}
We remark that the notation is here slightly different from Appendix~\ref{app:posterior_consistency}, in which we considered $ L $ to be a function of $ \theta $ and $ \ddobsn $ (compare Eq.~\ref{Eq:generalized_Bayes_posterior_empP} with Eq.~\ref{Eq:general_Bayes_post}). The reason of this will be clear in the following.

We start by stating the result we will rely on, to which we provide proof for ease of reference.

\begin{lemma}\label{lemma:robustness}
	Let $ \pi_L(\theta|\empP_n ) $ be the generalized posterior defined in Eq.~\eqref{Eq:generalized_Bayes_posterior_empP} for fixed $ n\in \mathbb{N} $, with a generic loss $ L(\theta, \empP_n) $ and prior $ \pi(\theta) $. Let $ \Delta_n = \sup_{\theta \in \Theta} L\left(\theta ; \empP_n\right) - \inf_{\theta \in \Theta} L\left(\theta ; \empP_n\right) $ and $\mathrm{D} L\left(z, \theta, \empP_{n}\right) =\left.({d} / {d} \epsilon) L\left(\theta , \empP_{n, \epsilon, z}\right)\right|_{\epsilon=0}$.
		
Then,
	\begin{equation}
\sup _{\theta \in \Theta} \sup _{z \in \mathcal{X}}\left|\operatorname{PIF}\left(z, \theta, \empP_n\right)\right| \le 2 w n e^{w n \Delta_n} \cdot \sup_{\theta\in\Theta} \sup_{z \in \X} \left|\mathrm{D} L\left(z, \theta, \empP_{n}\right)\right|\cdot \sup_{\theta\in\Theta} \pi(\theta).
\end{equation}

\end{lemma}
\begin{proof}

	First of all, Eq.~(17) of \cite{ghosh2016robust} demonstrates that
	\begin{equation}
	\begin{aligned}
	\operatorname{PIF}\left(z, \theta, \empP_n\right)&=w n \pi_L(\theta|\empP_n )\left(-\mathrm{D} L\left(z, \theta, \empP_n\right)+\int_{\Theta} \mathrm{D} L\left(z, \theta^{\prime}, \empP_n\right) \pi_L(\theta'|\empP_n ) \mathrm{d} \theta^{\prime}\right)\\
	&\le w n \pi_L(\theta|\empP_n )\left(\sup_{\theta'\in\Theta}\mathrm{D} L\left(z, \theta^{\prime}, \empP_n\right) -\mathrm{D} L\left(z, \theta, \empP_n\right)  \right),
	\end{aligned}
	\end{equation}
	where $ \operatorname{PIF} $ denotes the posterior influence function defined in Sec.~\ref{sec:robustness} in the main text and where the inequality holds due to the mean of a random variable always being smaller than the maximum value the variable can get.

	We can now get the following upper bound:
	\begin{equation}
	\begin{aligned}
	\sup _{\theta \in \Theta} \sup _{z \in \mathcal{X}}\left|\operatorname{PIF}\left(z, \theta, \empP_n\right)\right| &\leq w n \sup _{\theta \in \Theta}\left\{ \pi_L(\theta|\empP_n )\left(\sup _{z \in \mathcal{X}}\left|\mathrm{D} L\left(z, \theta, \empP_n\right)\right|+\sup _{{z} \in \mathcal{X}}\sup _{\theta' \in \Theta}\left|\mathrm{D} L\left(z, \theta^{\prime}, \empP_n\right)\right| \right)\right\}\\
	&\le  w n \sup _{\theta \in \Theta} \left\{ \pi_L(\theta|\empP_n )\sup _{z \in \mathcal{X}}\left|\mathrm{D} L\left(z, \theta, \empP_n\right)\right|\right\}+ w n \left\{ \sup _{\theta \in \Theta}\pi_L(\theta|\empP_n ) \cdot \sup _{{z} \in \mathcal{X}}\sup _{\theta' \in \Theta} \left|\mathrm{D} L\left(z, \theta^{\prime}, \empP_n\right)\right| \right\} \\
	&\le 2w n \left\{ \sup _{\theta \in \Theta}\pi_L(\theta|\empP_n ) \cdot\sup _{{z} \in \mathcal{X}}\sup _{\theta \in \Theta} \left|\mathrm{D} L\left(z, \theta, \empP_n\right)\right| \right\} .
	\end{aligned}
	\end{equation}

	Recall now that 
\begin{equation}\label{}
\begin{aligned}
	\pi_L(\theta|\empP_n )&=\frac{\pi(\theta) \exp \left\{-w n L\left(\theta ; \empP_n\right)\right\}}{\int_{\Theta}\pi(\theta) \exp \left\{-w n L\left(\theta ; \empP_n\right)\right) d\theta}\le \frac{\pi(\theta) \exp \left\{-w n \inf _{\theta \in \Theta} L\left(\theta ; \empP_n\right)\right\}}{\int_{\Theta}\pi(\theta) \exp \left\{-w n L\left(\theta ; \empP_n\right)\right)d\theta} \\
	&\le \frac{\pi(\theta) \exp \left\{-w n \inf_{\theta \in \Theta} L\left(\theta ; \empP_n\right)\right\}}{\inf_{\theta \in \Theta} \exp \left\{-w n L\left(\theta ; \empP_n\right)\right) }  =  \frac{\pi(\theta) \exp \left\{-w n \inf_{\theta \in \Theta} L\left(\theta ; \empP_n\right)\right)}{ \exp \left\{-w n \sup_{\theta \in \Theta} L\left(\theta ; \empP_n\right)\right) } \\
	 &= \pi(\theta) \exp \left\{w n \left(\sup_{\theta \in \Theta} L\left(\theta ; \empP_n\right) - \inf_{\theta \in \Theta} L\left(\theta ; \empP_n\right)\right)\right) , 
\end{aligned}
\end{equation}	
	Let us now denote $ \Delta_n = \sup_{\theta \in \Theta} L\left(\theta ; \empP_n\right) - \inf_{\theta \in \Theta} L\left(\theta ; \empP_n\right) $. From the upper bound above, we have:
	\begin{equation}
	\sup _{\theta \in \Theta} \sup _{z \in \mathcal{X}}\left|\operatorname{PIF}\left(z, \theta, \empP_n\right)\right| \leq 2 w n e^{w n \Delta_n} \sup _{\theta \in \Theta} \pi(\theta) \sup _{{z} \in \mathcal{X}}\sup _{\theta \in \Theta} \left|\mathrm{D} L\left(z, \theta, \empP_n\right)\right|
	\end{equation}
	as claimed.
\end{proof}

Next, we give the explicit form for $ \mathrm{D} L\left(z, \theta, \empP_{n}\right) $ in our case in the following Lemma:
\begin{lemma}\label{lemma:robustness_2}
	For $ L\left(\theta , \empP_{n, \epsilon, z}\right) = \E_{\Dobs\sim \empP_{n, \epsilon, z}} S(P_\theta, \Dobs) $, we have:
	\begin{equation}\label{}
	\mathrm{D} L\left(z, \theta, \empP_{n}\right) = S(P_\theta, z) - \E_{\Dobs\sim \empP_n} S(P_\theta, \Dobs);
	\end{equation}
	further, setting $ S = S_k $, where $ S_k $ is the kernel scoring rule with kernel $ k $, we have:
	\begin{equation}\label{}
	\mathrm{D} L\left(z, \theta, \empP_{n}\right) = 2 \E_{X \sim P_\theta}\left[ \E_{\Dobs\sim \empP_n} k(\Dsim,\Dobs) - k(\Dsim,z) \right];
	\end{equation}
	finally, the form for the energy score can be obtained by setting $ k(x,y) = - || x-y||_2^\beta $.
\end{lemma}
\begin{proof}
	For the first statement, notice that:
	\begin{equation}\label{}
	\E_{\Dobs\sim  \empP_{n, \epsilon, z}} S(P_\theta, \Dobs) = (1-\epsilon) \E_{\Dobs\sim  \empP_{n}} S(P_\theta, \Dobs) + \epsilon S(P_\theta, z),
	\end{equation}
	from which differentiating with respect to $ \epsilon $ gives the statement.

	For the second statement, recall the form for the kernel SR:
	\begin{equation}\label{}
	S_k(P, z) = \E_{\Dsim, \Dsim' \sim P}[k(\Dsim,\Dsim')] - 2 \E_{\Dsim\sim P} [k(\Dsim, z)], \end{equation}
	from which:
	\begin{equation}\label{}
	\begin{aligned}
	&\mathrm{D} L\left(z, \theta, \empP_{n}\right) = S_k(P_\theta, z) - \E_{\empP_n} S_k(P_\theta, \Dobs) \\&= \E_{\Dsim, \Dsim' \sim P_\theta}[k(\Dsim,\Dsim')] - 2 \E_{\Dsim\sim P_\theta} [k(\Dsim, z)] - \E_{\Dobs\sim \empP_n} \left[ \E_{\Dsim, \Dsim' \sim P_\theta}[k(\Dsim,\Dsim')] -2 \E_{\Dsim\sim P_\theta} [k(\Dsim, \Dobs)] \right] \\
	&= \cancel{\E_{\Dsim, \Dsim' \sim P_\theta}[k(\Dsim,\Dsim')]} - 2 \E_{\Dsim\sim P_\theta} [k(\Dsim, z)] - \cancel{\E_{\Dsim, \Dsim' \sim P_\theta}[k(\Dsim,\Dsim')]} +2 \E_{\Dobs\sim \empP_n} \E_{\Dsim\sim P_\theta} [k(\Dsim, \Dobs)] \\
	&= 2 \E_{X \sim P_\theta}\left[ E_{\Dobs\sim \empP_n} k(\Dsim,\Dobs) - k(\Dsim,z) \right].
	\end{aligned}
	\end{equation}		\end{proof}

Finally, we state the proof for Theorem~\ref{Th:robustness}:
\begin{proof}[Proof of Theorem~\ref{Th:robustness}]
	The proof consists in verifying the conditions necessary for Lemma \ref{lemma:robustness} for the Kernel and energy score posteriors

	First, let us consider the kernel score posterior; recall that, for positive-definite kernels, Cauchy-Schwarz inequality holds:
	\begin{equation}\label{}
	|k(x,y)| \le \sqrt{k(x,x) k(y,y)}.
	\end{equation}
	Hence, the boundedness assumption in Theorem~\ref{Th:robustness} yields:
	\begin{equation}\label{}
	\sup_{x,y \in \X}|k(x,y)| \le \sup_{x,y \in \X} \sqrt{k(x,x) k(y,y)} \le \kappa.
	\end{equation}
	Thus, we have:
	\begin{equation}\label{}
	\begin{aligned}
	\left| L(\theta, \empP_n)\right| &\le \frac{1}{n} \sum_{i=1}^n \left|S_k(P_\theta, \dobs_i) \right|\le \frac{1}{n} \sumi\E \left[ \left|k(\Dsim,\Dsim')\right| + \left| 2 k(\Dsim, \dobs_i)  \right| \right] \\
	&\le \frac{1}{n} \sumi \left[\kappa + 2\kappa\right] = 3 \kappa,
	\end{aligned}
	\end{equation}
	where all expectations are over $ X, X'\sim P_\theta $ and the bound exploits the fact that $ |k(x,y)|\le \kappa $. This implies that 
\begin{equation}\label{key}
	\Delta_n = \sup_{\theta \in \Theta} L\left(\theta ; \empP_n\right) - \inf_{\theta \in \Theta} L\left(\theta ; \empP_n\right) \le 6\kappa.
\end{equation}
	
	Using a similar argument as above, 
	notice that, for the kernel SR (using Lemma~\ref{lemma:robustness_2}):
	\begin{equation}\label{}
	\begin{aligned}
	\left|\mathrm{D} L\left(z, \theta, \empP_{n}\right)\right| &= 2 \left| \E_{X \sim P_\theta}\E_{\Dobs\sim \empP_n}\left[  k(\Dsim,\Dobs) - k(\Dsim,z) \right] \right| \\
	&\le 2  \E_{X \sim P_\theta}\E_{\Dobs\sim \empP_n}\left[  \left|k(\Dsim,\Dobs)\right| + \left|k(\Dsim,z)\right| \right] \\
	&\le 2  \E_{X \sim P_\theta}\E_{\Dobs\sim \empP_n}\left[ \kappa + \kappa \right] = 4\kappa.
	\end{aligned}
	\end{equation}
	Hence, by Lemma~\ref{lemma:robustness} we have, for the kernel score posterior
	\begin{equation}\label{key}
	\sup _{\theta \in \Theta} \sup _{z \in \mathcal{X}}\left|\operatorname{PIF}\left(z, \theta, \empP_n\right)\right| \le 8 w n \kappa e^{6wn\kappa } \sup_{\theta\in\Theta} \pi(\theta),
	\end{equation}
	as claimed.

	For the statement about the energy score posterior, we proceed in similar manner. First, let us show that, under the assumptions of the Theorem, $ L(\theta, \empP_{n}) $ for the energy score $ \SE^{(\beta)} $ is lower bounded; in fact:
	\begin{equation}\label{}
	\begin{aligned}
	L(\theta, \empP_n) &= \frac{1}{n} \sum_{i=1}^n \SE^{(\beta)}(P_\theta, \dobs_i) = \frac{1}{n} \sumi \E\left[  2  ||\Dsim- \dobs_i||_2^\beta   - ||\Dsim-\Dsim'||_2^\beta \right] =  \frac{2}{n} \sumi \E  ||\Dsim- \dobs_i||_2^\beta  - \E||\Dsim-\Dsim'||_2^\beta  \\
	&=  \underbrace{\frac{2}{n} \sumi \E  ||\Dsim- \dobs_i||_2^\beta - \E||\Dsim-\Dsim'||_2^\beta  - \frac{1}{n^2} \sum_{i,j=1}^n ||y_i-y_j||_2^\beta}_{=\DEbeta(P_\theta, \empP_{n})} + \frac{1}{n^2} \sum_{i,j=1}^n ||y_i-y_j||_2^\beta,
	\end{aligned}
	\end{equation}
	where $ \DEbeta(P_\theta, \empP_{n}) $ is the squared Energy Distance between $ P_\theta $ and the empirical distribution $ \empP_{n} $; as the Energy Distance is a distance between probability measures \citep{rizzo2016energy}, $ \DEbeta(P_\theta, \empP_{n})\ge 0  $, from which:
	\begin{equation}\
	L(\theta, \empP_n) =  \DEbeta(P_\theta, \empP_{n}) + \frac{1}{n^2} \sum_{i,j=1}^n ||y_i-y_j||_2^\beta \ge 0.
	\end{equation}

	Additionally, recall that, as we assume $ \X $ to be bounded, there exists $ B < \infty $ such that $ \sup_{x,y \in \X}\|x-y\|_2 \le B $. 
	Thus:
	\begin{equation}\label{}
	\begin{aligned}
	L(\theta, \empP_n) & = \frac{1}{n} \sumi \E\left[  2  ||\Dsim- \dobs_i||_2^\beta   - ||\Dsim-\Dsim'||_2^\beta \right] \le 
	\frac{2}{n}\sumi \E||\Dsim- \dobs_i||_2^\beta  \le 2 B^\beta
	\end{aligned}
	\end{equation}
		Hence, we have 
	\begin{equation}\label{key}
	\Delta_n = \sup_{\theta \in \Theta} L\left(\theta ; \empP_n\right) - \inf_{\theta \in \Theta} L\left(\theta ; \empP_n\right) \le 2B^\beta.
	\end{equation}

	Moreover, for the Energy SR (using Lemma~\ref{lemma:robustness_2}):
	\begin{equation}\label{}
	\begin{aligned}
	\left|\mathrm{D} L\left(z, \theta, \empP_{n}\right)\right| &= 2 \left| \E_{X \sim P_\theta}\E_{\Dobs\sim \empP_n}\left[  ||\Dsim-z||_2^\beta -||\Dsim-\Dobs||_2^\beta \right] \right| \\
	&\le 2  \E_{X \sim P_\theta}\E_{\Dobs\sim \empP_n}\left[  \left|||\Dsim-z||_2^\beta\right| + \left|||\Dsim-\Dobs||_2^\beta \right| \right] \\
	&\le 2  \E_{X \sim P_\theta}\E_{\Dobs\sim \empP_n}\left[ B^\beta + B^\beta \right] = 4B^\beta,
	\end{aligned}
	\end{equation}
	where the last inequality is due to $ z \in \X $ and the boundedness assumptions for $ \X $.
	Hence, by Lemma~\ref{lemma:robustness} we have, for the energy score posterior
\begin{equation}\label{key}
\sup _{\theta \in \Theta} \sup _{z \in \mathcal{X}}\left|\operatorname{PIF}\left(z, \theta, \empP_n\right)\right| \le 8 w n B^\beta e^{2wnB^\beta } \sup_{\theta\in\Theta} \pi(\theta),
\end{equation}
as claimed.
\end{proof}

\section{Changing data coordinates}\label{app:coordinates}

We give here some more details on the behavior of the SR posterior when the coordinate system used to represent the data is changed, as mentioned in Remark~\ref{remark:non-invariance}.

\paragraph{Frequentist estimator}
First, we investigate whether the minimum scoring rule estimator (for a strictly proper scoring rule) is affected by a transformation of the data. Specifically, considering a strictly proper $ S $, we are interested in whether $ \theta^\star_\Dobs = \argmin_{\theta \in \Theta} S(P_\theta^\Dobs,Q_\Dobs)  = \argmin_{\theta \in \Theta} D(P_\theta^\Dobs,Q_\Dobs) $ is the same as $ \theta^\star_Z = \argmin_{\theta \in \Theta} S(P_\theta^Z,Q_Z)  = \argmin_{\theta \in \Theta} D(P_\theta^Z,Q_Z) $, where $ Z = f(\Dobs) \implies \Dobs \sim Q_\Dobs \iff Z \sim Q_Z$ and $ \Dobs \sim P_\theta^\Dobs \iff Z \sim P_\theta^Z $. If the model is well specified, $ P^\Dobs_{\theta^\star_\Dobs} = Q_\Dobs ,P^Z_{\theta^\star_Z} = Q_Z \implies \theta^\star_\Dobs = \theta^\star_Z $. If the model is misspecified, for a generic SR the minimizer of the expected SR may change according to the parametrization. We remark how this is not a drawback of the frequentist minimum SR estimator but rather a feature, as such estimator is the parameter value corresponding to the model minimizing the chosen expected scoring rule from the data generating process \textit{in that coordinate system}, and is therefore completely reasonable for it to change when the coordinate system is modified.

Notice that a sufficient condition for $\theta_\Dobs^\star = \theta^\star_Z $ is $ S(P_\theta^\Dobs, \dobs) = a \cdot S(P_\theta^Z, z) + b$ for $ a>0, b \in \mathbb{R} $. This condition is verified when $ S $ is chosen to be the log-score, as in fact: $$ S(P_\theta^Z, f(\dobs)) = - \ln p_Z(f(\dobs) |\theta ) =  S(P_\theta^Z, \dobs) + \ln |J_f(\dobs)|,$$ where we assumed $ f $ to be a one-to-one function and we applied the change of variable formula to the density $ p_Z $.

\paragraph{Generalized Bayesian posterior}
For a single observation, let $ \pis^{\Dobs} $ denote the SR posterior conditioned on values of $ \Dobs $, while $ \pis^{Z} $ denote instead the posterior conditioned on values of $ Z=f(\Dobs) $ for some one-to-one function $ f $; in general, $  \pis^{\Dobs}(\theta|\dobs)\neq  \pis^{Z}(\theta|f(\dobs)) $. By denoting as $ w_Z $ (respectively $ w_\Dobs $) and $ P_\theta^Z $ (respectively $ P_\theta^\Dobs $) the weight and model distributions appearing in $ \pis^{Z} $ (resp. $ \pis^{\Dobs} $), the equality would in fact require $ w_Z S(P_\theta^Z, f(\dobs)) = w_\Dobs S(P_\theta^\Dobs,  \dobs) + C \ \forall\ \theta, \dobs $ for some choice of $ w_Z, w_\Dobs $ and for all transformations $ f $, where $ C $ is a constant in $ \theta $. Notice that this is satisfied for the standard Bayesian posterior (i.e., with the log-score) with $ w_Z =  w_\Dobs  =1$. Instead, for other scoring rules the above condition cannot be satisfied in general for any choice of $ w_Z, w_\Dobs $. For instance, consider the kernel SR:
\begin{equation}\label{}
S(P_\theta^Z, f(\dobs)) = \E[k(Z, \tilde{Z})] - \E [k(Z, f(\dobs))] = \E[k(f(\Dobs), f(\tilde{\Dobs}))] - \E [k(f(\Dobs), f(\dobs))];
\end{equation}
for general kernels and functions $ f $, the above is different from $ S(P_\theta^\Dobs, \dobs) = \E[k(\Dobs, \tilde{\Dobs})] - \E [k(\Dobs, f(x))] $ up to a constant, unless the kernel is redefined as well. Therefore, the posterior shape depends on the chosen data coordinates. Considering the expression for the kernel SR, it is clear that is a consequence of the fact that the likelihood principle is not satisfied (as the kernel SR does not only depend on the likelihood value at the observation). Similar argument holds for the energy score posterior as well.

We also remark that this is also the case for BSL \citep{price2018bayesian}, as in that case the model is assumed to be multivariate normal, and changing the data coordinates impacts their normality (in fact it is common practice in BSL to look for transformations of data which yield distribution as close as possible to a normal one).

The theoretical semiBSL posterior \citep{an2020robust}, instead, is invariant with respect to one-to-one transformation applied independently to each data coordinate, which do not affect the copula structure. Notice however that different data coordinate systems may yield better empirical estimates of the marginal KDEs from model simulations.

\section{Checking convergence of MCMC with the kernelized Stein discrepancy}\label{app:KSD}

        As SG-MCMC algorithms in general exhibit an asymptotic bias, we require a convergence test which accounts for this bias in the stationary distribution. We thus utilise the method of Kernelized Stein Discrepancy (KSD) proposed in \cite{gorham2017measuring}, which is especially applicable in the case of stochastic gradient MCMC as it depends on the target distribution only through its gradient.

        Given the samples of our parameter $\{\theta_1,...,\theta_n\}$ where $\theta_i \in \mathbb{R}^d$, we denote the empirical distribution described by these samples as $\tilde{\pi}$, and our target distribution as $\pi$. We consider the  Integral Probability Metric (IPM) defined over a class of test function $\mathcal{H}$,
        $$
        d_{\mathcal{H}}(\tilde{\pi}, \pi):=\sup _{h \in \mathcal{H}}\left|\mathbb{E}_{\tilde{\pi}}[h(\theta)]-\mathbb{E}_\pi[h({\theta})]\right|
        $$
        For IPMs such as the Wasserstein distance, we obtain a desirable property that $d_{\mathcal{H}}\left(\tilde{\pi}_K, \pi\right) \rightarrow 0$ implies $\tilde{\pi}_K \Rightarrow \pi$ (weak convergence of measures). However, since $\pi$ is not available for integration, we instead utilise a class of IPMs called Stein Discrepancy, constructed such that the test functions give zero mean under $\pi$. We do this by defining a Stein operator, $\mathcal{T}$, which maps functions $g$ : $\mathbb{R}^d \rightarrow \mathbb{R}^d$ from our Stein set, the domain $\mathcal{G}$. This is chosen such that $\mathbb{E}_\pi[(\mathcal{T} g)(Z)]=0$ for all $g \in \mathcal{G}$. Then we can define the Stein discrepancy:
        $$
        \begin{aligned}
        \mathcal{S}(\tilde{\pi}, \mathcal{T}, \mathcal{G}) := d_{\mathcal{T} \mathcal{G}}(\tilde{\pi}, \pi)&=\sup _{g \in \mathcal{G}}\left|\mathbb{E}_{\tilde{\pi}}[(\mathcal{T} g)(X)]-\mathbb{E}_\pi[(\mathcal{T} g)(Z)]\right| \\&=\sup _{g \in \mathcal{G}}\left|\mathbb{E}_{\tilde{\pi}}[(\mathcal{T} g)(X)]\right|
        \end{aligned}
        $$
        Thus, such a Stein operator and Stein set must be chosen to fulfil the Stein discrepancy condition and the desired convergence property. In \cite{gorham2017measuring}, the Stein operator is proposed to be the Langevin Stein operator,
        $$
        \left(\mathcal{T}_P g\right)(x) := \langle g(x), \nabla \log p(x)\rangle+\langle\nabla, g(x)\rangle
        $$
        and the corresponding Stein set, which is defined using a Reproducing Kernel Hilbert space of function $\mathcal{K}_k$. We denote $\|\cdot\|_{\mathcal{K}_k}$ to be the induced norm from the inner product in $\mathcal{K}_k$, and $k: \mathbb{R}^d \times \mathbb{R}^d \rightarrow \mathbb{R}$ be the reproducing kernel of $\mathcal{K}_k$. This is the kernelized Stein set:
        $$
        \mathcal{G}_{k,\|\cdot\|} := \left\{g=\left(g_1, \ldots, g_d\right) \mid\|v\|^* \leq 1 \text { for } v_j := \left\|g_j\right\|_{\mathcal{K}_k}\right\}
        $$
        where $g=\left(g_1, \ldots, g_d\right)$ is a vector-valued function. This combination of the Langevin Stein operator and the kernelized Stein set is known as the kernelized Stein Discrepancy (KSD) $S\left(\mu, \mathcal{T}_P, \mathcal{G}_{k,\|\cdot\|}\right)$, for a probability measure $\mu$. In \cite{gorham2017measuring}, the KSD was proven to have a closed form solution for any $\|\cdot\|$, which of particular interest to us is when $S\left(\tilde{\pi}, \mathcal{T}_P, \mathcal{G}_{k,\|\cdot\|}\right)$, 
        $$
        S\left(\tilde{\pi}, \mathcal{T}_P, \mathcal{G}_{k,\|\cdot\|}\right):=\sum_{j=1}^d \sqrt{\sum_{i, i^{\prime}=1}^n \frac{k_j^0\left(\theta_i, \theta_{i^{\prime}}\right)}{n^2}}
        $$
        where the Stein kernel for $j \in\{1, \ldots, d\}$ is given by
        $$
        \begin{aligned}
        k_j^0\left({\theta}, {\theta}^{\prime}\right) &=\left(\nabla_{{\theta}^{(j)}} U({\theta}) \nabla_{{\theta}^{(j)}} U\left({\theta}^{\prime}\right)\right) k\left({\theta}, {\theta}^{\prime}\right)+\nabla_{{\theta}^{(j)}} U({\theta}) \nabla_{{\theta}^{\prime(j)}} k\left({\theta}, {\theta}^{\prime}\right) \\
        &+\nabla_{{\theta}^{\prime(j)}} U\left({\theta}^{\prime}\right) \nabla_{{\theta}^{(j)}} k\left({\theta}, {\theta}^{\prime}\right)+\nabla_{{\theta}^{(j)}} \nabla_{{\theta}^{(j)}} k\left({\theta}, {\theta}^{\prime}\right),
        \end{aligned}
        $$
        where $U(\theta)$ is such that $\pi(\theta)\propto e^{-U(\theta)}$.
        Note that \cite{gorham2017measuring} recommended the use of the inverse multi quadric kernel, $k\left({\theta}, {\theta}^{\prime}\right)=\left(c^2+\left\|{\theta}-{\theta}^{\prime}\right\|_2^2\right)^\beta$ which gives desired convergence properties when $c>0$ and  $\beta \in(-1,0)$. 

        In our specific case of the SR posterior, $U(\theta) = w\cdot \sum_{i=1}^n S(P_\theta, y_i)$. As for the energy and kernel scores we cannot exactly evaluate $\nabla_\theta U(\theta)$, we reaplce it with an unbiased estimate when computing the KSD. 
\section{More details on related techniques}\label{app:related_techniques}

\subsection{Energy Distance}\label{app:energy}

The squared energy distance is a metric between probability distributions \citep{rizzo2016energy}, and is defined by:
\begin{equation}
\DEbeta(P,Q) = 2\cdot \E \left[\| \Dsim - \Dobs \|_2^\beta \right]- \E\left[\|\Dsim- \Dsim'\|_2^\beta \right]  - \E \left[\|\Dobs- \Dobs'\|_2^\beta \right],
\end{equation}
for $ \Dsim\independent \Dsim' \sim P $ and $ \Dobs \independent \Dobs' \sim Q$.

The probabilistic forecasting literature \citep{gneiting2007strictly} use a different convention of the energy score and distance, which amounts to multiplying our definitions by $ 1/2 $. We follow here the convention used in the statistical inference literature \citep{rizzo2016energy, cherief2020mmd, nguyen2020approximate}.

\subsection{Maximum Mean Discrepancy (MMD)}\label{app:MMD}
We follow here Section 2.2 in \cite{gretton2012kernel}; all proofs of our statements can be found there. Let $ k(\cdot, \cdot): \X \times \X \to \R $ be a positive-definite and symmetric kernel; notice that this implies $ k(x,x) \ge 0$. Under these conditions, there exists a unique Reproducing kernel Hilbert space (RKHS) $ \mathcal{H}_k $ of real functions on $ \X $ associated to $ k $.

Now, let's define the Maximum Mean Discrepancy (MMD).

\begin{definition}
	Let $ \mathcal{F} $ be a class of functions $ f : \X\to\R $; we define the MMD relative to $\mathcal{F}  $ as:
	\begin{equation}\label{}
	\MMD_\mathcal{F} (P,Q) = \sup_{f\in \mathcal{F}} \left[\E_{X \sim P} f(X) - \E_{Y \sim Q} f(Y) \right].
	\end{equation}
\end{definition}

We will show here how choosing $ \mathcal{F} $ to be the unit ball in an RKHS $ \mathcal{H}_k $ turns out to be computationally convenient, as it allows to avoid computing the supremum explicitly. First, let us define the mean embedding of the distribution $ P $ in $ \mathcal{H}_k $:
\begin{lemma}[Lemma 3 in \cite{gretton2012kernel}]\label{lemma:MMD_1}
	If $ k(\cdot, \cdot) $ is measurable and $ \E_{X \sim P}\sqrt{k(X,X)}<\infty $, then the mean embedding of the distribution $ P $ in $ \mathcal{H}_k $ is:$$ \mu_P = \E_{X \sim P} \left[k(X, \cdot) \right]\in \mathcal{H}_k .$$
\end{lemma}
Using this fact, the following Lemma shows that the MMD relative to $ \mathcal{H}_k $ can be expressed as the distance in $ \mathcal{H}_k  $ between the mean embeddings:
\begin{lemma}[Lemma 4 in \cite{gretton2012kernel}]\label{lemma:MMD_2}
	Assume the conditions in Lemma~\ref{lemma:MMD_1} are satisfied, and let $ \mathcal{F} $ be the unit ball in $ \mathcal{H}_k $; then:
	\begin{equation}\label{}
	\MMD_{\mathcal{F}} ^2(P,Q)= ||\mu_P - \mu_Q||_\mathcal{H}^2.
	\end{equation}
\end{lemma}

In general, the MMD is a \textit{pseudo-metric }for probability distributions (i.e., it is symmetric, satisfies the triangle inequality and $ \MMD_{\mathcal{F}}(P,P)=0 $, \citealt{briol2019statistical}).
For the probability measures on a compact metric space $ \X $, the next Lemma states the conditions under which the MMD is a \textit{metric}, which additionally ensures that $ \MMD_{\mathcal{F}} (P,Q) = 0 \implies P=Q $. Specifically, this holds when the kernel is universal, which requires that $ k(\cdot, \cdot) $ is continuous, and $ \mathcal{H}_k $ being dense in $ C(\X) $ with respect to the $ L_\infty $ norm (these conditions are satisfied by the Gaussian and Laplace kernel).
\begin{lemma}[Theorem 5 in \cite{gretton2012kernel}]
	Let $ \mathcal{F} $ be the unit ball in $ \mathcal{H}_k$, where $ \mathcal{H}_k$ is defined on a compact metric space $ \X $ and has associated continuous kernel $ k(\cdot, \cdot) $. Then:
	\begin{equation}\label{}
	\MMD_{\mathcal{F}} (P,Q) = 0 \iff P=Q.
	\end{equation}
\end{lemma}
This result can be generalized to more general spaces $ \X $, by considering the notion of characteristics kernel, for which the mean map is injective; it can be shown that the Laplace and Gaussian kernels are characteristics \citep{gretton2012kernel}, so that MMD for those two kernels is a metric for distributions on $ \R^d $.

Additionally, the form of MMD for a unit-ball in an RKHS allows easy estimation, as shown next:
\begin{lemma}[Lemma 6 in \cite{gretton2012kernel}]
	Assume that the form for MMD given in Lemma~\ref{lemma:MMD_2} holds; say $ X\independent X'\sim P $, $ Y\independent Y' \sim Q $, and let $ \mathcal{F} $ be the unit ball in $ \mathcal{H}_k$. Then, you can write:
	\begin{equation}\label{}
	\MMD_{\mathcal{F}}^2(P,Q)  = \E [k(X, X')] + \E [k(Y, Y')]- 2 \E [k(X,Y)].
	\end{equation}
\end{lemma}

\subsubsection{Equivalence between MMD-Bayes posterior and $ \pik $}\label{app:MMDBayes_equivalence}

\cite{cherief2020mmd} considered the following posterior, termed MMD-Bayes:

\begin{equation}\label{}
\piMMD(\theta|\ddobsn) \propto \pi(\theta)\exp \left\{-\beta \cdot {D}_{k}\left(P_{\theta}, \hat{P}_{n}\right)\right\}
\end{equation}
where $ \beta>0 $ is a temperature parameter and $ {D}_{k}\left(P_{\theta}, \hat{P}_{n}\right) $ denotes the squared MMD between the empirical measure of the observations $ \empP_n = \frac{1}{n} \sum_{i=1}^{n} \delta_{\dobs_i} $  and the model distribution $ P_\theta $.

From the properties of MMD (see Appendix~\ref{app:MMD}), notice that:
\begin{equation}\label{}
\begin{aligned}
{D}_{k}\left(P_{\theta}, \hat{P}_{n}\right) &= \E_{\Dsim, \Dsim' \sim P_\theta} k(\Dsim, \Dsim') + \frac{1}{n^2} \sum_{i,j=1}^{n} k(\dobs_i, \dobs_j) -\frac{2}{n} \sum_{i=1}^n\E_{\Dsim \sim P_\theta} k(\Dsim, \dobs_i)\\
&= \frac{1}{n} \left(n\cdot  \E_{\Dsim, \Dsim' \sim P_\theta}k(\Dsim,\Dsim') - 2 \sum_{i=1}^n   \E_{\Dsim \sim P_\theta} k(\Dsim, \dobs_i)\right) + \frac{1}{n^2} \sum_{i,j=1}^{n} k(\dobs_i, \dobs_j)\\
&= \frac{1}{n} \left(\sum_{i=1}^n  S_k(P_\theta, \dobs_i)\right) + \frac{1}{n^2} \sum_{i,j=1}^{n} k(\dobs_i, \dobs_j),
\end{aligned}
\end{equation}
where we used the expression of the SR scoring rule $ S_k $, and where the second term is independent on $ \theta $. Therefore, the MMD-Bayes posterior is equivalent to the SR posterior with kernel scoring rule $ S_k $, by identifying $ w = \beta/n $.

\subsection{The Dawid--Sebastiani score}\label{app:DS_score}

As mentioned in Sec.~\ref{sec:BSL_multi_gk}, the BSL posterior can be seen as a scoring rule posterior with $ w=1 $ considering the Dawid--Sebastiani (DS) score, which is defined as:
	\begin{equation}\label{Eq:DS_score}
	S_{\operatorname{DS}}(P,\dobs) = \ln |\Sigma_P| + (\dobs- \mu_P)^T\Sigma_P^{-1} (\dobs- \mu_P),
	\end{equation}
	where $ \mu_P $ and $ \Sigma_P $ are the mean vector and covariance matrix of $ P $. The DS score is the negative log-likelihood of a multivariate normal distribution with mean $ \mu_P $ and covariance matrix $ \Sigma_P $, up to some constants. Therefore, it is equivalent to the log score when $ P $ is a multivariate normal distribution. For a set of distributions $ \mathcal{P}(\X) $ with well-defined second moments, this SR is proper but not strictly so: several distributions of that class may yield the same score, as long as the two first moments match \citep{gneiting2007strictly}. It is strictly proper if distributions in $\mathcal{P}(\X)  $ are determined by their first two moments, as it is the case for the normal distribution.

\subsection{Semi-Parametric Synthetic Likelihood}\label{app:semiBSL}

We review here the semiBSL approach \citep{an2020robust}.%

\paragraph{Copula theory}
First, recall that a copula is a multivariate Cumulative Density Function (CDF) such that the marginal distribution for each variable is uniform on the interval $ [0,1] $. Consider now a multivariate random variable $ X = (X^1, \ldots, X^d) $, for which the marginal CDFs are denoted by $ F_j(x) = \mathbb{P}(X^j<x) $; then, the multivariate random variable built as:
\begin{equation}\label{}
(U^1, U^2, \ldots, U^d) = (F_1(X^1), F_2(X^2), \ldots, F_d(X^d))
\end{equation}
has uniform marginals on $ [0,1] $.

Sklar's theorem exploits copulas to decompose the density $ h $ of $ X $\footnote{Provided that the density exists in the first place; a more general version of Sklar's theorem is concerned with general random variables, but we restrict here to the case where densities are available.}; specifically, it states that the following decomposition is valid:
\begin{equation}\label{}
h(x^1, \ldots, x^d) = c(F_1(x^1), \ldots, F_d(x^d)) f_1(x^1) \cdots f_d(x^d),
\end{equation}
where $ f_j $ is the marginal density of the $ j$-th coordinate, and $ c $ is the density of the copula.

We now review definition and properties of the Gaussian copula,
which is defined by a correlation matrix $ {R} \in [-1,1]^{d\times d}$, and has cumulative density function:
\begin{equation}\label{}
C_{R}(u) = \Phi_{R}(\Phi^{-1}(u^1), \ldots, \Phi^{-1}(u^d)),
\end{equation}
where $ \Phi^{-1} $ is the inverse cdf (quantile function) of a standard normal, and $ \Phi_{R} $ is the cdf of a multivariate normal with covariance matrix $ {R} $ and 0 mean. If you define as $ U $ the random variable which is distributed according to $ C_{R} $, it can be easily seen that $ {R} $ is the covariance matrix of the multivariate normal random variable $ {Z} = \Phi^{-1}({U}) $, where $ \Phi^{-1} $ is applied element-wise. In fact:
\begin{equation}\label{}
P({Z} \le \eta) = P({U} \le \Phi(\eta)) = 	C_{R}(\Phi(\eta)) = \Phi_{R}(\eta),
\end{equation}
where the inequalities are intended component-wise.

By defining as $ \eta $ a \textit{d}-vector with components $ \eta^k =\Phi^{-1}(u^k) $, the Gaussian copula density is:
\begin{equation}\label{}
c_{{R}}(u)=\frac{1}{\sqrt{|{R}|}} \exp \left\{-\frac{1}{2} \eta^{\top}\left({R}^{-1}-\mathbf{I}_{d}\right) \eta\right\},
\end{equation}
where $ \mathbf{I}_{d} $ is a \textit{d}-dimensional identity matrix, and $ |\cdot| $ denotes the determinant.

\paragraph{Semiparametric Bayesian Synthetic Likelihood (semiBSL)}
The semiBSL approach assumes that the likelihood for the model has a Gaussian copula; therefore, the likelihood for a single observation $ \dobs $ can be written as:
\begin{equation}
p_{\text{semiBSL}}\left(\dobs | \theta\right)=c_{{R}_\theta}(F_{\theta, 1}(\dobs^{1}), \ldots , F_{\theta, d}(\dobs^{d}))  \prod_{k=1}^{d} {f}_{\theta, k}\left(\dobs^{k}\right),
\end{equation}
where $ \dobs^k $ is the \textit{k}-th component of $ \dobs $, $ {f}_{\theta, k} $ is the marginal density of the \textit{k}-th component and $ F_{\theta, k} $ is the CDF of the \textit{k}-th component.

In order to obtain an estimate for it, we exploit simulations from $  P_\theta$ to estimate $ {R}_\theta $, $ f_{\theta,k} $ and $ F_{\theta, k} $; this leads to:
\begin{equation}
\begin{aligned}
\hat p_{\text{semiBSL}}\left(\dobs | \theta\right)&= c_{\hat{ {R}}_\theta}(\hat{F}_{\theta, 1}(\dobs^{1}), \ldots , \hat{F}_{\theta, d}(\dobs^{d}))  \prod_{k=1}^{d} \hat{f}_{\theta, k}\left(\dobs^{k}\right) \\
&= \frac{1}{\sqrt{|\hat{{R}}_\theta|}} \exp \left\{-\frac{1}{2} \hat{\eta}_{\dobs}^{\top}\left(\hat{{R}}_\theta^{-1}-\mathbf{I}_{d}\right) \hat{\eta}_{\dobs}\right\} \prod_{k=1}^{d} \hat{f}_{\theta, k}\left(\dobs^{k}\right),
\end{aligned}
\end{equation}
where $ \hat{f}_{\theta, k} $ and $ \hat F_{\theta, k} $ are estimates for $ {f}_{\theta, k} $ and $ F_{\theta, k} $, $ \hat{\eta}_{\dobs} = (\hat{\eta}_{\dobs}^1, \ldots, \hat{\eta}_{\dobs}^d) $, $ \hat{\eta}_{\dobs}^k = \Phi^{-1}(\hat u^k) $, $ \hat u^k = \hat F_{\theta, k}(\dobs^k) $.
Moreover, $ \hat{{R}}_\theta $ is an estimate of the correlation matrix.

We discuss now how the different quantities are estimated. First, a Kernel Density Estimate (KDE) is used for the marginals densities and cumulative density functions. Specifically, given samples $ \dsim_1 ,\ldots, \dsim_m \sim P_\theta$, a KDE estimate for the $ k $-th marginal density is:
\begin{equation}\label{}
\hat f_{\theta,k}(\dobs^k) = \frac{1}{m} \sum_{j=1}^m K_h(\dobs^k-\dsim_j^k),
\end{equation}
where $ K_h $ is a normalized kernel which is chosen to be Gaussian in the original implementation \citep{an2020robust}. The CDF estimates are obtained by integrating the KDE density.

Next, for estimating the correlation matrix, \cite{an2020robust} proposed to use a robust procedure based on the ranks (grc, Gaussian rank correlation, \citealp{boudt2012gaussian}); specifically, given $ m $ simulations $ \dsim_1 ,\ldots, \dsim_m \sim P_\theta$, the estimate for the $ (k,l) $-th entry of $ {R}_\theta $ is given by:
\begin{equation}\label{}
\left[\hat{{R}}_\theta^{\mathrm{grc}}\right]_{k, l}=\frac{\sum_{j=1}^{m} \Phi^{-1}\left(\frac{r\left(\dsim_{j}^{k}\right)}{m+1}\right) \Phi^{-1}\left(\frac{r\left(\dsim_{j}^{l}\right)}{m+1}\right)}{\sum_{j=1}^{m} \Phi^{-1}\left(\frac{j}{m+1}\right)^{2}},
\end{equation}
where $ r(\cdot) : \R \to \mathcal{A} $, where $ \mathcal{A} = \{1, \ldots, m\} $ is the rank function.

\paragraph{Copula scoring rule}
Finally, we write down the explicit expression of the copula scoring rule $ S_{Gc} $, associated to the Gaussian copula. We show that this is a proper, but not strictly so, scoring rule for copula distributions. Specifically, let $ C $ be a distribution for a copula random variable, and let $ u \in [0,1]^d$. We define:
\begin{equation}\label{}
S_{Gc}(C, u) = \frac{1}{2}\log |{R}_C| + \frac{1}{2} \left(\Phi^{-1} (u)\right)^T ({R}_C^{-1} - \mathbf{I}_d) \Phi^{-1} (u),
\end{equation}
where $ \Phi^{-1} $ is applied element-wise to $ u $, and $ {R}_C $ is the correlation matrix associated to $ C $ in the following way: define the copula random variable $ V \sim C $ and its transformation $ \Phi^{-1}(V) $; then, $ \Phi^{-1}(V) $ will have a multivariate normal distribution with mean 0 and covariance matrix $ {R}_C $.

Similarly to the Dawid--Sebastiani score (see Appendix~\ref{app:DS_score}), this scoring rule is proper but not strictly so as it only depends on the first 2 moments of the distribution of the random variable $ \Phi^{-1}(V) $ (the first one being equal to 0). To show this, assume the copula random variable $ U $ has an exact distribution $ Q $ and consider the expected scoring rule:
\begin{equation}\label{}
\begin{aligned}
S_{Gc} (C, Q) & = \E_{U \sim Q} S_{Gc} (C, U) = \frac{1}{2} \log |{R}_C| + E_{U \sim Q} \left[\left(\Phi^{-1} (U)\right)^T ({R}_C^{-1} - \mathbf{I}_d) \Phi^{-1} (U) \right];
\end{aligned}
\end{equation}
now, notice that $ \Phi^{-1}(U) $ is a multivariate normal distribution whose marginals are standard normals. Therefore, let us denote as $ {R}_Q $ the covariance matrix of $ \Phi^{-1}(U) $, which is a correlation matrix. From the well-known form for the expectation of a quadratic form\footnote{$ \E\left[X^T\Lambda X\right]=\operatorname{tr}\left[\Lambda \Sigma\right] + \mu^T\Lambda\mu $, for a symmetric matrix $ \Lambda $, and where $ \mu $ and $ \Sigma $ are the mean and covariance matrix of $ X $ (which in general does not need to be normal, but only needs to have well defined second moments). }, it follows that:
\begin{equation}\label{}
\begin{aligned}
S_{Gc} (C, Q) &=\frac{1}{2} \log |{R}_C|  + \frac{1}{2} \Tr\left[({R}_C^{-1}-\mathbf{I}_d) \cdot {R}_Q\right] \\
&= \frac{1}{2} \log |{R}_C|  + \frac{1}{2} \Tr\left[{R}_C^{-1} \cdot {R}_Q\right] - \frac{1}{2}\Tr\left[{R}_Q\right]\\
&= \underbrace{\frac{1}{2}   \left\{\log \frac{|{R}_C|}{|{R}_Q|}-d + \Tr \left[{R}_C^{-1} \cdot {R}_Q\right] \right\}}_{D_{KL}(Z_Q||Z_C)}   + \frac{1}{2} \log {R}_Q + \frac{d}{2} - \frac{1}{2} \Tr\left[{R}_Q\right],
\end{aligned}
\end{equation}
where $ D_{KL}(Z_Q||Z_C) $ is the KL divergence between two multivariate normal distributions $ Z_Q $ and $ Z_C $ of dimension $ d $, with mean 0 and covariance matrix $ {R}_Q $ and $ {R}_C $ respectively. Further, notice that the remaining factors do not depend on the distribution $ C $. Therefore, $ S_{Gc} (C, Q)  $ is minimized whenever $ {R}_C $ is equal to $ {R}_Q $; this happens when $ C = Q $, but also for all other choices of $ C $ which share the associated covariance matrix with $ Q $. This implies that the Gaussian copula score is a proper, but not strictly so, scoring rule for copula distributions.

\subsection{Ratio estimation}\label{app:ratio_estimation}

	The standard Bayes posterior can be written as $ \pi(\theta|\dobs) = \pi(\theta)\cdot r(\dobs; \theta) $, with $ r(\dobs; \theta) = \frac{p(\dobs|\theta)}{p(\dobs)} $. The Ratio Estimation (RE) approach \citep{thomas2020likelihood} builds an approximate posterior by estimating $ \log r(\dobs;\theta) $ with some function $ \hat h^\theta(\dobs)  $ and considering $ \pi_{\operatorname{re}}(\theta|\dobs) \propto  \pi(\theta) \exp(\hat h^\theta(\dobs))$.

	\cite{thomas2020likelihood} run an MCMC where, for each proposed $ \theta $, $ m $ samples $ \ddsimmtheta $ are generated from $ P_\theta  $. These, together with a set of $ m $ reference samples $ \ddsimm^{(r)} = \{\dsim_j^{(r)}\}_{j=1}^m$ from the marginal data distribution\footnote{Which are obtained by drawing $\theta_j \sim p(\theta)$, $\dsim_j \sim p(\cdot|\theta_j)$, and discarding $ \theta_j $. \\In general, the number of reference samples and samples from the model can be different, see Appendix~\ref{app:ratio_estimation}; we make this choice here for the sake of simplicity.}, are used to fit a logistic regression yielding $ \hat h^\theta(\dobs) $.
	Logistic regression is an optimization problem in which the best function of $ \X $ in distinguishing between the two sets of samples is selected. If $ m\to\infty $ and all scalar functions are considered, the optimum $ h_\star^\theta$ is equal to $\log r(\dobs;\theta) $. For finite data, however, the corresponding optimum $ \hat h_m^\theta $ is only an approximation of the ratio (as discussed in Appendix~\ref{app:ratio_estimation}).
	RE is therefore a specific case of our SR posterior framework with $ w=1 $ and:
	\begin{equation}
	\hat S_{\operatorname{RE}}(\ddsimmtheta,\ddsimm^{(r)} ,\dobs) = - \hat h_m^\theta(\dobs)
	\end{equation}
	which, differently from the other SR estimators considered previously, also depends on the reference samples. Due to what we discussed above, $ \hat S_{\operatorname{RE}} $ converges in probability to the log-score (up to a constant term in $ \theta $) for $ m\to\infty $.

	The above argument relies on optimizing over all functions in logistic regression; in practice, the optimization is restricted to a set of functions $ \mathcal{H} $ (for instance, a linear combination of predictors). In this case, the infinite data optimum $ h^\theta_{\mathcal H \star}(\dobs) $ does not correspond to $ \log r(\dobs; \theta) $ (see Appendix~\ref{app:ratio_estimation}), but to the best possible approximation in $ \mathcal H $ in some sense. Therefore, Ratio Estimation with a restricted set of functions $ \mathcal H $ cannot be written exactly under our SR posterior framework. However, very flexible function classes (as for instance neural networks) can produce reasonable approximations to the log score for large values of $ m $.

\section{Tuning the bandwidth of the Gaussian kernel}\label{app:bandwidth}

Consider the Gaussian kernel:

\begin{equation}\label{Eq:gau_k}
k(x, y)=\exp \left(-\frac{\|x-y\|_{2}^{2}}{2 \gamma^{2}}\right);
\end{equation}
inspired by \cite{park2016k2}, we fix the bandwidth $ \gamma $ with the following procedure:
\begin{enumerate}
	\item Simulate a value $ \theta_j \sim \pi(\theta) $ and a set of samples $ \dsim_{jk} \sim P_{\theta_j} $, for $ k=1, \ldots, m_\gamma $ .
	\item Estimate the median of $ \{ ||\dsim_{jk} - \dsim_{jl} ||_2 \}_{kl}^{m_\gamma} $ and call it $ \hat \gamma_j $.
	\item Repeat points 1) and 2) for $ j=1,\ldots, m_{\theta, \gamma}  $.
	\item Set the estimate for $ \gamma $ as the median of $ \{ \hat \gamma_j\}_{j=1}^{m_{\theta, \gamma}} $.
\end{enumerate}

Empirically, we use $ m_{\theta, \gamma}=1000 $ and we set $ m_\gamma $ to the corresponding value of $ m $ for the different models.

\section{Further details on simulation studies reported in the main text}\label{app:experimental_details}

\subsection{The g-and-k model}\label{app:gk_SG}
We report here additional experimental details on the g-and-k model experiments.

\subsubsection{Univariate g-and-k}

\paragraph{SG-MCMC and PM-MCMC comparison}
We ran our inference with observations of $n=10$. Both energy score posteriors for PM-MCMC and SG-MCMC was set to $w=1$. 
\begin{itemize}
    \item For the SR posterior with SG-MCMC, we utilised the adSGLD algorithm, with the step-size $\epsilon$ tuned with the Multi-Armed Bandit algorithm \cite{coullon2021efficient} as discussed previously. The chain was initialized at a parameter value of $0$. This resulted in $\epsilon = 3 \times 10^{-3}$.
    \item  For the SR posterior with PM-MCMC, we utilised a proposal size of $\sigma = 1$.
\end{itemize}

\paragraph{Concentration study}
For our concentration study, we ran our inference with increasing observations of $n=1,10,20,50,70,100,200$. Generally, we ran the chain with the Multi-Armed Bandit algorithm \cite{coullon2021efficient} as discussed previously. The chains were started from an initial optimization step of $250$ iterations ran with the Adam optimizer \cite{kingma2014adam}. In Table~\ref{tab:uni_gk_step_size}, we report the final step-size determined by the Multi-Armed Bandit algorithm for different values of $n$. We detail below the settings for the different SR posteriors.

\begin{itemize}
    \item The energy score posteriors were set to $w=1$.
    \item For the kernel score posteriors, we set $w$ using our heuristic procedure discussed in Sec.~\ref{sec:fix_w} with the energy score posterior as a reference, resulting in $w=28.1$. The Gaussian kernel bandwidth $\gamma$, was tuned using the procedure detailed in \ref{app:bandwidth}, resulting in $\gamma=5.47$.
\end{itemize}

\begin{table}[h!]
	\centering
	\caption{Step-sizes for the two SR posteriors in the univariate g-and-k model, determined with the Multi-Armed Bandit algorithm of \cite{coullon2021efficient}.}
	\begin{adjustbox}{max width=\textwidth}

		\begin{tabular}{ l|cccccccc }
			\toprule
			\textbf{Observations} & $ n=1 $  & $ n=10 $ & $ n=20 $ & $ n=50 $ & $ n=70 $ & $ n=100 $ & $ n=200 $ & $ n=400 $  \\
			\midrule
			\textbf{Energy score} & $3 \times 10^{-2}$ & $3 \times 10^{-2}$ & $3 \times 10^{-3}$ & $3 \times 10^{-4}$ & $1 \times 10^{-3}$ & $1 \times 10^{-4}$ & $1 \times 10^{-4}$ & $3 \times 10^{-6}$  \\
			\textbf{Kernel score} & $1 \times 10^{-1}$ & $3 \times 10^{-2}$ & $1 \times 10^{-2}$ & $1 \times 10^{-3}$ & $1 \times 10^{-3}$ & $1 \times 10^{-4}$ & $3 \times 10^{-5}$ & $3 \times 10^{-5}$  \\
			\bottomrule
		\end{tabular}
	\end{adjustbox}
	\label{tab:uni_gk_step_size}
\end{table}

\subsubsection{Multivariate g-and-k}\label{app:gk_Cauchy_prop_size}
Similar to the univariate model, we ran our inference with increasing observations of $n=1,10,20,50,$ $70,100,200$ and with the Multi-Armed Bandit algorithm \cite{coullon2021efficient} as discussed previously. The chains were started from an initial optimization step of $250$ iterations ran with the Adam optimizer \cite{kingma2014adam}. In Table~\ref{tab:multi_gk_well_spec_step_size} and Table~\ref{tab:multi_gk_miss_spec_step_size}, we report the final step-size determined by the Multi-Armed Bandit algorithm for different values of $n$ for the well-specified case and the misspecified case respectively.

 We detail below the settings for the different SR posteriors and for the BSL posterior.

\paragraph{Well-specified case}
\begin{itemize}
    \item The energy score posteriors were set to $w=1$.
    \item For the kernel score posteriors, we set $w$ using our heuristic procedure discussed earlier with the energy score posterior as a reference, resulting in $w=191$. The Gaussian kernel bandwidth $\gamma$, was tuned using the procedure detailed in \ref{app:bandwidth}, resulting in $\gamma=45$.
    \item For the BSL posteriors, we set $\sigma = 1$. However, the chain was unable to converge for any $n > 10$, and so we ran the BSL posterior with an additional $n=5$ observations.
\end{itemize}

\begin{table}[h!]
	\centering
	\caption{Step-sizes for the two SR posteriors in the multivariate g-and-k model with well-specified observations, determined with the Multi-Armed Bandit algorithm of \cite{coullon2021efficient}.}
	\begin{adjustbox}{max width=\textwidth}

		\begin{tabular}{ l|cccccccc }
			\toprule
			\textbf{Observations} & $ n=1 $  & $ n=10 $ & $ n=20 $ & $ n=50 $ & $ n=70 $ & $ n=100 $ & $ n=200 $ & $ n=400 $  \\
			\midrule
			\textbf{Energy score} & $1 \times 10^{-1}$ & $3 \times 10^{-3}$ & $1 \times 10^{-3}$ & $1 \times 10^{-3}$ & $1 \times 10^{-4}$ & $3 \times 10^{-5}$ & $1 \times 10^{-5}$ & $3 \times 10^{-6}$  \\
			\textbf{Kernel score} & $1 \times 10^{-1}$ & $3 \times 10^{-4}$ & $3 \times 10^{-4}$ & $1 \times 10^{-4}$ & $1 \times 10^{-5}$ & $3 \times 10^{-6}$ & $1 \times 10^{-5}$ & $1 \times 10^{-6}$  \\
			\bottomrule
		\end{tabular}
	\end{adjustbox}
	\label{tab:multi_gk_well_spec_step_size}
\end{table}

\paragraph{Misspecified case}
Due to the misspecified model, for certain values of $n$, the SG-MCMC algorithm resulted in proposal values that were outside our specified parameter range. For these cases, we manually tuned the step-size such that the SG-MCMC algorithm ran successfully. These cases are indicated in Table~\ref{tab:multi_gk_miss_spec_step_size} with an asterisk $(*)$.
\begin{itemize}
    \item The energy score posteriors were set to $w=1$.
    \item For the kernel score posteriors, in order to have coherent results with respect to the well specified case, we use here the values determined in the well-specified case. ($w=191$, $\gamma=45$)
    \item For the BSL posteriors, we set $\sigma = 1$. However, the chain was unable to converge for any $n > 5$, and so we ran the BSL posterior with an additional $n=5$ observations.
\end{itemize}

\begin{table}[h!]
	\centering
	\caption{Step-sizes for the two SR posteriors in the multivariate g-and-k model with misspecified observations, determined with the Multi-Armed Bandit algorithm of \cite{coullon2021efficient}.}
	\begin{adjustbox}{max width=\textwidth}

		\begin{tabular}{ l|cccccccc }
			\toprule
			\textbf{Observations} & $ n=1 $  & $ n=10 $ & $ n=20 $ & $ n=50 $ & $ n=70 $ & $ n=100 $ & $ n=200 $ & $ n=400 $  \\
			\midrule
			\textbf{Energy score} & $1 \times 10^{-1}$ & $1 \times 10^{-2}$ & $3 \times 10^{-3}$ & $1 \times 10^{-3}$ & $1 \times 10^{-3}$ & $1 \times 10^{-4}$ & $1 \times 10^{-4}$ & $3 \times 10^{-5}$  \\
			\textbf{Kernel score} & $5 \times 10^{-2}$ (*) & $1 \times 10^{-2}$ & $3 \times 10^{-3}$ & $3 \times 10^{-3}$ & $3 \times 10^{-3}$ & $3 \times 10^{-3}$ & $1 \times 10^{-4}$ (*) & $6 \times 10^{-5}$ (*)  \\
			\bottomrule
		\end{tabular}
	\end{adjustbox}
	\label{tab:multi_gk_miss_spec_step_size}
\end{table}

\subsection{Additional details on misspecified normal location model}\label{app:normal_location}

As mentioned in the main text (Sec.~\ref{sec:normal_location}), we set the weight $ w  $ such that the variance achieved by our SR posteriors is approximately the same as the one achieved by the standard Bayes distribution for the well specified case ($ \epsilon=0 $). This resulted in $ w=1 $ for the energy score posterior and $ w=2.8 $ for the kernel score posterior. Additionally, the bandwidth for the Gaussian kernel was tuned to be $ \gamma\approx 0.9566 $ (with the strategy discussed in Appendix~\ref{app:bandwidth}).

In Figure~\ref{fig:normal_loc_2} we report the full set of posterior distributions for the different values of $ \epsilon $ and $ z $ obtained with the standard Bayes posterior and with our SR posteriors.

\begin{figure}
	\includegraphics[width=1\linewidth]{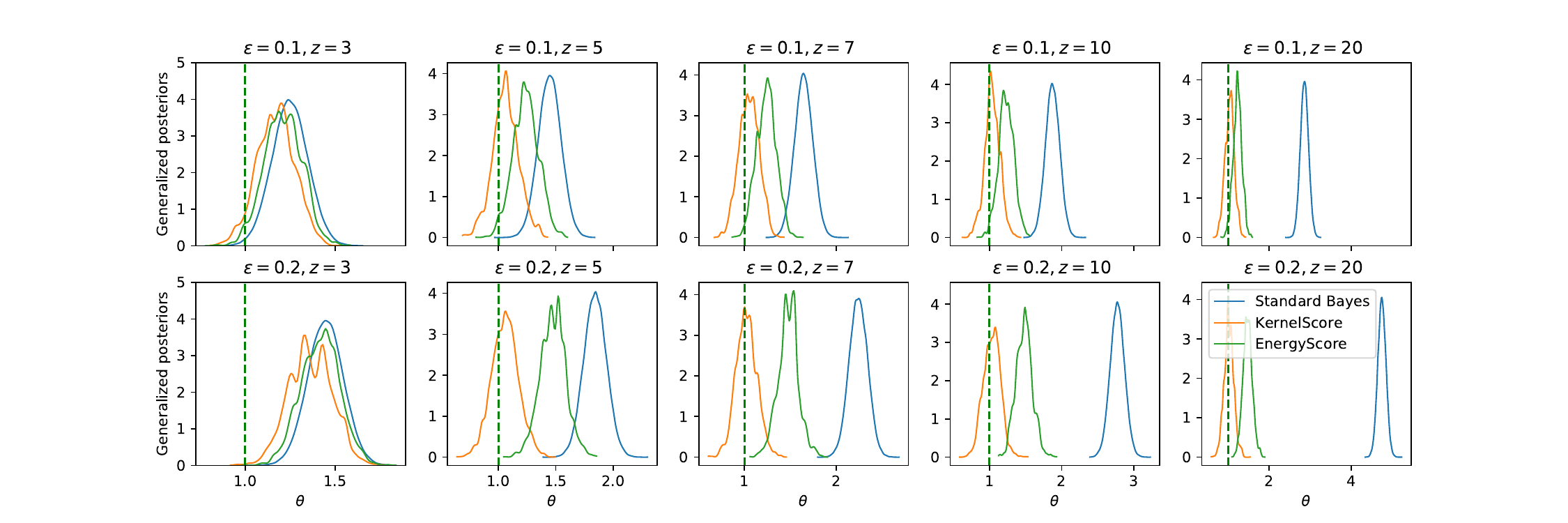}
	\caption{Posterior distribution obtained with the scoring rules and exact Bayes for the misspecified normal location model; each panel represents a different choice of $ \epsilon $ and $ z $. It can be seen that both Kernel and energy score are more robust with respect to Standard Bayes, with the kernel score one being extremely robust. The densities are obtained by KDE on the MCMC output thinned by a factor 10.}
	\label{fig:normal_loc_2}
\end{figure}

In the MCMC with the SR posteriors, a proposal size $ \sigma=2 $ is used for all values of $ \epsilon $ and $ z $. For all experiments, Table~\ref{tab:location_normal_acc_rate} reports acceptance rates obtained with the SR posteriors, while Table~\ref{tab:location_normal_variance} reports the obtained posterior standard deviation with SR posteriors and for the standard Bayes distribution (for which we do not give the proposal size and acceptance rate as it was sampled using more advanced MCMC techniques than standard Metropolis-Hastings using the \texttt{PyMC3} library \citep{salvatier2016probabilistic}).

\begin{table}[h!]
	\centering
	\caption{Acceptance rates for MCMC targeting the energy and kernel score posteriors for the different outlier setups, for the misspecified normal location model.}
	\begin{adjustbox}{max width=\textwidth}

		\begin{tabular}{ l|ccccccccccc }
			\toprule
			\multirow{2}{*}{\textbf{Setup}}  & \multicolumn{1}{c}{$\epsilon=0$}  &   \multicolumn{5}{c} {$\epsilon=0.1$}  &   \multicolumn{5}{c} {$\epsilon=0.2$} \\ \cmidrule(r){2-2}\cmidrule(r){3-7} \cmidrule(r){8-12}
			& -  & $ z=3 $ & $ z=5 $ & $ z=7 $ & $ z=10 $ & $ z=20 $ & $ z=3 $ & $ z=5 $ & $ z=7 $ & $ z=10 $ & $ z=20 $ \\
			\midrule
			\textbf{Kernel score} & 0.076 & $0.086$ & $0.089$ & $0.085$ & $0.086$ & $0.087$ & $0.080$ & $0.086$ & $0.089$ & $0.091$ & $0.090$ \\
			\textbf{Energy score} & 0.076 & $0.084$ & $0.087$ & $0.082$ & $0.083$ & $0.085$ & $0.082$ & $0.079$ & $0.077$ & $0.082$ & $0.082$ \\
			\bottomrule
		\end{tabular}
	\end{adjustbox}
	\label{tab:location_normal_acc_rate}

\end{table}

\begin{table}[h!]
	\centering
	\caption{Obtained posterior standard deviation for the standard Bayes and the energy and kernel score posteriors, for the different outlier setups, for the misspecified normal location model.}
	\begin{adjustbox}{max width=\textwidth}

		\begin{tabular}{ l|ccccccccccc }
			\toprule
			\multirow{2}{*}{\textbf{Setup}}  & \multicolumn{1}{c}{$\epsilon=0$}  &   \multicolumn{5}{c} {$\epsilon=0.1$}  &   \multicolumn{5}{c} {$\epsilon=0.2$} \\ \cmidrule(r){2-2}\cmidrule(r){3-7} \cmidrule(r){8-12}
			& -  & $ z=3 $ & $ z=5 $ & $ z=7 $ & $ z=10 $ & $ z=20 $ & $ z=3 $ & $ z=5 $ & $ z=7 $ & $ z=10 $ & $ z=20 $ \\
			\midrule
			\textbf{Standard Bayes}& $0.100$ & $0.100$ & $0.099$ & $0.099$ & $0.099$ & $0.100$ & $0.099$ & $0.099$ & $0.100$ & $0.099$ & $0.099$ \\
			\textbf{Kernel score} & $0.101$ & $0.106$ & $0.114$ & $0.108$ & $0.105$ & $0.113$ & $0.121$ & $0.116$ & $0.113$ & $0.117$ & $0.116$ \\
			\textbf{Energy score}  & $0.098 $ & $0.105$ & $0.112$ & $0.106$ & $0.106$ & $0.107$ & $0.109$ & $0.112$ & $0.111$ & $0.114$ & $0.113$ \\
			\bottomrule
		\end{tabular}
	\end{adjustbox}
	\label{tab:location_normal_variance}
\end{table}

Finally, as mentioned in the main text (Sec.~\ref{sec:normal_location}), we attempted using BSL in this scenario. As the model is Gaussian, we expected the BSL posterior to be very close to the standard posterior. Indeed, this is what we observed in the well specified case and for small $ z $ (Figure~\ref{fig:normal_loc_SL}). When however $ z $ is increased, the MCMC targeting the BSL posterior does not perform satisfactorily (see the trace plots in Figure~\ref{fig:normal_loc_SL_traces}). Neither reducing the proposal size nor running the chain for a longer number of steps seems to solve this issues.%

\begin{figure}
	\centering
	\includegraphics[width=0.6\linewidth]{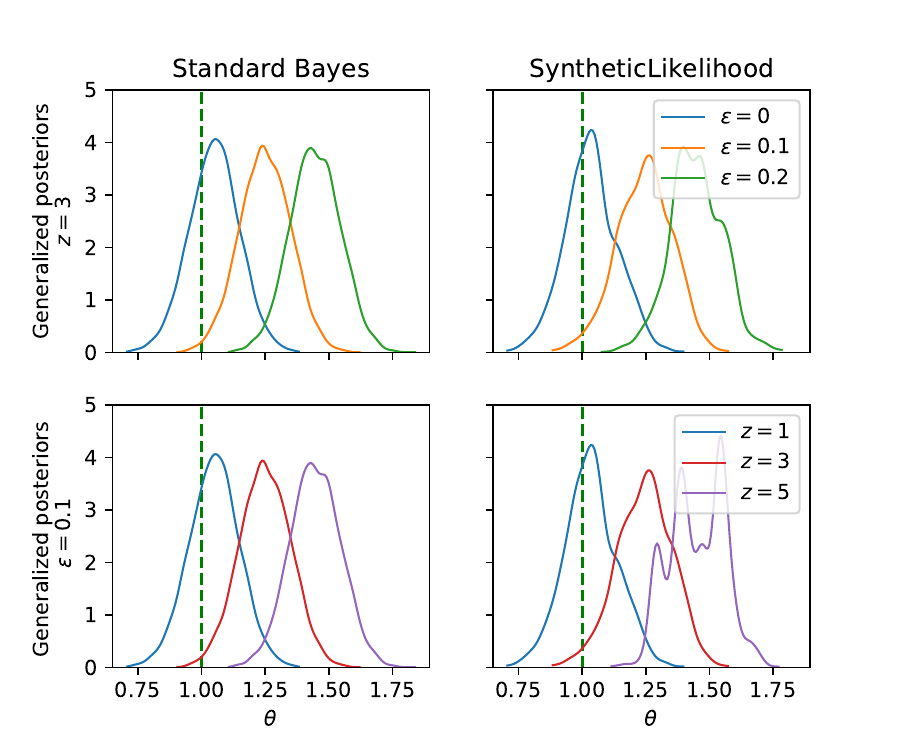}
	\caption{Standard Bayes and BSL posteriors for the normal location model, for different choices of $ \epsilon $ and $ z $.  First row: fixed outliers location $ z=3 $ and varying proportion $ \epsilon $; second row: fixed outlier proportion $ \epsilon $, varying location $ z $. As expected, the BSL posterior is very close to the standard Bayes posterior. The densities are obtained by KDE on the MCMC output thinned by a factor 10.}
	\label{fig:normal_loc_SL}
\end{figure}

\begin{figure}
	\centering
	\includegraphics[width=1\linewidth]{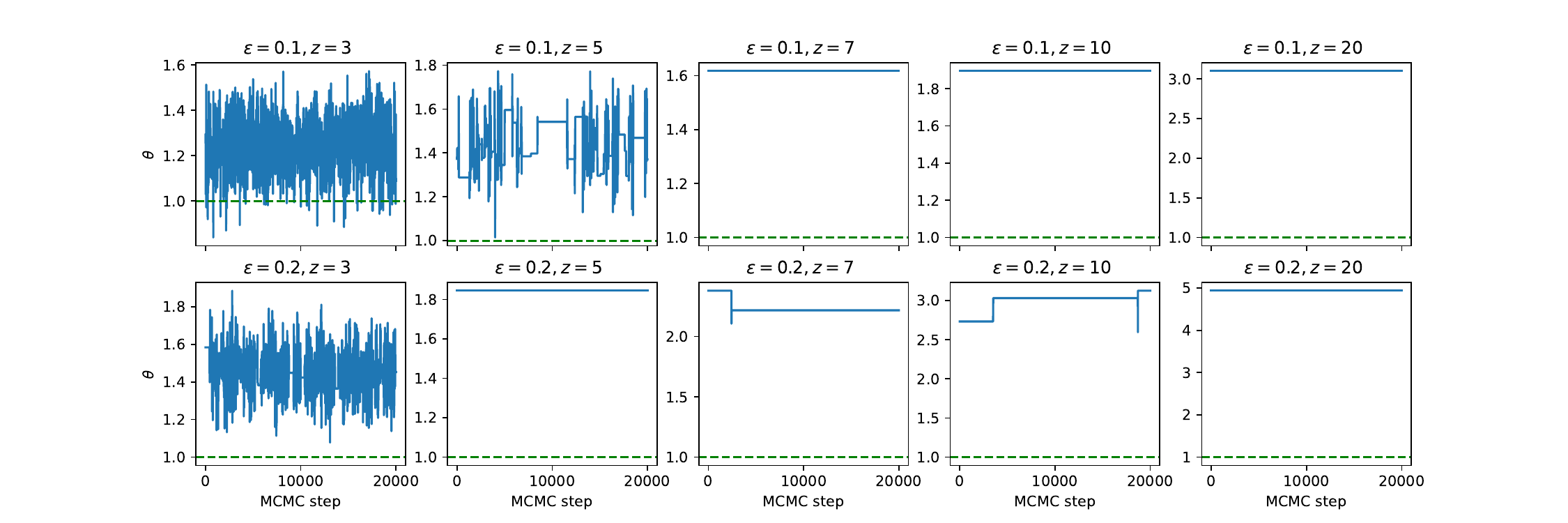}
	\caption{Trace plots for MCMC targeting the BSL posterior with different choices of $ z $ and $ \epsilon $, for the misspecified normal location model. We used here proposal size $ \sigma=2 $ and 60000 MCMC steps, of which 40000 were burned in; reducing the proposal size or increasing the number of steps did not seem to solve this issue.}
	\label{fig:normal_loc_SL_traces}
\end{figure}

\subsection{The Lorenz96 model}\label{app:Lorenz}
In both comparisons, we utilize the energy score posterior with $w=1$, except for the case where the SMC-ABC algorithm is used.

\paragraph{Comparison with ABC}
We ran the inference using the adSGLD algorithm, and the SMC-ABC algorithm, both with observations of $n=10$. For the energy score posterior, a step size of $\epsilon = 3 \times 10^{-2}$ was set, and the chain was initialised at a parameter value of $0$.

\paragraph{High dimensional neural stochastic parametrization}
We ran the inference using both the adSGLD algorithm with the linear stochastic parametrization and the pSGLD algorithm with the high dimensional neural parametrization, both with observations of $n=1$ which was first standardised. For both cases, chains were started from an initial optimization step of $250$ iterations ran with the Adam optimizer \cite{kingma2014adam}. For the adSGLD algorithm, a step size of $\epsilon = 1 \times 10^{-4}$ was set, while for the pSGLD algorithm this was set to $\epsilon = 1 \times 10^{-7}$.

\section{Results with pseudo-marginal MCMC on g-and-k model}\label{app:PM_results}

We report here some parallel results to those in the main text of the paper obtained with pseudo-marginal (PM) MCMC. To obtain these results, we use the correlated pseudo-marginal MCMC \citep{dahlin2015accelerating, deligiannidis2018correlated, picchini2022sequentially} mentioned in Sec.~\ref{sec:PM} with independent normal proposals on each component of the parameter space; we indicate by $ \sigma $ the standard deviation of the normal proposal distribution, which we report below. %
In all cases, whenever the parameter space is bounded, we run PM-MCMC on a transformed unbounded space obtained via a logistic transformation. Therefore, the proposal sizes refer to that unbounded space.

Besides our SR posteriors, we consider here the BSL and the semi-parametric BSL (Appendix~\ref{app:semiBSL}; notice that the latter is only well-defined for multivariate models). When performing these studies, we aimed at comparing the performance of our SR posteriors with BSL. Hence, we set the value of $w$ for the energy and kernel score posteriors with the strategy discussed in Sec.~\ref{sec:fix_w} using BSL as a reference.

\subsection{Well-specified setup}\label{app:gk_prop_size}\label{app:gk_PM}

For both univariate and multivariate case, we consider synthetic observations generated from parameter values $ A^\star =3,$ $ B^\star=1.5,$ $ g^\star=0.5,$ $ k^\star=1.5$ and $ \rho^\star=-0.3 $ (notice $ \rho $ is not used in the univariate case).

We first present results and discuss specific settings below.
For the univariate g-and-k, Fig.~\ref{fig:uni_gk} reports the marginal posterior distributions for each parameter at different number of observations for the considered methods. With increasing $ n $, the BSL posterior does not concentrate (except for the parameter $ k $); the energy score posterior concentrates close to the true value for all parameters (green vertical line), while the kernel score posterior performs slightly worse, not being able to concentrate for the parameter $ g $ (albeit this may happen with an even larger $ n $, which we did not consider here). The poor performance of BSL is due to violation of the underlying normality assumption (which is to say, the scoring rule used by BSL is not strictly proper for this example), while the concentration of the energy and kernel score posteriors are in line with them being strictly proper SRs.

\begin{figure}[tb]
    \includegraphics[width=1.0\linewidth]{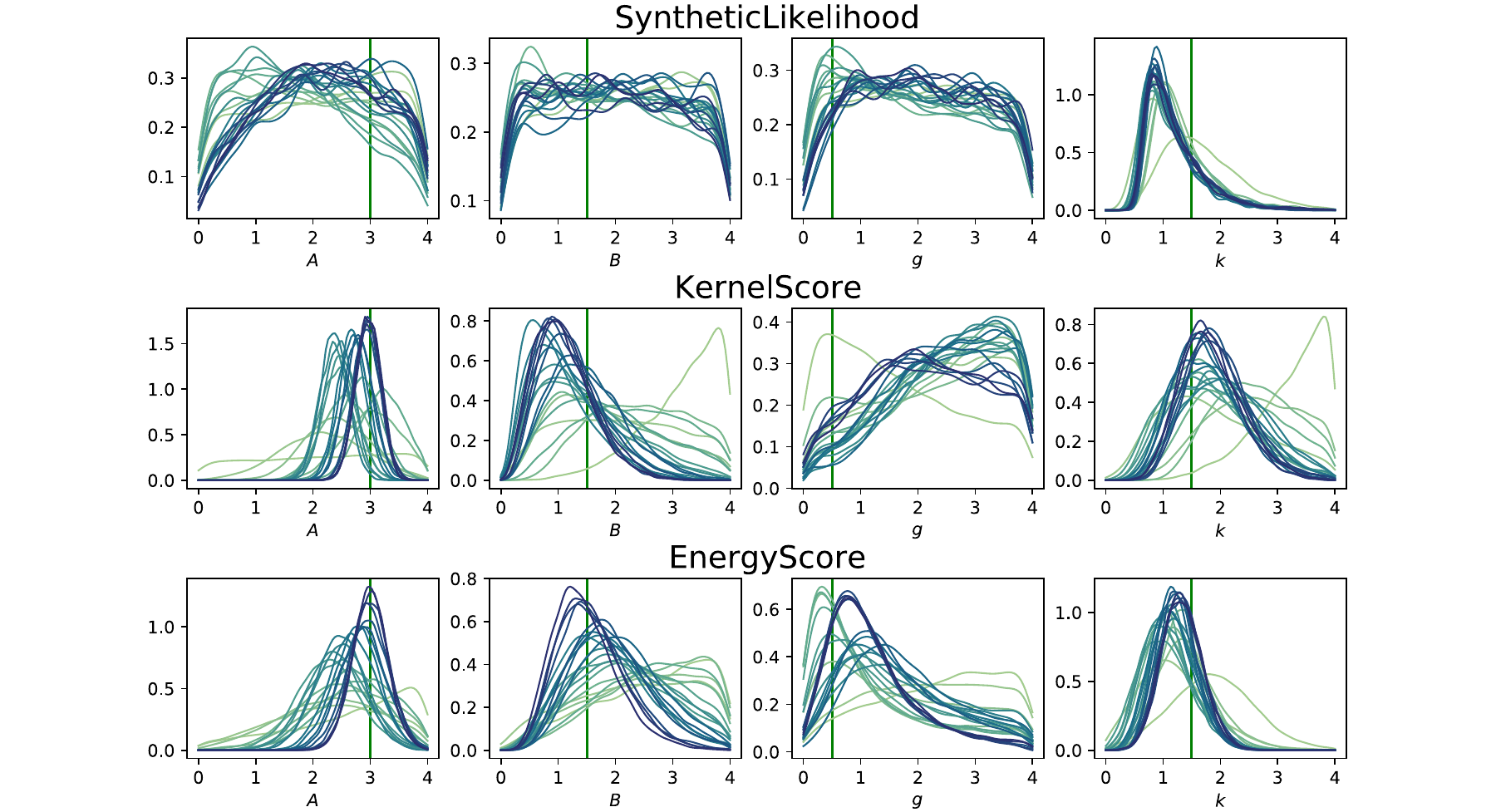}
    \caption{Marginal posterior distributions for the different parameters for the well-specified univariate g-and-k model, with increasing number of observations ($n= 1, 5, 10, 15, \ldots, 100) $, with PM-MCMC. Darker (respectively lighter) colors denote a larger (smaller) number of observations. The densities are obtained by KDE on the MCMC output thinned by a factor 10. The energy and kernel score posteriors concentrate around the true parameter value (green vertical line), while BSL does not.}
    \label{fig:uni_gk}
\end{figure}

Similar results for the multivariate g-and-k are reported in Fig.~\ref{fig:gk}. 
For this example, the PM-MCMCs targeting the semiBSL and BSL posteriors do not converge beyond respectively 1 and 10 observations;
instead, with the Kernel and energy scores we do not experience such a problem. The energy score concentrates well on the exact parameter value in this case too, while the kernel score is able to concentrate well for some parameters ($ g $ and $ k $) and some concentration can be observed for $ \rho $; however, the kernel score posterior marginals for $ A $ and $ B $ are flatter and noisier (it may be that larger $ n $ leads to more concentrate posterior for $ A $ and $ B $ as well, but we did not research this further).

\begin{figure}[tb]
    \includegraphics[width=1\linewidth]{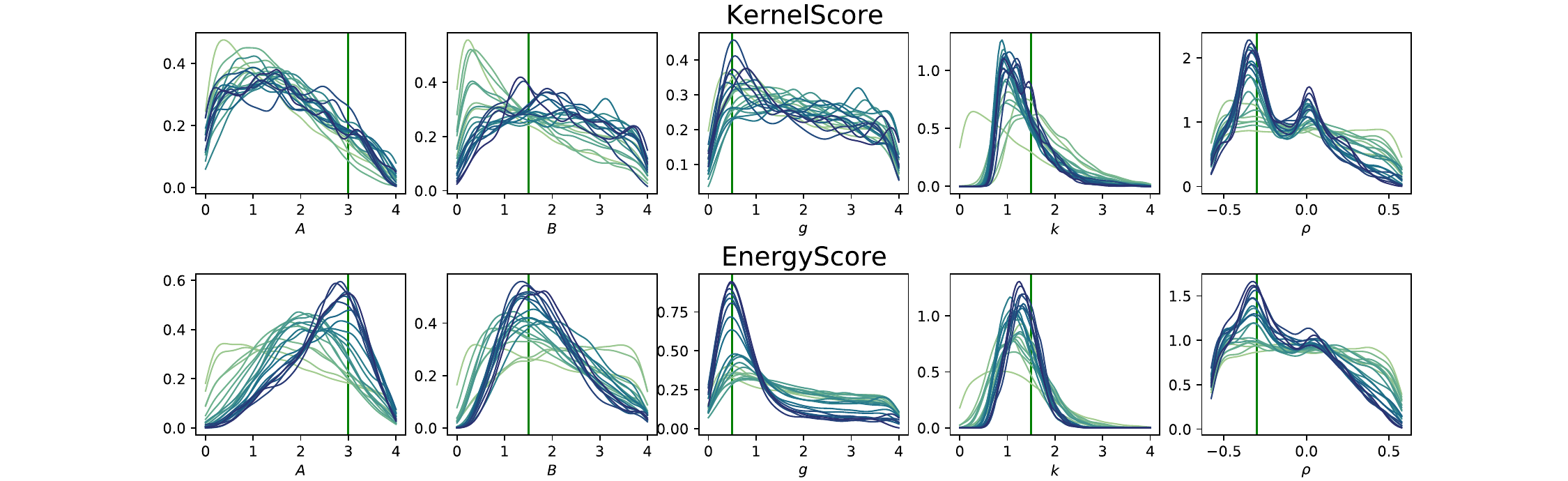}
	\includegraphics[width=1\linewidth]{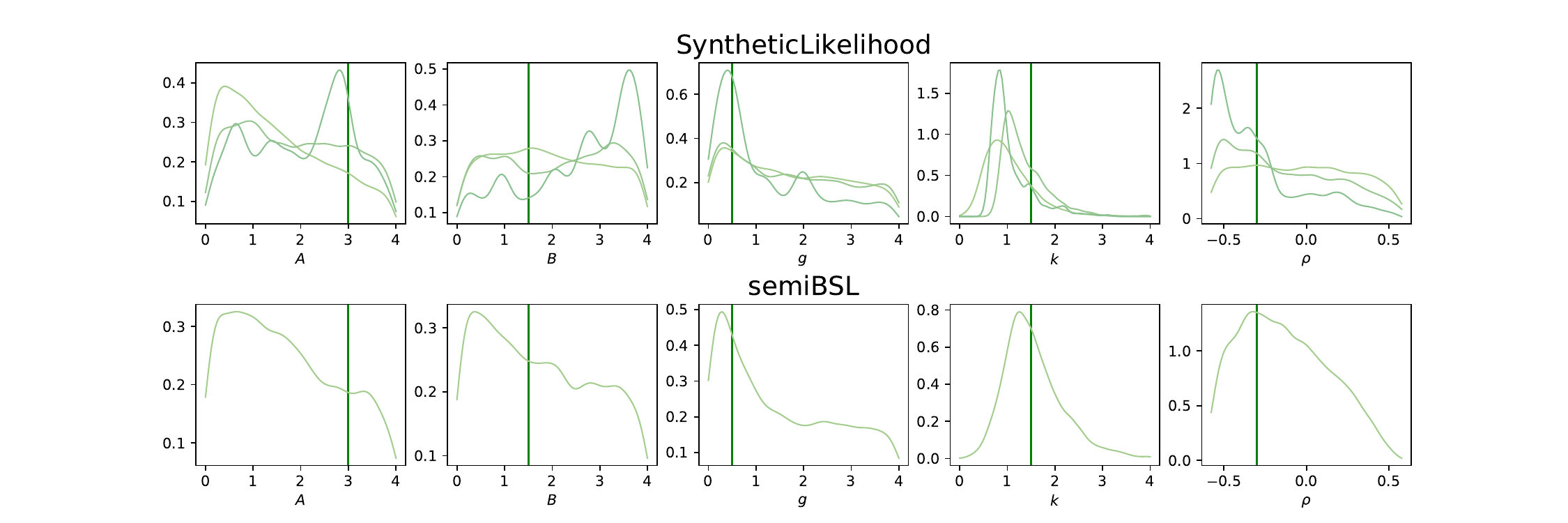}
    \caption{Marginal posterior distributions for the different parameters for the  well-specified multivariate g-and-k model, with increasing number of observations ($n= 1, 5, 10, 15, \ldots, 100) $, with PM-MCMC. Darker (respectively lighter) colors denote a larger (smaller) number of observations. The densities are obtained by KDE on the MCMC output thinned by a factor 10. The energy score posterior concentrates well around the true parameter value (green vertical line), with the kernel score one performing slightly worse. For BSL, we were able to run the inference for $ n=1,5,10 $, while we were only able to do so for $ n=1 $ for semiBSL.}
    \label{fig:gk}
\end{figure}

We use the following settings for the SR posteriors:

\begin{itemize}
	\item For the energy score posterior, our heuristic procedure (Sec.~\ref{sec:fix_w}) for setting $ w $ using BSL as a reference resulted in $ w\approx 0.35 $ for the univariate model and $ w\approx 0.16$ for the multivariate one.
	\item For the kernel score posterior, we first fit the value of the Gaussian kernel bandwidth parameter as described in Appendix~\ref{app:bandwidth}, which resulted in $ \gamma\approx 5.50 $ for the univariate case and $ \gamma \approx 52.37$ for the multivariate one. Then, the heuristic procedure for $ w $ using BSL as a reference resulted in $ w\approx18.30 $ for the univariate model and $ w\approx 52.29 $ for the multivariate one.
\end{itemize}

Next, we discuss the proposal sizes for PM-MCMC; recall that we use independent normal proposals on each component of $ \theta $, with standard deviation $ \sigma $. We report here the values for $ \sigma $ used in the experiments; we stress that, as the PM-MCMC is run in the transformed unbounded parameter space (obtained applying a logit transformation), these proposal sizes refer to that space.

For the univariate g-and-k, the proposal sizes we use are the following:
\begin{itemize}
	\item For BSL, we use $ \sigma =1$ for all values of $ n $.
	\item For energy and kernel scores, we take $ \sigma=1 $ for $ n $ from 1 up to 25 (included), $ \sigma=0.4 $ for $ n $ from 30 to 50, and $ \sigma=0.2 $ for $ n $ from 55 to 100.
\end{itemize}

For the multivariate g-and-k:
\begin{itemize}
	\item For BSL and semiBSL, we use $ \sigma =1$ for all values of $ n $ for which the chain converges. We stress that we tried decreasing the proposal size, but that did not solve the non-convergence issue (discussed in the main text in Sec.~\ref{sec:experiments_gk_well_spec}).
	\item For energy and kernel scores, we take $ \sigma=1 $ for $ n $ from 1 up to 15 (included), $ \sigma=0.4 $ for $ n $ from 20 to 35, $ \sigma=0.2 $ for $ n $ from 40 to 50 and $ \sigma=0.1 $ for $ n $ from 55 to 100.
\end{itemize}

In Table~\ref{Tab:acc_rates_gk}, we report the acceptance rates the different methods achieve for all values of $ n $, with the proposal sizes mentioned above. We denote by ``/'' the experiments for which we did not manage to run PM-MCMC satisfactorily. We remark how the energy score achieves a larger acceptance rates in all experiments compared to the kernel score.

\begin{table}
	\centering
		\caption{Acceptance rates for the univariate and multivariate g-and-k experiments with different values of $ n $, with the PM-MCMC proposal sizes reported in Appendix~\ref{app:gk_prop_size}. ``/'' denotes experiments for which PM-MCMC did not run satisfactorily.}
	\begin{tabular}{ c|ccccccc }
		\toprule
		\multirow{2}{*}{\textbf{N. obs.}}  & \multicolumn{3}{c}{Univariate g-and-k}  &   \multicolumn{4}{c} {Multivariate g-and-k} \\ \cmidrule(r){2-4} \cmidrule(r){5-8}
		& BSL & Kernel score & Energy score & BSL & semiBSL & Kernel score & Energy score \\
		\midrule

		1         & $0.362$ & $0.507$ & $0.420$     & $0.216$ & $0.190$ & $0.468$ & $0.445$ \\
		5         & $0.221$ & $0.329$ & $0.375$     & $0.069$ & / & $0.136$ & $0.224$ \\
		10        & $0.133$ & $0.252$ & $0.272$     & $0.036$ & / & $0.127$ & $0.216$ \\
		15        & $0.109$ & $0.253$ & $0.217$     & / & / & $0.077$ & $0.154$ \\
		20        & $0.100$ & $0.154$ & $0.207$     & / & / & $0.151$ & $0.278$ \\
		25        & $0.092$ & $0.149$ & $0.208$     & / & / & $0.126$ & $0.233$ \\
		30        & $0.085$ & $0.218$ & $0.343$     & / & / & $0.124$ & $0.222$ \\
		35        & $0.080$ & $0.172$ & $0.315$     & / & / & $0.076$ & $0.166$ \\
		40        & $0.076$ & $0.152$ & $0.293$     & / & / & $0.119$ & $0.246$ \\
		45        & $0.070$ & $0.130$ & $0.256$     & / & / & $0.103$ & $0.223$ \\
		50        & $0.062$ & $0.121$ & $0.220$     & / & / & $0.103$ & $0.219$ \\
		55        & $0.060$ & $0.189$ & $0.317$     & / & / & $0.139$ & $0.297$ \\
		60        & $0.059$ & $0.185$ & $0.324$     & / & / & $0.129$ & $0.286$ \\
		65        & $0.057$ & $0.173$ & $0.314$     & / & / & $0.133$ & $0.273$ \\
		70        & $0.052$ & $0.172$ & $0.289$     & / & / & $0.119$ & $0.256$ \\
		75        & $0.048$ & $0.161$ & $0.273$     & / & / & $0.123$ & $0.247$ \\
		80        & $0.048$ & $0.159$ & $0.267$     & / & / & $0.117$ & $0.233$ \\
		85        & $0.045$ & $0.150$ & $0.252$     & / & / & $0.098$ & $0.213$ \\
		90        & $0.044$ & $0.143$ & $0.247$     & / & / & $0.087$ & $0.198$ \\
		95        & $0.044$ & $0.136$ & $0.244$     & / & / & $0.089$ & $0.198$ \\
		100       & $0.042$ & $0.129$ & $0.236$     & / & / & $0.076$ & $0.190$ \\
		\bottomrule
	\end{tabular}
	\label{Tab:acc_rates_gk}
\end{table}

\subsubsection{Investigating the poor PM-MCMC performance for BSL and semiBSL}\label{app:gk_nsim}

The correlated pseudo-marginal MCMC for BSL and semiBSL performed poorly for the multivariate g-and-k example, not being able to converge when using more than respectively 1 and 10 observations
We investigate now this poor performance, by fixing $ n=20 $ and running PM-MCMC with 10 different initializations,  for 10000 MCMC steps with no burn-in, for BSL and semiBSL, with $ m=500 $. The chains look ``sticky'' and, after a short transient, get stuck in different regions of $ \Theta  $ (see Fig.~\ref{fig:gk_traces}).

\begin{figure}[t]
	\includegraphics[width=1\linewidth]{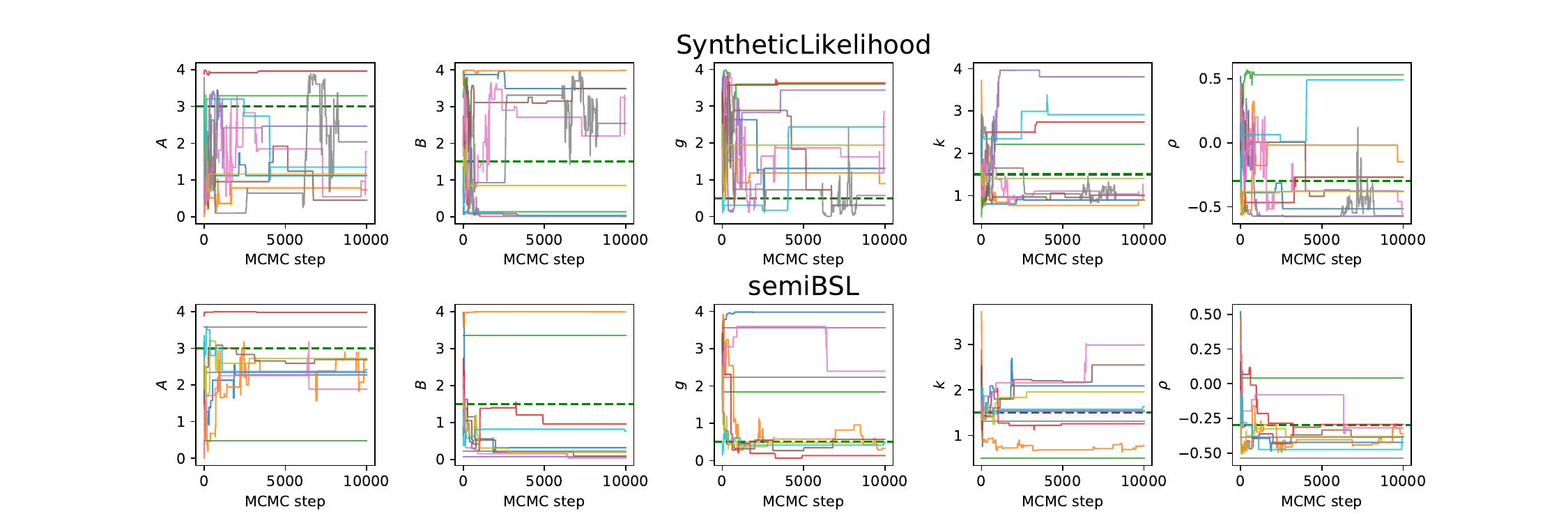}
	\caption{Traceplots for semiBSL and BSL for $ n=20 $ for 10 different initializations (different colors), with 10000 PM-MCMC steps (no burn-in); the green dashed line denotes the true parameter value. It can be seen that the chains are very sticky, and that they explore different parts of the parameter space. }
	\label{fig:gk_traces}
\end{figure}

In order to understand the reason for this result, we investigate whether the poor performance is due to large variance in the estimate of the target; as increasing the number of simulations reduces such variance, we study the effect of this on the PM-MCMC performance. Therefore, we report here the results of a study increasing the number of simulations for a fixed number of observations $ n=20 $ for the g-and-k model. Specifically, we tested $ m=500,1000,1500,2000,2500,3000,30000 $; as discussed in Appendix~\ref{app:gk_prop_size}, we used a proposal size $ \sigma=0.4 $, with which the energy and kernel score posteriors performed well. We report traceplots in Fig.~\ref{fig:gk_traces_increase_sim} and corresponding acceptance rates in Table~\ref{tab:acc_rates_gk_increase_sim}; from this experiment, we note that BSL achieves acceptance rate as large as few percentage points with larger $ m $ values, but there is no constant trend (for instance, acceptance rate with $ m=3000 $ is smaller than with $ m=2000 $), which means that the method is still prone to getting stuck. For semiBSL, the acceptance rate is abysmal even for very large $ m $.

Additionally, while the BSL assumptions are unreasonable for this model, the multivariate g-and-k fulfills the assumptions underlying semiBSL: in fact, applying a one-to-one transformation to each component of a random vector does not change the copula structure, which is Gaussian in this case. It is therefore surprising that the performance of semiBSL degrades so rapidly when $ n $ increases.

\begin{figure}
	\includegraphics[width=1\linewidth]{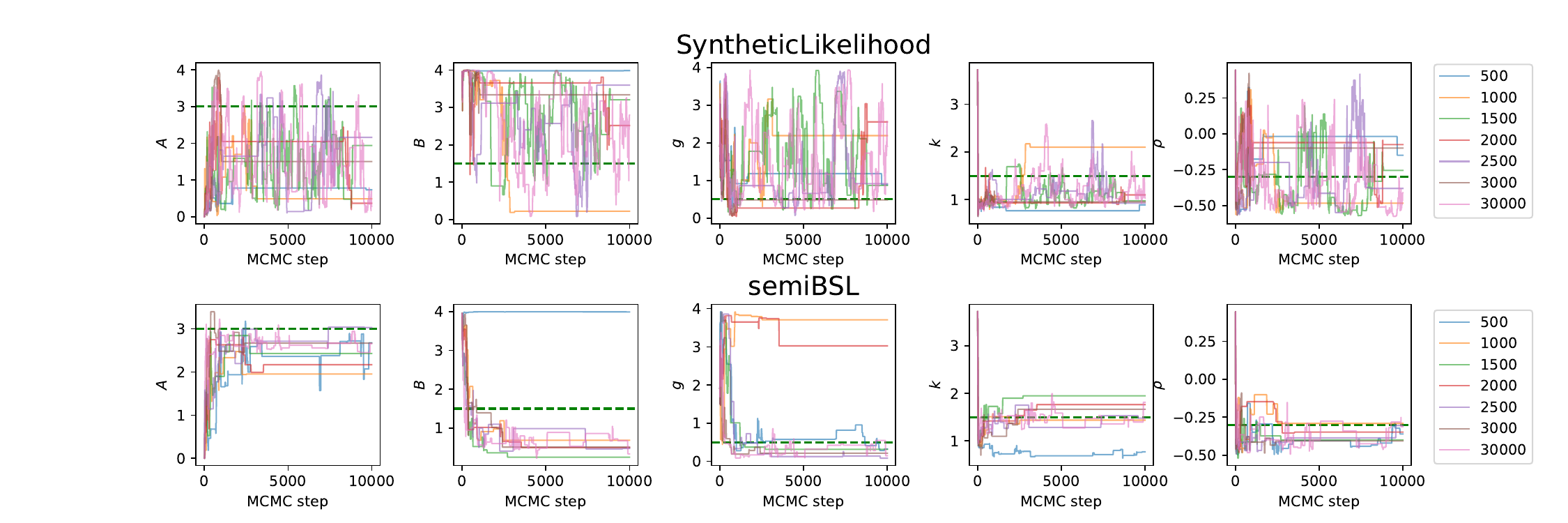}
	\caption{Traceplots for BSL and semiBSL and BSL for $ n=20 $ using different number of simulations $ m $, reported in the legend for each row; green dashed line denotes the true parameter value. There is no improvement in the mixing of the chain for increasing the number of simulations. }
	\label{fig:gk_traces_increase_sim}
\end{figure}

\begin{table}
	\centering
		\caption{Acceptance rates for BSL and semiBSL and BSL for $ n=20 $ using different number of simulations $ m $; there is no improvement in the acceptance rate for increasing number of simulations. We recall that we were not able to run semiBSL for $ m=30000 $ due to its high computational cost.}
	\begin{adjustbox}{max width=\textwidth}
		\begin{tabular}{l|ccccccc}
			\toprule
			\textbf{N. simulations $ \boldsymbol m $} & 500 & 1000 & 1500 & 2000 & 2500 & 3000 & 30000 \\
			\midrule
			\textbf{Acc. rate BSL}  & $6.0\cdot 10^{-3}$ & $1.1\cdot 10^{-2}$ & $3.3\cdot 10^{-2}$ & $9.9\cdot 10^{-3}$ & $1.8\cdot 10^{-2}$ & $7.5\cdot 10^{-3}$ & $5.1\cdot 10^{-2}$ \\
			\midrule
			\textbf{Acc. rate semiBSL} & $7.0\cdot 10^{-3}$ & $3.4\cdot 10^{-3}$ & $3.7\cdot 10^{-3}$ & $2.8\cdot 10^{-3}$ & $4.2\cdot 10^{-3}$ & $3.6\cdot 10^{-3}$ & $9.2\cdot 10^{-3}$  \\
			\bottomrule
		\end{tabular}
	\end{adjustbox}
	\label{tab:acc_rates_gk_increase_sim}
\end{table}

\subsection{Misspecified setup}\label{app:gk_Cauchy_PM}

The observations are here generated by a Cauchy distribution. For the univariate case, the univariate Cauchy is used; for the multivariate case, the observations are generated as in Sec.~\ref{sec:experiments_gk_misspec} (i.e., no correlation between components).

In order to have coherent results with respect to the well specified case, we use here the values of $ w $ and $ \gamma $ determined in the well specified case (reported in Appendix~\ref{app:gk_prop_size})

For the univariate g-and-k, we report the marginal posteriors in Fig.~\ref{fig:uni_gk_Cauchy}. The energy and kernel score posteriors concentrate on a similar parameter value; the BSL posterior concentrates as well (differently from the well-specified case), albeit on a slightly different parameter value (especially for $ B $ and $ k $). Therefore, with this kind of misspecification, $ \thetastar $ is unique both when using the strictly proper Kernel and energy scores, as well as the non-strictly proper Dawid--Sebastiani Score (corresponding to BSL).

\begin{figure}[tb]
    \includegraphics[width=1.0\linewidth]{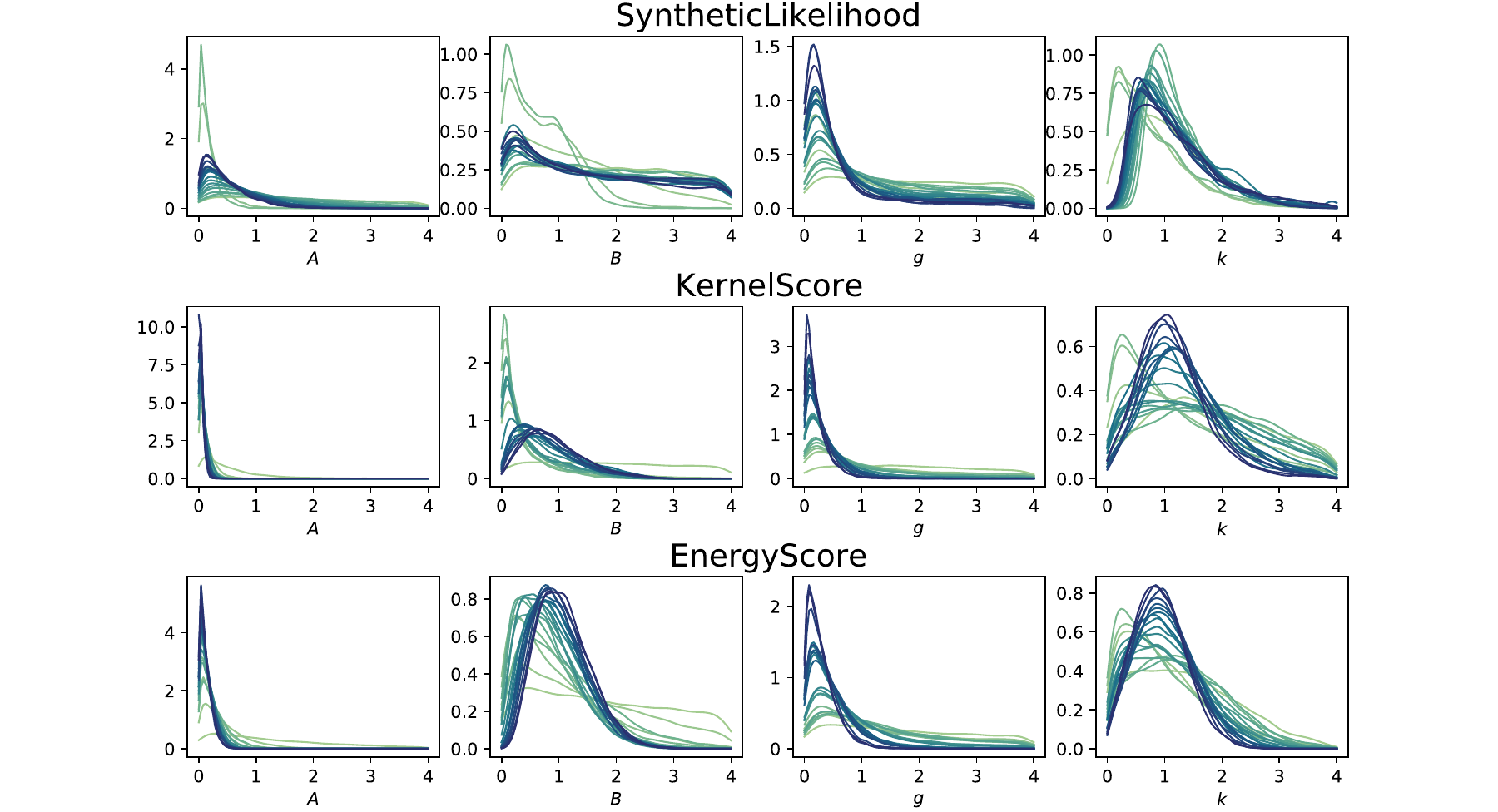}
    \caption{Marginal posterior distributions for the different parameters for the univariate g-and-k model, with increasing number of observations ($n= 1, 5, 10, 15, \ldots, 100) $ generated from the Cauchy distribution, with PM-MCMC. Darker (respectively lighter) colors denote a larger (smaller) number of observations. The densities are obtained by KDE on the MCMC output thinned by a factor 10. The energy and kernel score posteriors concentrate around the same parameter value, while BSL concentrates on slightly different one (specially for $ B $ and $ k $).}
    \label{fig:uni_gk_Cauchy}
\end{figure}

For the multivariate g-and-k, we experienced the same issue with PM-MCMC as in the well-specified case for BSL and semiBSL; therefore, we do not report those results. Marginals for the energy and kernel score posteriors can be seen in Fig.~\ref{fig:gk_Cauchy}; both posteriors concentrate for all parameters except for $ \rho $ (which describes correlation among different components in the observations, here absent). For the other parameters, the two methods concentrate on very similar parameter values, with slightly larger difference for $ k $, for which the kernel score posterior does not concentrate very well.

\begin{figure}[tb]
    \includegraphics[width=1\linewidth]{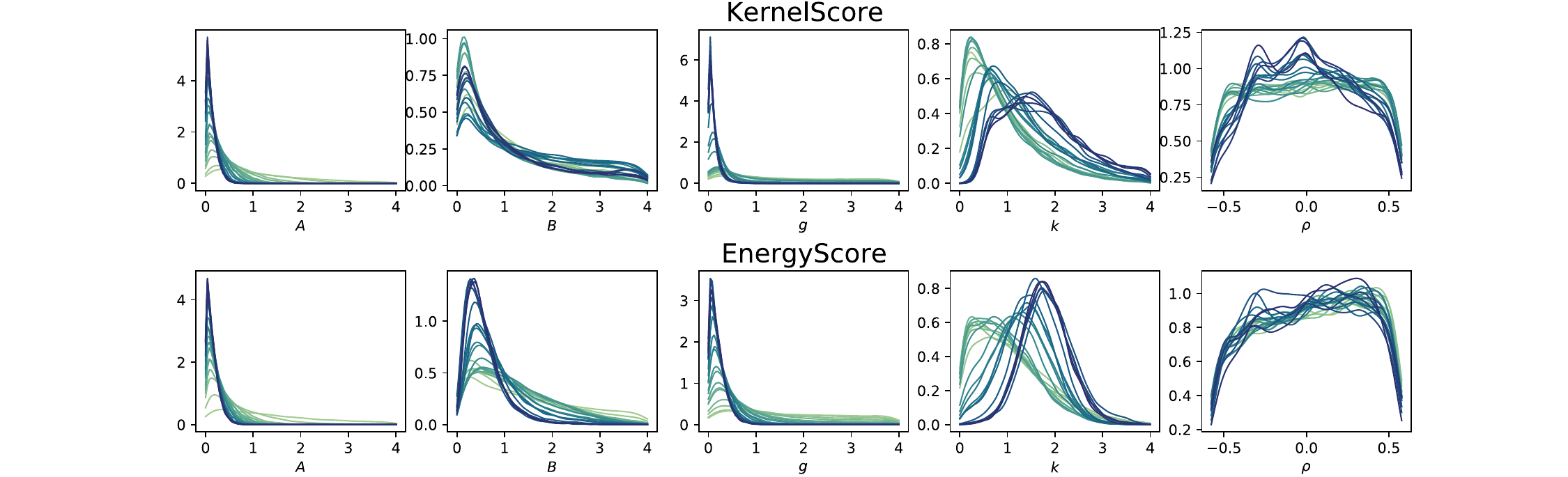}
    \caption{		Marginal posterior distributions for the different parameters for the multivariate g-and-k model, with increasing number of observations ($n= 1, 5, 10, 15, \ldots, 100) $ generated from the Cauchy distribution, with PM-MCMC. Darker (respectively lighter) colors denote a larger (smaller) number of observations The densities are obtained by KDE on the MCMC output thinned by a factor 10. Both energy and kernel score posteriors
        concentrate on a very similar parameter value, with slightly larger difference for $ k $.}
    \label{fig:gk_Cauchy}
\end{figure}

The above resuts are obtained with the following proposal sizes for PM-MCMC (which is run with independent normal proposals on each component of $ \theta $ with standard deviation $ \sigma $, in the same way as in the well specified case, after applying a logit transformation to the parameter space).
\begin{itemize}
	\item For the univariate g-and-k, for all methods (BSL, energy and kernel scores), we take $ \sigma=1 $ for $ n $ from 1 up to 25 (included), $ \sigma=0.4 $ for $ n $ from 30 to 50, and $ \sigma=0.2 $ for $ n $ from 55 to 100.
	\item For the multivariate g-and-k, recall that we did not report results for BSL and semiBSL here as we were not able to sample the posteriors with PM-MCMC for large $ n $, as already experienced in the well specified case. For the remaining techniques, we used the same values of $ \sigma  $ as in the well specified experiments (Appendix~\ref{app:gk_Cauchy_prop_size}).
\end{itemize}

In Table~\ref{Tab:acc_rates_gk_Cauchy}, we report the acceptance rates the different methods achieve for all values of $ n $, with the proposal sizes discussed above. We remark how the energy score achieves a larger acceptance rates in all experiments compared to the kernel score.

\begin{table}[h!]
	\centering
		\caption{Acceptance rates for the misspecified univariate and multivariate g-and-k experiments with different values of $ n $, with the PM-MCMC proposal sizes reported in Appendix~\ref{app:gk_Cauchy_PM}. }
	\begin{tabular}{ c|ccccc }
		\toprule
		\multirow{2}{*}{\textbf{N. obs.}}  & \multicolumn{3}{c}{Misspecified univariate g-and-k}  &   \multicolumn{2}{c} {Misspecified multivariate g-and-k} \\ \cmidrule(r){2-4} \cmidrule(r){5-6}
		& BSL & Kernel score & Energy score & Kernel score & Energy score \\
		\midrule
		1      & $0.457$   & $0.482$ & $0.521$       & $0.472$ & $0.470$ \\
		5      & $0.302$   & $0.436$ & $0.454$       & $0.324$ & $0.373$ \\
		10     & $0.193$   & $0.450$ & $0.425$       & $0.362$ & $0.330$ \\
		15     & $0.146$   & $0.441$ & $0.390$       & $0.361$ & $0.276$ \\
		20     & $0.102$   & $0.264$ & $0.311$       & $0.544$ & $0.410$ \\
		25     & $0.093$   & $0.288$ & $0.314$       & $0.530$ & $0.377$ \\
		30     & $0.153$   & $0.426$ & $0.471$       & $0.536$ & $0.359$ \\
		35     & $0.144$   & $0.349$ & $0.448$       & $0.537$ & $0.336$ \\
		40     & $0.134$   & $0.340$ & $0.440$       & $0.631$ & $0.432$ \\
		45     & $0.130$   & $0.344$ & $0.429$       & $0.523$ & $0.373$ \\
		50     & $0.125$   & $0.255$ & $0.393$       & $0.383$ & $0.343$ \\
		55     & $0.167$   & $0.318$ & $0.501$       & $0.471$ & $0.436$ \\
		60     & $0.176$   & $0.303$ & $0.490$       & $0.412$ & $0.407$ \\
		65     & $0.164$   & $0.293$ & $0.481$       & $0.389$ & $0.391$ \\
		70     & $0.164$   & $0.276$ & $0.455$       & $0.372$ & $0.374$ \\
		75     & $0.156$   & $0.272$ & $0.445$       & $0.278$ & $0.329$ \\
		80     & $0.157$   & $0.262$ & $0.436$       & $0.232$ & $0.306$ \\
		85     & $0.153$   & $0.254$ & $0.430$       & $0.247$ & $0.300$ \\
		90     & $0.147$   & $0.231$ & $0.415$       & $0.239$ & $0.299$ \\
		95     & $0.152$   & $0.226$ & $0.410$       & $0.235$ & $0.291$ \\
		100    & $0.141$   & $0.223$ & $0.407$       & $0.232$ & $0.277$ \\
		\bottomrule
	\end{tabular}
	\label{Tab:acc_rates_gk_Cauchy}
\end{table}

\section{Effect of $ m $ on pseudo-marginal MCMC}\label{app:different_n_sims}

Here, we consider the univariate and multivariate g-and-k, both well specified and misspecified, %
and study the impact of varying $ m $ in the resulting PM-MCMC target. As we span from very small to large values of $ m $, we use here the vanilla pseudo-marginal MCMC of \cite{andrieu2009pseudo} instead of the correlated pseudo-marginal MCMC which was used for all other simulations.

The choice of $ m $ has two different impacts on the PM-MCMC:
\begin{enumerate}
	\item first, it changes the pseudo-marginal MCMC target, as discussed in Section~\ref{sec:simulator_models} in the main text; recall how, there, we proved that, for $ m\to\infty $, the pseudo-marginal MCMC target converges to the original SR posterior defined in Eq.~\eqref{Eq:SR_posterior} in the main text. Therefore, we expect, for large enough $ m $, the pseudo-marginal MCMC target to be roughly constant.
	\item  Additionally, smaller values of $ m $ imply that the target estimate has a larger variance. Therefore, we expect sampling to be harder for small $ m $, in terms of acceptance rate of the MCMC, and easier for large $ m $ (albeit that is more computationally intensive).
\end{enumerate}

In our simulation study below, we consider $ m $ values from 10 to 1000. Our results empirically verify our expectations above. In particular, we find that, for $ m $ larger than a threshold which is typically few hundreds, the pseudo-marginal MCMC target is roughly constant. Additionally, very small values of $ m $ (few tens) make sampling impractical.

Moreover, our empirical results suggest that larger values of $ m $ are required for the PM-MCMC for semiBSL to be stable. For the other methods, the required $ m $ seem to be fairly similar, with slightly larger values for BSL for some models.

Typically, we found $ m  $ values in the few hundreds to strike a good balance between larger computational cost and improved acceptance rate with larger $ m $. Additionally, this consideration depends also on how quickly the simulation cost scales with $ m $: even when not parallelizing model simulations across different processors, if the implementation is vectorized, the computational cost can scale sub-linearly in $ m $, which means a better PM-MCMC efficiency is reached for a larger $ m $.
A more extensive study considering for instance the effective sample size per CPU time could be carried out.

In all experiments, except where said otherwise, we use the value of $ w $ found via our heuristics strategy (Section~\ref{sec:fix_w} in the main text) and reported above.

\FloatBarrier
\subsection{Univariate g-and-k}
\FloatBarrier

Here, we report results considering $ n=10 $ observations.

\begin{table}[h!]
	\centering
		\caption{Acceptance rate and trace of the posterior covariance matrix for different values of $ m $ for the well specified univariate g-and-k, for the BSL, Kernel and energy score posteriors. }
	\begin{adjustbox}{max width=\textwidth}
		\begin{tabular}{c|cccccc}
			\toprule
			\multirow{2}{*}{\textbf{$ m $}} &  \multicolumn{2}{c}{\textbf{BSL}} & \multicolumn{2}{c}{\textbf{Kernel score}} & \multicolumn{2}{c}{\textbf{Energy score}} \\
			\cmidrule(r){2-3} \cmidrule(r){4-5} \cmidrule(r){6-7}
			&   Acc. rate & $ \text{Tr}\left[\Sigma_{\text{post}}\right] $  &       Acc. rate &  $ \text{Tr}\left[\Sigma_{\text{post}}\right] $ &        Acc. rate &  $ \text{Tr}\left[\Sigma_{\text{post}}\right] $ \\
			\midrule
			10 &      0.104 &                4.5245 &      0.011 &                3.6030  &      0.063 &                3.9822 \\
			20 &      0.122 &                4.4439 &      0.035 &                3.6679 &      0.115 &                3.9642 \\
			50 &      0.129 &                4.3778 &      0.098 &                3.3803 &      0.179 &                3.6105 \\
			100 &      0.134 &                4.4095 &      0.157 &                3.2220  &      0.219 &                3.5335 \\
			200 &      0.136 &                4.1753 &      0.204 &                3.1628 &      0.243 &                3.4730  \\
			300 &      0.135 &                4.2261 &      0.220  &                3.1181 &      0.252 &                3.3537 \\
			400 &      0.135 &                4.1769 &      0.229 &                3.0716 &      0.257 &                3.3553 \\
			500 &      0.132 &                4.1702 &      0.234 &                3.1079 &      0.262 &                3.4362 \\
			600 &      0.130  &                4.2095 &      0.239 &                3.0295 &      0.259 &                3.2612 \\
			700 &      0.133 &                4.2417 &      0.243 &                3.0536 &      0.265 &                3.3629 \\
			800 &      0.132 &                4.2421 &      0.247 &                3.0216 &      0.265 &                3.3077 \\
			900 &      0.132 &                4.1084 &      0.248 &                3.0477 &      0.267 &                3.3815 \\
			1000 &      0.137 &                4.2930  &      0.253 &                3.1181 &      0.269 &                3.3570  \\
			\bottomrule
		\end{tabular}
	\end{adjustbox}
	\label{Tab:uni_gk_m}
\end{table}

\begin{figure}[tb]
	\centering
	\includegraphics[width=1\linewidth]{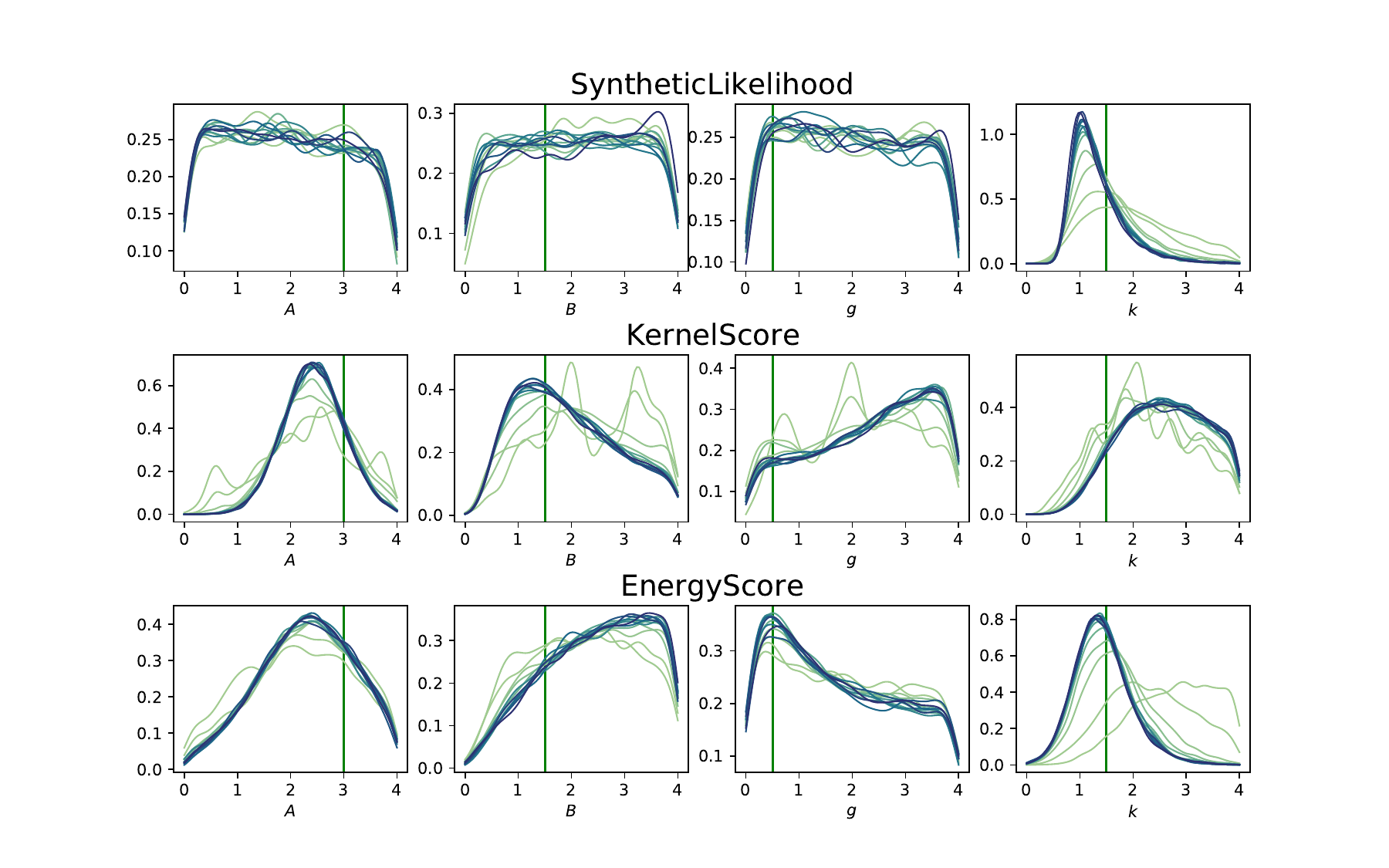}
	\caption{Univariate posterior marginals for different $ m $ values for the well specified univariate g-and-k distribution, for the BSL, Kernel and energy score posteriors, with PM-MCMC. Lighter (respectively darker) colors denote smaller (resp. larger) values of $ m $. For small values of $ m $, the marginals are spiky, which is due to unstable PM-MCMC. The densities are obtained by KDE on the MCMC output thinned by a factor 10.}
	\label{fig:uni_gk_m}
\end{figure}

\FloatBarrier
\subsection{Misspecified univariate g-and-k}
\FloatBarrier
Here, we report results considering $ n=10 $ observations.

\begin{table}[h!]
	\centering
		\caption{Acceptance rate and trace of the posterior covariance matrix for different values of $ m $ for the misspecified univariate g-and-k, for the BSL, Kernel and energy score posteriors. }
	\begin{adjustbox}{max width=\textwidth}
		\begin{tabular}{c|cccccc}
			\toprule
			\multirow{2}{*}{\textbf{$ m $}} &  \multicolumn{2}{c}{\textbf{BSL}} & \multicolumn{2}{c}{\textbf{Kernel score}} & \multicolumn{2}{c}{\textbf{Energy score}} \\
			\cmidrule(r){2-3} \cmidrule(r){4-5} \cmidrule(r){6-7}
			&   Acc. rate & $ \text{Tr}\left[\Sigma_{\text{post}}\right] $  &       Acc. rate &  $ \text{Tr}\left[\Sigma_{\text{post}}\right] $ &        Acc. rate &  $ \text{Tr}\left[\Sigma_{\text{post}}\right] $ \\
			\midrule
			10 &      0.038 &                3.3664 &      0.047 &                3.4141 &      0.164 &                3.8095 \\
			20 &      0.072 &                2.3207 &      0.069 &                3.2060  &      0.216 &                3.4900   \\
			50 &      0.130  &                1.9729 &      0.184 &                2.6690  &      0.306 &                2.9483 \\
			100 &      0.159 &                2.0145 &      0.298 &                2.4529 &      0.364 &                2.7232 \\
			200 &      0.179 &                1.8829 &      0.359 &                2.4037 &      0.391 &                2.7153 \\
			300 &      0.187 &                2.0198 &      0.389 &                2.3623 &      0.402 &                2.6055 \\
			400 &      0.188 &                1.9498 &      0.405 &                2.3403 &      0.410  &                2.6164 \\
			500 &      0.189 &                1.9092 &      0.412 &                2.3756 &      0.413 &                2.5579 \\
			600 &      0.191 &                1.8259 &      0.422 &                2.3461 &      0.414 &                2.5704 \\
			700 &      0.186 &                1.9207 &      0.430  &                2.3452 &      0.417 &                2.5484 \\
			800 &      0.184 &                1.9509 &      0.432 &                2.3810  &      0.419 &                2.6276 \\
			900 &      0.190  &                1.9475 &      0.434 &                2.4472 &      0.423 &                2.6468 \\
			1000 &      0.194 &                1.9763 &      0.436 &                2.3434 &      0.425 &                2.6386 \\
			\bottomrule
		\end{tabular}
	\end{adjustbox}
	\label{Tab:uni_Cauchy_gk_m}
\end{table}

\begin{figure}[tb]
	\centering
	\includegraphics[width=1\linewidth]{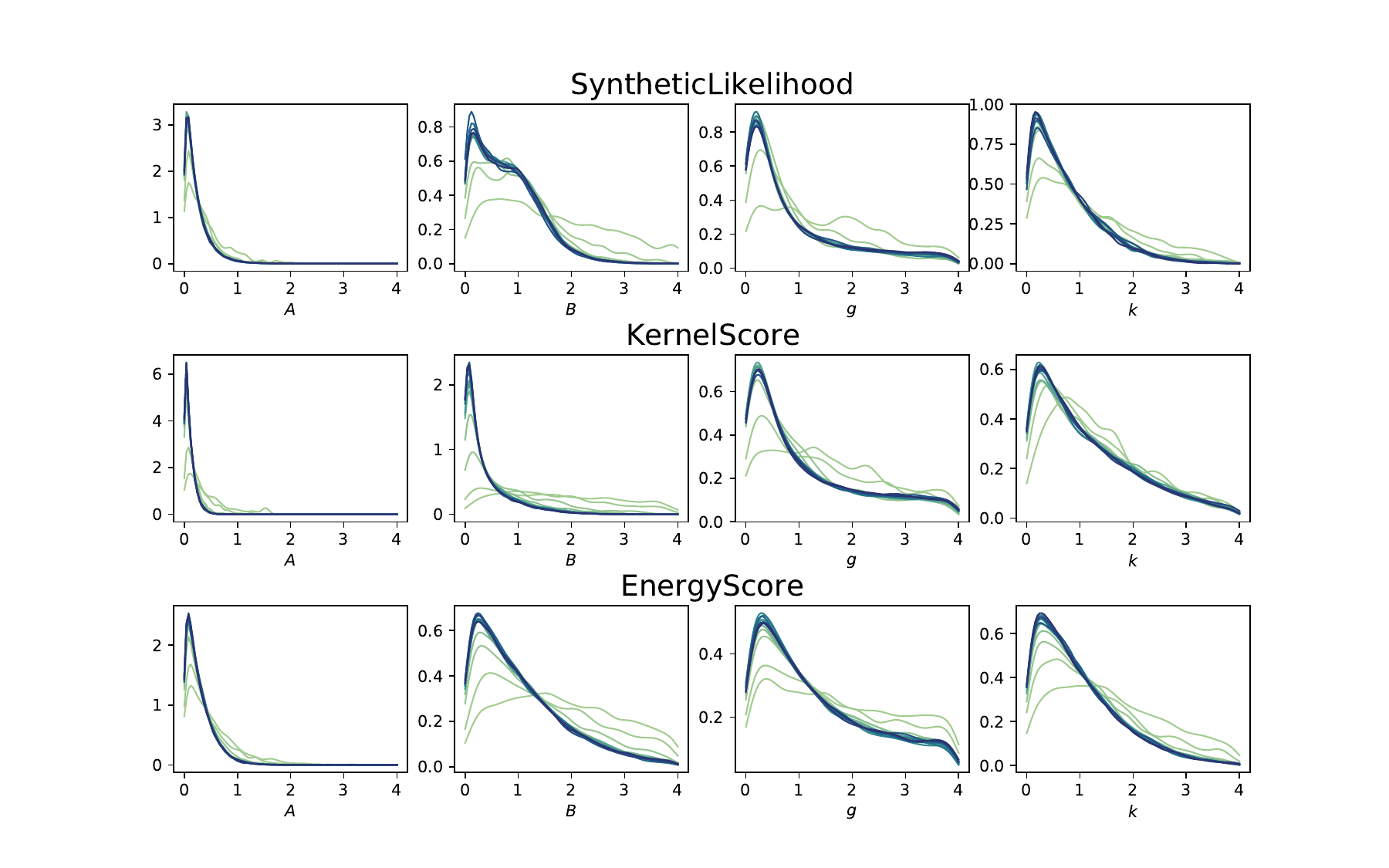}
	\caption{Univariate posterior marginals for different $ m $ values for the misspecified univariate g-and-k distribution, for the BSL, Kernel and energy score posteriors, with PM-MCMC. Lighter (respectively darker) colors denote smaller (resp. larger) values of $ m $. The densities are obtained by KDE on the MCMC output thinned by a factor 10.}
	\label{fig:uni_Cauchy_gk_m}
\end{figure}

\FloatBarrier
\subsection{Multivariate g-and-k}
\FloatBarrier
Here, we report results considering $ n=10 $ observations.

For this model, small $ m $ lead to extremely small acceptance rates for BSL and semiBSL (Table~\ref{Tab:gk_m}); in those cases, the trace of the posterior covariance matrix is also very small due to the chain being almost still. Additionally, even large $ m $ values lead to small acceptance rate for semiBSL; that is consequence of the issues discussed in Appendix~\ref{app:gk_nsim}. We report nevertheless the results here.

\begin{table}[h!]
	\centering
		\caption{Acceptance rate and trace of the posterior covariance matrix for different values of $ m $ for the well specified multivariate g-and-k, for the BSL, semiBSL, Kernel and energy score posteriors. }
	\begin{adjustbox}{max width=\textwidth}
		\begin{tabular}{c|cccccccc}
			\toprule
			\multirow{2}{*}{\textbf{$ m $}} &  \multicolumn{2}{c}{\textbf{BSL}} &\multicolumn{2}{c}{\textbf{semiBSL}} & \multicolumn{2}{c}{\textbf{Kernel score}} & \multicolumn{2}{c}{\textbf{Energy score}} \\
			\cmidrule(r){2-3} \cmidrule(r){4-5} \cmidrule(r){6-7} \cmidrule(r){8-9}
			&   Acc. rate & $ \text{Tr}\left[\Sigma_{\text{post}}\right] $  &   Acc. rate &  $ \text{Tr}\left[\Sigma_{\text{post}}\right] $ &       Acc. rate &  $ \text{Tr}\left[\Sigma_{\text{post}}\right] $ &      Acc. rate &  $ \text{Tr}\left[\Sigma_{\text{post}}\right] $  \\
			\midrule
			10 &      $<$0.001     &                1.0566 &      $<$0.001     &                0.4227 &      0.006 &                3.6061 &      0.070  &                4.5255 \\
			20 &      $<$0.001     &                0.3674 &      $<$0.001     &                0.6383 &      0.023 &                4.0455 &      0.123 &                3.9212 \\
			50 &      0.003 &                2.8320  &      $<$0.001     &                0.6331 &      0.055 &                3.8924 &      0.170  &                3.8571 \\
			100 &      0.002 &                2.3666 &      $<$0.001     &                0.6131 &      0.078 &                4.1250  &      0.194 &                3.8126 \\
			200 &      0.001 &                0.7140  &      0.001 &                0.8603 &      0.099 &                3.9624 &      0.206 &                3.7142 \\
			300 &      0.008 &                2.8229 &      0.002 &                2.2184 &      0.108 &                4.2766 &      0.208 &                3.9078 \\
			400 &      0.009 &                2.5694 &      0.001 &                0.6885 &      0.113 &                3.9710  &      0.212 &                3.8284 \\
			500 &      0.009 &                3.3583 &      0.002 &                1.2885 &      0.116 &                4.0250  &      0.217 &                3.8383 \\
			600 &      0.013 &                2.9646 &      0.005 &                1.3359 &      0.120  &                3.9632 &      0.216 &                3.7698 \\
			700 &      0.010  &                3.7043 &      0.005 &                0.6511 &      0.119 &                4.0173 &      0.214 &                3.7437 \\
			800 &      0.016 &                3.3017 &      0.006 &                0.6679 &      0.122 &                3.9607 &      0.214 &                3.7512 \\
			900 &      0.022 &                2.9915 &      0.005 &                0.6411 &      0.126 &                4.1293 &      0.216 &                3.9202 \\
			1000 &      0.017 &                3.1304 &      0.006 &                0.5892 &      0.122 &                3.9757 &      0.216 &                3.7959 \\
			\bottomrule
		\end{tabular}
	\end{adjustbox}
	\label{Tab:gk_m}
\end{table}

\begin{figure}[tb]
	\centering
	\includegraphics[width=1\linewidth]{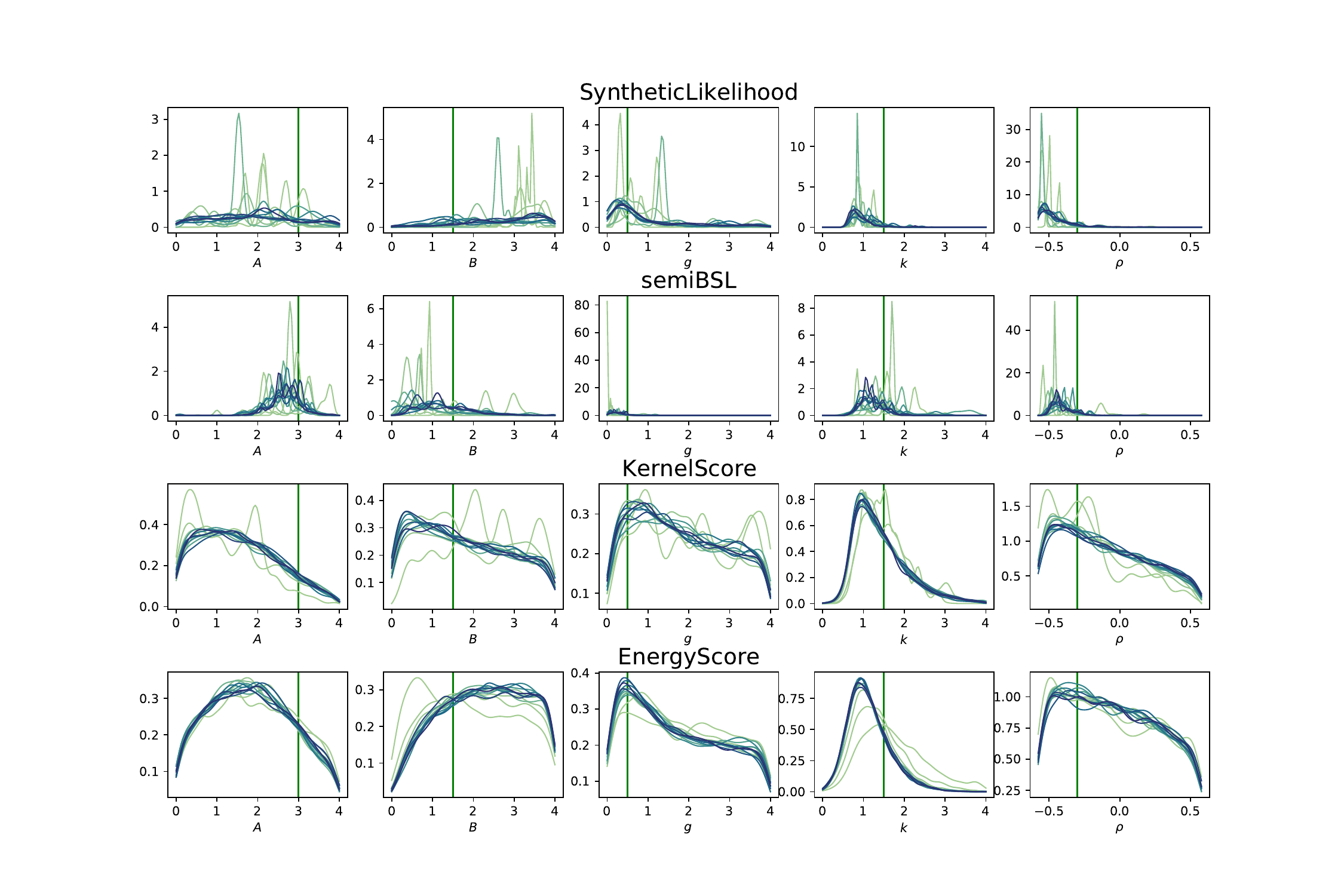}
	\caption{Univariate posterior marginals for different $ m $ values for the well specified multivariate g-and-k distribution, for the BSL, semiBSL, Kernel and energy score posteriors, with PM-MCMC. Lighter (respectively darker) colors denote smaller (resp. larger) values of $ m $. For small values of $ m $, the marginals are spiky, which is due to unstable MCMC. The densities are obtained by KDE on the MCMC output thinned by a factor 10.}
	\label{fig:gk_m}
\end{figure}

\FloatBarrier
\subsection{Misspecified multivariate g-and-k}
\FloatBarrier
Here, we report results considering $ n=10 $ observations. We do not report results for BSL and semiBSL as those were unable to run satisfactorily for that number of observations, for all considered values of $ m $.

\begin{table}[h!]
	\centering
		\caption{Acceptance rate and trace of the posterior covariance matrix for different values of $ m $ for the misspecified multivariate g-and-k, for the Kernel and energy score posteriors. }
	\begin{adjustbox}{max width=\textwidth}
		\begin{tabular}{c|cccc}
			\toprule
			\multirow{2}{*}{\textbf{$ m $}} &  \multicolumn{2}{c}{\textbf{Kernel score}} & \multicolumn{2}{c}{\textbf{Energy score}} \\
			\cmidrule(r){2-3} \cmidrule(r){4-5}
			&   Acc. rate & $ \text{Tr}\left[\Sigma_{\text{post}}\right] $  &       Acc. rate &  $ \text{Tr}\left[\Sigma_{\text{post}}\right] $\\
			\midrule
			10 &      0.017 &                4.5045 &      0.174 &                3.4306 \\
			20 &      0.108 &                3.6950  &      0.252 &                3.2373 \\
			50 &      0.243 &                3.4612 &      0.300   &                3.0291 \\
			100 &      0.308 &                3.4759 &      0.316 &                3.0081 \\
			200 &      0.344 &                3.4666 &      0.323 &                2.9303 \\
			300 &      0.348 &                3.4583 &      0.321 &                2.9160  \\
			400 &      0.355 &                3.4158 &      0.331 &                3.0031 \\
			500 &      0.359 &                3.4047 &      0.332 &                2.9743 \\
			600 &      0.363 &                3.3847 &      0.330  &                2.9321 \\
			700 &      0.360  &                3.3485 &      0.329 &                2.9249 \\
			800 &      0.361 &                3.3505 &      0.332 &                2.9854 \\
			900 &      0.363 &                3.3627 &      0.331 &                3.0155 \\
			1000 &      0.363 &                3.3307 &      0.330  &                2.9277 \\
			\bottomrule
		\end{tabular}
	\end{adjustbox}
	\label{Tab:Cauchy_gk_m}
\end{table}

\begin{figure}[tb!]
	\centering
	\includegraphics[width=1\linewidth]{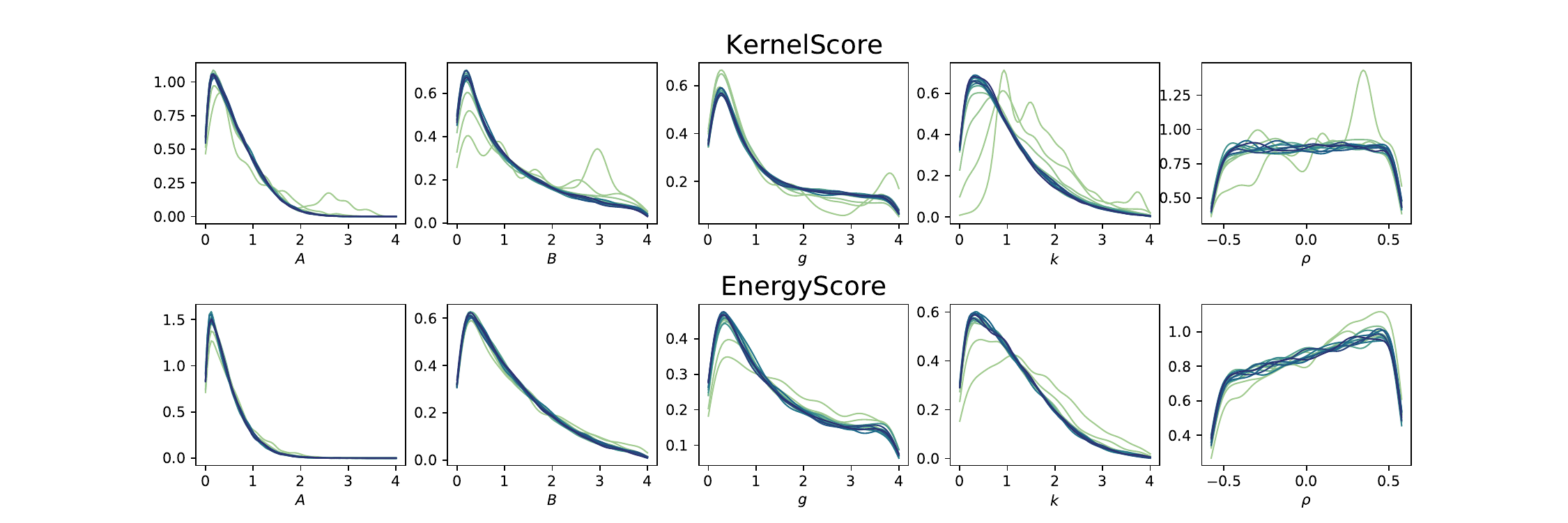}
	\caption{Univariate posterior marginals for different $ m $ values for the misspecified multivariate g-and-k distribution, for the Kernel and energy score posteriors, with PM-MCMC. Lighter (respectively darker) colors denote smaller (resp. larger) values of $ m $. For small values of $ m $, the marginals are spiky, which is due to unstable MCMC. The densities are obtained by KDE on the MCMC output thinned by a factor 10.}
	\label{fig:Cauchy_gk_m}
\end{figure}

\end{document}